\documentclass[onecolumn, superscriptaddress, 
longbibliography, notitlepage, showkeys, 
10pt]{revtex4-2}

\usepackage{graphicx}
\usepackage{environ}
\usepackage{appendix}
\usepackage{rotating,amsmath,amssymb}
\usepackage{array}
\usepackage{nicematrix}
\allowdisplaybreaks
\usepackage{subcaption}
\usepackage[margin=1in]{geometry}
\usepackage{booktabs, multirow}
\usepackage{csquotes}
\usepackage{enumitem}
\setlist[enumerate,1]{label=(\arabic*),ref=\arabic*}
\usepackage{qcircuit}
\usepackage{amsthm}
\usepackage{amssymb}
\usepackage{amsfonts}
\usepackage[colorlinks=true,
	urlcolor=blue!70!black!80!,
	citecolor=blue!70!black!80!,
	linkcolor=red!70!black!80!
	]{hyperref}
\usepackage{todonotes}
\usepackage{nicefrac}
\usepackage{mathrsfs}
\usepackage{newtxtext}
\usepackage{newtxmath}

\newcommand{\ket}[1]{|#1\rangle}
\newcommand{\bra}[1]{\langle #1 |}
\newcommand{\ketbra}[1]{\ket{#1}\bra{#1}}
\newcommand{\tr}{{\normalfont\text{tr}}}
\newcommand{\abs}[1]{\big\lvert #1\big\rvert}
\newcommand{\norm}[1]{\big\Vert #1\big\Vert}
\newcommand{\comm}[2]{\big[#1,#2\big]}
\newcommand{\order}[1]{\mathcal{O}\big(#1\big)}
\newcommand{\ordertilde}[1]{\widetilde{\mathcal{O}}\big(#1\big)}

\newcommand{\secref}[1]{Section~\ref{#1}}
\newcommand{\figref}[1]{Figure~\ref{#1}}
\newcommand{\thmref}[1]{Theorem~\ref{#1}}

\newcommand{\propref}[1]{Proposition~\ref{#1}}

\newtheorem{theorem}{Theorem}
\newtheorem{lemma}[theorem]{Lemma}
\newtheorem{corollary}[theorem]{Corollary}
\newtheorem{proposition}[theorem]{Proposition}
\newtheorem{assumption}[theorem]{Assumption}

\theoremstyle{remark}

\newtheorem{example}[theorem]{Example}

\theoremstyle{definition}
\newtheorem{definition}[theorem]{Definition}

\newcommand{\Real}{\mathbb{R}}

\newcommand{\vect}[1]{\boldsymbol{#1}}
\newcommand{\myeq}[1]{\stackrel{#1}{=}}
\newcommand{\myle}[1]{\stackrel{#1}{\le}}

\newcommand{\obs}{O}
\newcommand{\unit}{\boldsymbol{1}}
\newcommand{\hbt}{\mathcal{H}}

\newcommand{\loss}{\mathcal{L}}
\newcommand{\dloss}{\frac{\delta\loss}{\delta\rho}}

\newcommand{\ee}{\mathbb{E}}

\newcommand{\Hhydro}{H_{\text{hydro}}}
\newcommand{\numtime}{{N_t}}
\newcommand{\numop}{K}
\newcommand{\bigO}{\mathcal{O}}

\newcommand{\bregion}{\mathcal{B}}

\newcommand{\sample}{M_s}
\newcommand{\controlswap}{\mathcal{C}_{\text{swap}}}

\newcommand{\PL}{Polyak–Łojasiewicz}
\newcommand{\Lipschitzmax}{\mathsf{L}_{{\normalfont \text{Lip}}}}
\newcommand{\Niter}{N_{\text{iter}}}

\newcommand{\muG}{\mu_{\mathcal{G}}}
\newcommand{\adj}{{\normalfont \text{adj}}}
\newcommand{\aux}{{\normalfont \text{aux}}}

\begin{document}

\title{Quantum Neural Ordinary and Partial Differential Equations}
\author{Yu Cao}
\email{yucao@sjtu.edu.cn}
\affiliation{Institute of Natural Sciences, Shanghai Jiao Tong University, Shanghai 200240, China}
\affiliation{Ministry of Education Key Laboratory in Scientific and Engineering Computing, Shanghai Jiao Tong University, Shanghai 200240, China}
\affiliation{School of Mathematical Sciences, Shanghai Jiao Tong University, Shanghai, 200240, China}
\author{Shi Jin}
\affiliation{Institute of Natural Sciences, Shanghai Jiao Tong University, Shanghai 200240, China}
\affiliation{Ministry of Education Key Laboratory in Scientific and Engineering Computing, Shanghai Jiao Tong University, Shanghai 200240, China}
\affiliation{School of Mathematical Sciences, Shanghai Jiao Tong University, Shanghai, 200240, China}
\author{Nana Liu}
\email{nana.liu@quantumlah.org}
\affiliation{Institute of Natural Sciences, Shanghai Jiao Tong University, Shanghai 200240, China}
\affiliation{Ministry of Education Key Laboratory in Scientific and Engineering Computing, Shanghai Jiao Tong University, Shanghai 200240, China}
\affiliation{School of Mathematical Sciences, Shanghai Jiao Tong University, Shanghai, 200240, China}
\affiliation{Global College, Shanghai Jiao Tong University, Shanghai 200240, China}

\date{\today}

\begin{abstract}
We introduce a unified framework --- Quantum Neural Ordinary and Partial Differential Equations (QNODEs and QNPDEs) --- which extends the continuous-time formalism of classical neural ordinary and partial differential equations into quantum machine learning and quantum control. QNODEs denote the evolution of finite-dimensional quantum systems, whereas QNPDEs denote their infinite-dimensional (continuous-variable) counterparts; both are governed by generalised Schrödinger-type Hamiltonian dynamics, coupled with a corresponding loss function. This formalism permits gradient estimation via an adjoint-state method, facilitating efficient learning of quantum dynamics, and other dynamics that can be mapped (relatively easily) to quantum dynamics. Using this method, we present quantum algorithms for computing gradients with and without time discretisation, achieving efficient gradient computation that would otherwise be intractable on classical devices. We provide detailed resource estimates for these algorithms and investigate the local energy landscape for training.
The formalism subsumes a wide array of applications, including quantum state preparation, Hamiltonian learning, learning dynamics in open systems, and the learning of both autonomous and non-autonomous classical ODEs and PDEs. In many cases of interest, the Hamiltonian is composed of a relatively small number of local operators, yet the corresponding classical simulation remains inefficient, making quantum approaches advantageous for gradient estimation. This continuous-time perspective can also serve as a blueprint for designing novel quantum neural network architectures, generalising discrete-layered models into continuous-depth models.
\end{abstract}

\maketitle

\tableofcontents

\section{Introduction}
The interplay between the study of classical dynamical systems and machine learning has been a fruitful one and holds much promise. The theory of dynamical systems provides a mathematical framework to study the evolution of systems over time, offering insights into stability, chaos, and long-term behavior. Machine learning, particularly in areas like recurrent neural networks and neural ODEs, draws on this theory to model temporal data and learn complex dynamics. Conversely, machine learning facilitates the  discovery of governing equations and predicting behaviors in high-dimensional or partially observed dynamical systems. This interplay enriches both fields, enabling data-driven discovery of physical laws and more robust and interpretable learning models. \\

A key step to generalising machine learning algorithms for the study of dynamical systems is to generalise discrete-time algorithms into their continuous-time counterparts. In this direction, one notable milestone is the introduction of neural ordinary differential equations (neural ODEs) \cite{2018NODE, e_proposal_2017}, which has since attracted significant attention. This approach was motivated by the deep connection between ODEs and Residual Networks (ResNet) \cite{lu_beyond_2018, haber_stable_2017, ruthotto_deep_2020}. In parallel, the extension to infinite-dimensional systems is natural, and we refer to these as neural PDEs, although this concept is also known by other names, such as PDE-Net \cite{long_pde-net_2018} and physics-informed neural networks (PINN) \cite{raissi_physics-informed_2019}, depending on the specific architecture and training details. Further generalisations have been made to stochastic dynamics, see e.g., \cite{liu_neural_2019, tzen_neural_2019, li_scalable_2020, kidger_neural_2021}. Another related framework is the SDE-based diffusion model \cite{song_score-based_2021}, which unifies earlier discrete-time Markov chain-based generative models \cite{song_generative_2019, ho_denoising_2020}, and has inspired research in generative modeling. 
This continuous-time perspective of neural ODEs offers several key advantages: (1) efficient computation enabled by the use of the adjoint-state method --- the continuous-time analogue of backpropagation algorithm \cite{rumelhart_learning_1986} --- together with classical numerical integrators for dynamical systems; (2) a principled analytical framework for studying properties such as stability and informing architecture design \cite{e_proposal_2017,haber_stable_2017}; and (3) a unified paradigm for designing neural network architectures \cite{lu_beyond_2018, haber_stable_2017, ruthotto_deep_2020}, providing an alternative to mapping-based architectures.\\

Nevertheless, some fundamental bottlenecks remain. In particular, the gradient estimation required to train these models via gradient-based methods can be computationally expensive, because classical simulation of high-dimensional dynamical systems often suffers from the curse of dimensionality, particularly for many-body quantum dynamics. To mitigate this, quantum simulation for quantum dynamical systems can, in principle, avoid this curse of dimensionality \cite{feynman_simulating_1982}; in some scenarios, analogous benefits also extend to non-quantum dynamical systems. \\

This motivates us to consider machine learning on quantum devices from the perspective of quantum dynamical systems (an analog perspective on quantum machine learning), which is much less studied compared to its digital counterparts \cite{li_hybrid_2017,tangpanitanon_expressibility_2020,markovic_quantum_2020,magann_pulses_2021,liang_variational_2022,leng2022differentiable,oscar_rodrigo_analog_2024}. For all quantum algorithms that rely on fundamental quantum mechanics, the most basic dynamical equation is Schr\"odinger's equation. Hence, we introduce the notion of (closed) quantum neural ODEs and PDEs (QNODEs and QNPDEs) where the state $\rho(t, \theta)$ evolves according to 
\begin{align}
\label{eqn::qnde}
    \frac{\partial \rho(t, \theta)}{\partial t}=-i \big[ H(t, \theta), \rho(t, \theta) \big], \qquad \rho(0, \theta)=\rho_0,
\end{align}
and $H(t,\theta)$ is a time-dependent Hamiltonian with controllable parameters $\theta$, and there is an associated loss function $\mathcal{L}$, dependent on the machine learning task. When $\rho(t, \theta)$ is a finite-dimensional quantum system, we refers to this as a (closed) QNODE and if $\rho(t, \theta)$ is an infinite-dimensional quantum system, it is called a (closed) QNPDE. If $\rho(t, \theta)$, instead, obeys open quantum dynamics, then this is an (open) QNODE or QNPDE. \\

This formalism is natural for quantum machine learning:
\begin{itemize}
    \item {\bf (Natural ansatz):} Firstly, this is a very natural formalism for the learning of quantum dynamics itself (namely, Hamiltonian learning \cite{wiebe_2014_hamiltonian}), and lends itself easily to the more traditional problems in quantum control and quantum state preparation. 
    
    \item {\bf (Wide applications):} Secondly, the learning of unknown parameters of other dynamical systems can be made more efficient by the mapping of these dynamical systems to the corresponding quantum dynamics. In particular, we know that for linear and nonlinear ODEs, linear PDEs and certain nonlinear PDEs, there is a very simple mapping possible through Schr\"odingerisation \cite{analogPDE, PRLschr}, level set \cite{JL-nonlinear} or von Neumann-Koopman formulation \cite{joseph2020koopman} that exploits the fully continuous nature of these problems, in both time and space. In this case, the Hamiltonian representation often can be decomposed into a relatively smaller number of operators (since the form of the problem is known) and this kind of ansatz is much more physically-motivated than a generic neural network ansatz. 
    
    \item {\bf (Time is naturally continuous):} Thirdly, in the design of quantum algorithms with the common gate-based methods, those quantum circuits still require the preparation of pulse-level control of the quantum system, which, in fact, runs in continuous time. The so-called discrete-time settings in reality cannot be truly discrete since quantum speed limits \cite{del_campo_quantum_2013} fundamentally forbid the instantaneous application of two distinguishable quantum channels one after the other. 
    
    \item {\bf (Connection to QNN):} Another motivation is to study the connection to other quantum setups like in quantum neural networks (QNN), which run in discrete time steps. Here we have a continuous-depth version which is non-autonomous (i.e., explicit time dependence in the Hamiltonian) and, in principle, the different discretisations can reproduce all other (discrete) architectures for quantum neural networks \cite{Nonauto2}. Different discretisations can give rise to different architectures. In addition, backpropagation can be performed more directly (a continuous-version of backpropagation). Also, since Hamiltonian dynamics in Eq.~\eqref{eqn::qnde} is reversible, it is known that for all reversible maps, the loss function gradients can be computed via the adjoint sensitivity method with constant memory cost independent of depth \cite{gomez2017reversible,2018NODE}. This decoupling of depth and memory can be important in applications \cite{gomez2017reversible}.
\end{itemize}

In our work, we show the feasibility of generalising continuous-time classical backpropagation into the quantum regime and present two algorithms for gradient estimation: one based on partial time discretisation and one that does not require any time discretisation at all. 
These two algorithms are summarised in Table~\ref{table::1}, which utilised the back-propagation structure of classical neural ODEs. These algorithms differ from the parameter-shift rule
\cite{li_hybrid_2017,mitarai_quantum_2018,Schuld_2019_evaluating,kottmann_evaluating_2023,wierichs_general_2022} for digital quantum devices, which requires quantum measurements using a collection of different parameter values per iteration of the gradient estimate;  
In our case, the quantum measurements are performed using the same set of parameters at each stage, and almost no restrictions are imposed on the Hamiltonian, except the following general decomposition:
 \begin{align}
 \label{eqn::H_ansatz}
 H(t, \theta)=\sum_{k=1}^\numop f_k(t, \theta)H_k, \qquad f_k(t, \theta) \in \mathbb{R},
 \end{align}
 where  $H_k \neq H_{k'}$ for all $k \neq k'$ are time-independent Hermitian operators, then the cost for gradient estimation depends only on $K$, and thus not a priori dependent on the number of tuneable parameters $M$ for $\theta$; see more discussion in \secref{sec::cost}.\\

 After a very brief background section, we introduce the QNODE and QNPDE formalism in Section~\ref{sec:scopeapps}. In particular, we present the applications of this formalism to learning quantum state preparation and Hamiltonian learning (Section~\ref{sec:appclosed}), certain learning problems in open quantum systems (Section~\ref{sec:appopen}), learning (classical) ODEs/PDEs (Section~\ref{sec:appodepde}) and learning non-autonomous systems (Section~\ref{sec:appnonauto}). Furthermore, in Section~\ref{sec:appqnn} we show how this provides a general framework for a large class of quantum neural networks in the continuous-time limit, from which new quantum neural network architectures could emerge from choosing different discretisations.
 In Section~\ref{sec:mixed}, we present our quantum algorithms for gradient estimation, using partial time discretisation (Section~\ref{section::quantum_alg}) and without any time discretisation (Section~\ref{alg::no_discretization}), then discuss the costs in Section~\ref{sec:complexity}, as well as local energy landscape in Section~\ref{subsec::landscape}. Some simple numerical examples are demonstrated in Section~\ref{sec::example}. 
 
\begin{table}[h!]
\caption{Summary of the main quantum gradient estimation algorithms for QNODEs/QNPDEs, see Section~\ref{sec:mixed} for more details.}
\vspace{\baselineskip}
\begin{NiceTabular}[width=17cm]{X[3,l,m]X[2,l,m]X[2,l,m]X[3,l,m]X[6,l,m]}[]
\hline
{\bf Algorithm} &  {\bf Theorem} & {\bf Circuit} & {\bf Treatment of time} & \medskip \hspace{2em} {\bf Applications} \medskip \\
\hline
\medskip {\bf QNODE/QNPDE gradient estimation algorithm (v1)} &   \thmref{thm::main} & \figref{fig::alg1::parta} & Requires partial time discretisation & \multirow{2}{\textwidth}{\vspace{0.07\baselineskip}
\begin{minipage}{\textwidth}\raggedright 
\begin{itemize}
\item State preparation and Hamiltonian learning (closed system (\ref{sec:appclosed}); open system (\ref{sec:appopen}));
\item Learning ODEs/PDEs (\ref{sec:appodepde}); 
\item Learning non-autonomous systems (\ref{sec:appnonauto});
\item Continuous-time framework for QNNs (\ref{sec:appqnn}). 
\end{itemize}
\end{minipage}  }\vspace{5\baselineskip}  \\
\bigskip {\bf QNODE/QNPDE gradient estimation algorithm (v2)} \bigskip & \thmref{thm5} & \figref{fig::alg1::partb} & Fully continuous in time & \\
\hline
\end{NiceTabular}
\label{table::1}
\end{table}

\section{Background}
In this section, we briefly review the notation used in quantum information theory, and revisit classical neural ODEs/PDEs in classical machine learning.

\subsection{Basic notation in quantum information processing}

We first summarize the notation for discrete-variable (DV) and continuous-variable (CV) quantum states used in later sections. Quantum information processing is often described in terms of qubits, the quantum analogue of classical bit strings $\{0,1\}^{\otimes n}$. A qubit uses the eigenbasis $\{|0\rangle, |1\rangle\}$ for the computational basis and resides in a two-dimensional Hilbert space $\hbt_{\text{qubit}}$. An $n$-qubit system lives in a $D=2^n$-dimensional Hilbert space, given by the tensor product $\big(\mathcal{H}_{\text{qubit}}\big)^{\otimes n}$. The corresponding quantum state can be written  as $|u\rangle=\frac{1}{\|\mathbf{u}\|} \sum_{j=1}^D u_j|j\rangle$, where the normalisation is $\|\mathbf{u}\|^2=\sum_{j=1}^D |u_j|^2$. This notation can also be used to denote a $D$-dimensional {\it qudit} system. Mathematically this is equivalent to a system of $n=\log_2 D$ qubits, but in physical implementation, these are distinct. Both qubits and qudits are discrete-variable (DV) quantum systems. \\

A continuous-variable (CV) quantum state, or {\it qumode}, spans an infinite-dimensional Hilbert space $L^2(\Real)$ \cite{lloyd_quantum_1999}. A qumode is the quantum analogue of a continuous classical degree of freedom, such as position or momentum. It is equipped with observables with continuous spectra, such as the position $\hat{x}$ and momentum $\hat{p}$ operators. The eigenbasis $\{|x\rangle\}_{x \in \mathbb{R}}$ consists of eigenstates of $\hat{x}$, and a qumode can be expressed as $|u\rangle = \frac{1}{\|\mathbf{u}\|} \int u(x)|x\rangle dx$, with normalisation $\|\mathbf{u}\|^2=\int |u(x)|^2 dx$. Integrals $\int$ without terminals denote $\int=\int_{-\infty}^{\infty}$ unless otherwise specified. If $x=(x_1, \cdots, x_m)$ then the above represents a system of $m$ qumodes. Although CV quantum information processing is less common than qubit-based approaches, these states are in fact natural for many quantum systems, and is particularly important for the simulation of linear PDEs while preserving continuity in space \cite{analogPDE}. The quadrature operators of a qumode are $\hat{x}$ and $\hat{p}$, where $[\hat{x},\hat{p}]=i \unit$. If we let $|x\rangle$ and $|p\rangle$ denote the eigenvectors of $\hat{x}$ and $\hat{p}$ respectively, then $\langle x|p\rangle=\exp(ixp)/\sqrt{2\pi}$. The position and momentum eigenstates each form a complete eigenbasis so $\int  \ketbra{x} dx =\unit =\int  \ketbra{p} dp$. Here the quantised momentum operator $\hat{p}$ is also associated with the spatial derivative through $i\hat{p}|u\rangle \propto \int \partial u/\partial x|x\rangle dx$. Similarly for the position quadrature, we have the association $\hat{x}|u\rangle \propto \int x u(x)|x\rangle dx$. This is important for associating Hamiltonians acting on qumodes with PDEs \cite{analogPDE}. \\

We denote the Hilbert space for the Hamiltonian dynamics as $\hbt$ (which may be finite- or infinite-dimensional), the adjoint system as $\hbt_{\adj}$ (with the same dimension as $\hbt$), and the auxiliary qubit system as $\hbt_{\aux}$; see \secref{sec:mixed}. We use $\norm{\cdot}_p$ to denote the Schatten-$p$ norm of operators; in particular, $\norm{\cdot}_{\infty}$ is the usual operator norm (namely, largest singular value of the operator). 
For two operators $A$ and $B$, $A\succeq B$ or $B\preceq A$ means that $A - B$ is a positive semidefinite operator.
The notation $\mathcal{O}(\cdot)$ is the big-O notation and $\widetilde{\mathcal{O}}(\cdot)$ refers to the big-O notation while suppressing terms with logarithmic scaling; we may alternatively use notation $f\lesssim g$ to represent $f = \order{g}$ and $f\gtrsim g$ to represent $g = \order{f}$ for convenience in presenting detailed proofs.

\subsection{Classical neural ODE and PDE}

The concept of neural ordinary differential equations (ODEs) was formally introduced in \cite{2018NODE,e_proposal_2017}, referring to a general parametric family of ODEs of the form:
\begin{align}
\label{eqn::classical_ode}
\frac{d}{dt} y(t) = f\big(y(t), t, \theta\big),
\end{align}
where $f: \Real^n \times \Real \times \Omega \to \Real^n$ is a (potentially nonlinear) function parameterised by $\theta \in \Omega$, with $\Omega$ denoting the domain of parameters. As already discussed in introduction, this framework is motivated by the connection between ODEs and residual networks \cite{lu_beyond_2018, haber_stable_2017}. For example, applying the Euler discretization with time step $\Delta t$ yields
\begin{equation*}
y_{k+1} = y_k +  f(y_k, t, \theta)\Delta t,
\end{equation*}
which mirrors the structure of ResNet. The continuous-time perspective enables new approaches to neural network architecture design \cite{haber_stable_2017,lu_beyond_2018}, and facilitates efficient training via the adjoint equation, which generalises the classical backpropagation algorithm \cite{rumelhart_learning_1986} to the continuous time.
Schr{\"o}dinger's equation is also of the form in  Eq.~\eqref{eqn::classical_ode}, and one may utilise classical machine learning for quantum machine learning problems. However, there is a concern of the scalability in dimensionality of classical machine learning for high-dimensional quantum systems, which motivates the study in the follow-up sections. \\

The finite-dimensional state $y(t)$ in Eq.~\eqref{eqn::classical_ode} can be generalised to a function $u(t)$, with $f$ replaced by a parameterised differential operator $\mathscr{D}$, leading to what we refer to as {neural PDEs}:
\begin{align*}
\frac{d}{d t} u(t) = \mathscr{D}\big(u(t), t, \theta\big).
\end{align*}
Depending on the training scheme and loss function, this approach is known by various names in the literature, such as PDE-Net \cite{long_pde-net_2018} and physics-informed neural networks (PINN) \cite{raissi_physics-informed_2019}, among others.

\section{Quantum neural ODE and PDE: scope and applications}
\label{sec:scopeapps}

We will begin with introducing some general concepts about quantum machine learning, and then introduce the quantum neural ordinary and partial differential equations. Then the section is followed by many applications: Hamiltonian learning for closed and open quantum systems, learning ODEs and PDEs, and learning non-autonomous systems.
Finally, we will discuss how the framework of QNODE/QNPDE could offer insights into designing quantum neural networks.

\subsection{General concepts}

For \textit{any} learning algorithm whose task is to learn an approximation of unknown quantum states (parameterised as $\rho(\theta)$) or unknown quantum channels (parameterised as $\mathcal{E}(\theta)$), it must involve a functional $\mathcal{L}$:
\begin{align*}
    \mathcal{L}: \qquad \rho(\theta),\  \mathcal{E}(\theta) \rightarrow \mathbb{R}, \qquad \theta=(\theta_1, \cdots, \theta_M) \in \mathbb{R}^M.
\end{align*}
Then one aims to solve for $\theta$ that minimizes the loss function 
\begin{align} \label{eq:argL}
    \text{argmin}_{\theta}\, \mathcal{L}_{\theta},
\end{align}
where we use the notation $\mathcal{L}_{\theta}$ for simplicity, and it should be understood as $\loss_\theta = \loss\big(\rho(\theta)\big)$ or $\loss_\theta = \loss\big(\mathcal{E}(\theta)\big)$ depending on the context. Due to the close relationship between the quantum channel that creates the density matrix from a known state and the density matrix itself, without loss of generality, we will consider $\loss_\theta = \loss\big(\rho(\theta)\big)$ mostly below.\\

Since all quantum states or channels ultimately arise from physical processes, they must also evolve in real, continuous time $t \in \mathbb{R}$. Any discretisation in time that arises in some scenarios -- for instance in digital quantum simulation -- is some approximation of the continuous time setting. The so-called discrete-time settings in reality cannot be truly discrete since quantum speed limits \cite{del_campo_quantum_2013} fundamentally forbids the instantaneous application of two distinguishable quantum channels one after the other. Furthermore, we will see that there are many natural applications where continuity in time is an essential requirement. Thus, for these learning tasks, we must consider the dynamical laws -- in continuous time $t$ -- corresponding to the quantum state $\rho(t, \theta)$ and quantum channels $\mathcal{E}(t, \theta)$. The corresponding loss function $\mathcal{L}_{\theta}(t=T)$ can either result from $\rho(t=T, \theta)$ or $\mathcal{E}(t=T, \theta)$ at some terminal time $T$ or certain running cost e.g., $\loss_{\theta}=\int_0^T \loss(\rho(t,\theta)) dt$ or be a combination of both (for instance in general quantum optimal control problems). Since the running cost, after being discretised, could be viewed as a linear combination of terminal costs, we will simply consider the terminal cost throughout this work. \\

Given the form of $\mathcal{L}_{\theta}$ that is motivated by the learning problem at hand, the key next step is to approximately solve the optimization problem in Eq.~\eqref{eq:argL}. 
One may use non-gradient based optimisation methods which have also been used in \cite{tangpanitanon_expressibility_2020,liang_variational_2022}. For more efficiency in optimisation, we will focus on the gradient-based methods for Eq.~\eqref{eq:argL}.
Gradient-based methods require first and foremost the efficient computation of the gradients $\partial \mathcal{L}/\partial \theta_m$, $m=1, \cdots, M$. When the system dimension is large, in some scenarios it is possible to use quantum simulation to compute the gradient more efficiently. For digital quantum devices, a widely studied approach is the parameter-shift rule \cite{li_hybrid_2017,mitarai_quantum_2018,Schuld_2019_evaluating,kottmann_evaluating_2023,wierichs_general_2022}, which has recently also been used in \cite{leng2022differentiable} to design a forward-type algorithm. It remains an open question whether a backward-type algorithm is possible, which is addressed by our new quantum algorithms in Section~\ref{sec:mixed}; see full details therein.
We will see that these gradient computation algorithms require the ability to prepare -- as an initial state input -- either the functional derivative $\dloss\vert_{\rho=\rho(T)}$ being proportional to a density matrix or require $\dloss\vert_{\rho=\rho(T)}$ to be expressed as a known linear combination of density matrices which can each be prepared.\\

Note that $\mathcal{L}_{\theta}$ is a functional on the space of density operators, equipped with Hilbert-Schmidt inner product, the change of the loss function, denoted as $\dloss{}$, is defined as follows:  
\begin{align*}
\frac{d}{d\epsilon}\loss(\rho + \epsilon A)\Big\rvert_{\epsilon=0} = \tr\Big(\dloss{} A \Big), \qquad \text{for all feasible Hermitian operators } A.
\end{align*}
Note that $\dloss{}$ is a \emph{Hermitian operator} on the Hilbert space of the original quantum system $\hbt$. An important class of loss functions that can satisfy the above conditions includes
\begin{align} \label{eq:lossdef}
\begin{aligned}
   \mathcal{L}_{\theta}\ = &\ \tr\big(\rho(T, \theta)\obs\big), \qquad \obs=\sum_{j=1}^J c_j \obs_j=\obs^{\dagger}, \\
& \obs_j \succeq 0, \quad \tr(\obs_j)=1, \quad c_j \in \mathbb{R},
\end{aligned}
\end{align}
where $\rho(t, \theta)=\mathcal{E}(t, \theta)(\rho_0)$ can obey either closed quantum dynamics or open quantum dynamics, $c_j$ are known constants and we assume that we can be given access to $O_j$ as density matrices. For example, here it is clear that $\dloss \vert_{\rho=\rho(T)}=O$. Although we will examine applications in Sections~\ref{sec:appclosed} -- \ref{sec:appqnn} where the loss functions can mostly be expressed in the form in Eq.~\eqref{eq:lossdef}, the formalism itself is not solely confined to this loss function form. The quantum circuits for the gradient estimation algorithms in Section~\eqref{sec:mixed} can also be applied  to generic loss functions as long as $\dloss{}$ is proportional to or can be decomposed to density matrices.\\

\subsection{Quantum neural ODEs and PDEs}

To proceed with this fairly general setting, we first need to define the concept of quantum neural ODEs and PDEs, analogously to (classical) neural ODEs/PDEs, so we focus primarily the formulation where $\rho(t, \theta)$ evolves via dynamics in \textit{continuous time}, rather than discrete time.  

\begin{definition} \label{def:qnode}
    Let the unitary evolution of the initial quantum state with density matrix $\rho_0$ obey the dynamical equation
    \begin{align} \label{eq:schrdensity1}
    \frac{\partial \rho(t, \theta)}{\partial t}=-i \big[ H(t, \theta), \rho(t, \theta) \big], \qquad \rho(0, \theta)=\rho_0, \qquad H(t, \theta)=H^{\dagger}(t, \theta), \qquad \theta=(\theta_1, \cdots, \theta_M).
\end{align}
The parameters $\{\theta_m\}_{m=1}^M$ are tuned to optimise a given loss function $\mathcal{L}_{\theta}$. We refer to the system in Eq.~\eqref{eq:schrdensity1} with its corresponding $\mathcal{L}_{\theta}$ collectively as
\begin{enumerate}
\item a (closed) \textit{quantum neural ordinary differential equation (QNODE)}, if $\rho(t,\theta)$ are finite $D$-dimensional quantum states (DV); 
\item a (closed) \textit{quantum neural partial differential equation (QNPDE)} when $\rho(t,\theta)$ are a system of $d$ qumodes (CV); 
\item a (closed) \textit{hybrid QNPDE} when $\rho(t, \theta)$ is a hybrid CV-DV quantum state.
\end{enumerate}
If $\rho(t, \theta)=\mathcal{E}(t,\theta)(\rho_0)$ where $\mathcal{E}(t,\theta)$ is a quantum channel that instead describes open quantum system dynamics, then depending on if $\rho(t, \theta)$ is a DV or CV, or hybrid DV-CV quantum state, the corresponding ODE or PDE for $\rho(t, \theta)$ is referred to as an \textit{(open) QNODE}, \textit{(open) QNPDE} or \textit{(open) hybrid QNPDE}.  
\end{definition}
In this paper, we examine this formalism for wide range of applications of general closed QNODEs/QNPDEs for loss functions of the type in Eq.~\eqref{eq:lossdef}. In these cases, it is sufficient to learn the generating Hamiltonian $H(t, \theta)$ in Eq.~\eqref{eq:schrdensity1}, which has the general decomposition in Eq.~\eqref{eqn::H_ansatz} (also listed below):
\begin{align*}
H(t, \theta)=\sum_{k=1}^\numop f_k(t, \theta)H_k, \qquad f_k(t, \theta) \in \mathbb{R}, \qquad H_k=H_k^{\dagger},
\end{align*}
where $H_k$ are time-independent Hermitian operators.  If classical data $z \in \mathbb{R}^N$ is also included, then we can augment $H(t, \theta) \rightarrow H(t, \theta, z)=\sum_{k=1}^\numop f_k(t, \theta, z) H_k$, but we will suppress the $z$ parameters for simplicity unless explicitly required. \\

For finite $D$-dimensional quantum systems, if the Lie algebra generated by the set $\{iH_k\}_{k=1}^K$ under commutation and linear combination is dense in $SU(D)$, then $H(t,\theta)$ can generate an arbitrary unitary $U(t)=\mathcal{T}\exp(-i \int_0^t H(\tau, \theta) d\tau)$ with suitably chosen $f_k(t, \theta)$ and in general $K_{\max}=D^2-1$, although many terms $f_k(t,\theta)$ can also be zero. We can also have $K \ll D^2-1$ depending on applications that do not require universality. Below are a few examples of Eq.~\eqref{eqn::H_ansatz}. When $D=2$, $H_1=\sigma_x, H_2=\sigma_y, H_3=\sigma_z$. For continuous-variable quantum systems, for approximate universality in generating unitaries, it is sufficient for the set $\{H_k\}_{k=1}^K$ to contain only up to quadratic Hamiltonians (i.e., degree one or two in $\hat{x}$, $\hat{p}$) that generate only Gaussian operations, and only a single non-Gaussian Hamiltonian (e.g., $H_k=\hat{x}^3$) and access to arbitrary $f_k(t)$ \cite{lloyd_quantum_1999}. For example, for a single qumode, $H_1=\hat{x}, H_2=\hat{p}, H_3=\hat{x}^2, H_4=\hat{p}^2, H_5=\hat{x}\ \hat{p}, H_6=\hat{x}^3$ is one possible choice.\\

There are primarily two types of applications to consider and these demonstrate the two key advantages of this continuous-time formalism:  
\begin{itemize}

    \item {\bf (Special classes of data -- no universality required)}: It naturally captures the class of learning problems where the data arises directly from continuous-time dynamics associated with Eq.~\eqref{eq:schrdensity1} or with ODEs/PDEs that can be naturally mapped onto Eq.~\eqref{eq:schrdensity1}. This includes most scenarios that arise from physical phenomena obeying known linear dynamical laws and some obeying nonlinear dynamics. Here $H_k$ appears naturally and $K$ is not often too large; 
    
    \item {\bf (General data -- universality required)}: It can also capture the class of learning problems more appropriate for \enquote{messy dynamics}, where the form of the dynamical laws that naturally generates these data are not generally known nor can be written in simple form (e.g. dynamics generating pictures of cats versus those generating dogs). Then universality in representation is more appropriate, and the learning task is to approximate $f_k(t, \theta)$ using some basis expansion. Although it is sufficient to choose $K=K_{\max}$, but $K \ll K_{\max}$ is also possible by appropriate choices of $f_k(t, \theta)$, and these different choices correspond to different \enquote{learning architectures}.
\end{itemize}

Here the choices of $\{f_k(t, \theta)\}_{k=1}^K$ and the basis which we choose to approximate these functions is thus a central question. The form of $f_k(t, \theta)$ could arise from either (not deliberate) natural dynamics, or it could arise from (deliberate) quantum control. In quantum control $H(t, \theta)=H_0+\sum_{k=1}^K f_k(t, \theta)H_k$, 
where $H_0$ is the drift (free) Hamiltonian, $\sum_{k=1}^K f_k(t, \theta)H_k$ is known as the control Hamiltonian, and $f_k(t, \theta)$ are the time-dependent control functions. One key class of control fields is chirped fields, which dynamically modulates frequencies to achieve precise quantum state transitions. For example 
(1) $f_k(t, \theta)=\sum_{i} f^{(i)}_k(\theta)(\Theta(t-t_i)-\Theta(t-t_{i+n_i}))$ are piecewise constant, where the time interval of each non-trivial signal is $t_{i+n_i}-t_i$. The notation $\Theta$ denotes the Heaviside function on the real line.
This is known as bang-bang control and also corresponds to the gate-based circuit picture for quantum simulation; 
(2)$f_k(t, \theta)=\sum_n a^{(k)}_n(\theta) \cos(n \omega_k(\theta) t)+b^{(k)}_n(\theta) \sin (n \omega_k(\theta) t)$ corresponds to harmonic driving, and appears for example in NMR pulse shaping and Rabi oscillations; 
(3) $f_k(t, \theta)=\sum_n a^{(k)}_n (\theta)e^{-i\omega_{n,k}(\theta)t}$ for example used in GRAPE; 
(4) $f_k(t, \theta)=\sum_n c^{(k)}_n (\theta) t^n$ is used in higher-order optimal control; 
(5) $f_k(t, \theta)=\sum_n A^{(k)}_n(\theta) e^{- (t-T/2)^2/2\sigma^2_{n,k}(\theta)}$ corresponds to Gaussian or adiabatic control, and appears in Stimulated Raman adiabatic passage (STIRAP) in quantum state transfer and adiabatic quantum computing.\\

This formulation with simple loss functions in Eq.~\eqref{eq:lossdef} is already relevant to a broad range of applications, including:  
\begin{enumerate}
    \item learning an unknown quantum state preparation and  Hamiltonian learning (Section~\ref{sec:appclosed});
    \item certain learning problems in open quantum systems (Section~\ref{sec:appopen});
    \item learning classical ODEs/PDEs (Sections~\ref{sec:appodepde});
    \item learning non-autonomous systems (Section~\ref{sec:appnonauto}); 
    \item providing a general formulation for a large class of quantum neural networks (Section~\ref{sec:appqnn}) in the continuous-time limit, from which new quantum neural network architectures could emerge from choosing different discretisations.
    \end{enumerate}
We will now examine each application in turn below and provide simulations for toy examples in Section~\ref{sec::example}.

\subsection{Learning in closed quantum systems} \label{sec:appclosed}

\noindent \textit{Quantum state preparation --} The first application we can consider is quantum target state preparation from a known initial state $\rho_0$. Here the task is to approximate a given state $\sigma$, where we choose the loss function as follows:
\begin{align}
\label{eqn::loss_generative}
    \mathcal{L}_{\theta}=1-\tr\big(\sigma \rho(T,\theta)\big) \equiv \tr\big( (\unit - \sigma) \rho(T,\theta)\big),
\end{align}
where $\rho(T, \theta)=U(T, \theta)\rho_0 U(T, \theta)^{\dagger}$ and $\sigma=U\rho_0 U^{\dagger}$ and $U$ may be a black-box. This is a relevant loss function whether $\rho_0$ is a pure or a mixed state. When $\rho_0$ is a pure state, then the loss function is simply $1-F(\rho(T, \theta), \sigma)$ where $F(\rho(T, \theta), \sigma)$ is the quantum fidelity between $\sigma$ and $\rho(T, \theta)$. When $\rho_0$ is a mixed state, then minimising the above loss function is equivalent to minimising the Hilbert-Schmidt norm of quantum states, since
\begin{align*}
\begin{aligned}
   &\ \text{argmin}_{\theta}\ \frac{1}{2}\|\rho(T, \theta) -\sigma\|_{\text{HS}}^2 = \text{argmin}_{\theta}\ \frac{1}{2}\tr\big((\rho(T,\theta) - \sigma)^2\big) \\
   &\qquad = \text{argmin}_{\theta}\ \frac{1}{2}\Big(\tr\big(\rho(T,\theta)^2\big) + \tr\big(\sigma^2\big)\Big) - \tr\big(\sigma\rho(T,\theta) \big) \\
   &\qquad = \text{argmin}_{\theta}\ \tr\big(\rho_0^2) - \tr\big(\sigma\rho(T,\theta) \big) \\
   &\qquad =\text{argmin}_{\theta}\ \mathcal{L}_{\theta},
   \end{aligned}
\end{align*}
where $\| \cdot \|_{\text{HS}}$ denotes the Hilbert-Schmidt norm and the purity of $\rho(T, \theta)$ does not change under unitary transformation $U(T, \theta)$ so $\tr(\rho^2(T, \theta)) = \tr(\rho_0^2)$ and similarly applies to $\sigma = U \rho_0 U^\dagger$. For simulations of a toy example, see Section~\ref{eg::qsp}.\\

We note that this differs from the canonical target state preparation, since we assume the classical description of the target state is not provided. Although the target is a quantum state that cannot be cloned, we assume access to trusted copies of a pure quantum state $\dloss \vert_{\rho=\rho(T)}\propto \sigma$ from a third party, but the state preparation procedure (i.e., the black-box quantum dynamics $U$) is hidden from us. Our goal is to learn this preparation. This can be viewed as a certain quantum analogue of a generative model, whose classical version aims to find a dynamics (or mapping) to reproduce a probability distribution given finite samples. \\

\noindent \textit{Hamiltonian-learning with access to input and output states --} Suppose there is a Hamiltonian dynamics with $H(t)$ and we only have access to the inputs and outputs of the corresponding unitary evolution $U(t)$, and we can control the total evolution time \cite{wiebe_2014_hamiltonian}. The goal is to find a representation of this unknown system by identifying the Hamiltonian $H(t, \theta) \sim H(t)$ that generates this unitary. Such a Hamiltonian may describe either a finite-dimensional quantum system (Hamiltonian learning for ODEs) or an infinite-dimensional (continuous-variable) quantum system (Hamiltonian learning for PDEs). Unlike the previous example, where we assumed access to a single output state $\sigma$, here we assume the access to the entire unitary generated by the unknown Hamiltonian, thus allowing us to use different input states and possibly different evolution time. \\

To probe the system, we prepare a set of initial (pure) states $\{\rho_j\}_{j=1}^{\sample}$, which are input into the black box quantum system, which returns the corresponding quantum states at various times $T_j$, denoted as $\{\sigma_j(T_j) \}_{j=1}^{\sample}$. We can also input these initial states into a controllable and trustful unitary system generated by a parameterised Hamiltonian $H(t, \theta)$, with outputs $\big\{\rho_j(T_j, \theta)\big\}_{j=1}^{\sample}$. Our aim is to find a parameterised Hamiltonian $H(t, \theta)$ that best approximates the corresponding unitary evolution. The loss function is naturally defined as:
\begin{align}
\label{eqn::loss_eg2}
\mathcal{L}_{\theta} = \frac{1}{\sample} \sum_{j=1}^{\sample} \Big( 1 - \tr\big(\sigma_j(T_j) \rho_j(T_j,\theta)\big)\Big),
\end{align}
which is simply a linear combination of loss functions as in Eq.~\eqref{eqn::loss_generative} from the previous example.
We remark that this loss function is similar to the one used in gate synthesis tasks in \cite{leng2022differentiable} mathematically -- but their application settings are different. This approach has a similar flavor as the interactive quantum likelihood evaluation (IQLE) in \cite{wiebe_2014_hamiltonian}. For a numerical example, see Section~\ref{sec::example}.\\

\noindent \textit{Hamiltonian-learning with access to observables of output operators --} In this case we are still in the Hamiltonian-learning framework as in the previous example, except now we assume access only to the expectation values of observables, not the initial and final quantum state outputs. In this case, we minimize
\begin{align}
\label{eqn::loss_observable}
\mathcal{L}_{\theta}=\frac{1}{\sample N} \sum_{i=1}^{\sample} \sum_{j=1}^{N} \abs{\tr\big(\obs_j \sigma_i(T_i)\big) - \tr\big(\obs_j \rho_i(T_i,\theta)\big)}^2, 
\end{align}
where $O_j$ are observables. This loss function aims to match the tomography data of the target and parameterised states \cite{odonnell_efficient_2016}, or, in mathematical terms, estimates the discrepancy of target and parameterised states using the weak formalism in training.\\

For each term in the sum with observable operator $O_j$, fixed input state $\rho_i$, and time $T_i$ (we omit subscripts below for simplicity), the functional derivative is
\begin{align}
\label{eqn::deri_observable}
\frac{\delta \mathcal{L}}{\delta \rho}\Big\rvert_{\rho=\rho(T,\theta)} = 2 \Big(\tr\big(O \rho(T,\theta)\big) - \tr\big(O \sigma(T)\big)\Big) O,
\end{align}
so $\dloss \vert_{\rho=\rho(T,\theta)} \propto O$.  When $\rho(T, \theta)$ is a finite-dimensional quantum system, 
all observables $O$ can be written in terms of weighted sums of states in the computation basis: 
when $\obs$ are local Pauli matrices, e.g., $\obs = \sigma_X\otimes \unit$, then we can decompose this as $\obs/2 = \ketbra{+}\otimes \nicefrac{\unit}{2} - \ketbra{-} \otimes \nicefrac{\unit}{2}$, where each term corresponds to a density matrix; see also Appendix~\ref{appendix::pauli} for other examples. 
Since $\dloss \vert_{\rho=\rho(T,\theta)}$ above is not normalised, there is a prefactor term that affects the contribution of each gradient term when summed, and thus influences the learning rate. See a numerical example and discussion about dealing with the normalisation prefactor in Section~\ref{section::hl_observable}.

\subsection{Learning in open quantum systems}
\label{sec:appopen}

We examine the loss function in Eq.~\eqref{eq:lossdef} where $\rho(t, \theta)=\mathcal{E}(t, \theta)(\rho_0)$ now undergoes open quantum dynamics. This scenario  arises from the ability to probe only a sub-system of a large closed quantum system. 
Our goal is to find an effective closed Hamiltonian dynamics as a surrogate model for the original Hamiltonian. Hence, we parameterise the purification of $\rho(t,\ \theta)$ by defining a parameterised unitary $U(t, \theta)$ and a pre-determined initial ancilla state $|e\rangle \langle e|$ where $\mathcal{E}(t, \theta)(\rho_0)=\tr_e \big(U(t, \theta)(\rho_0 \otimes |e\rangle \langle e|) U(t, \theta)^{\dagger} \big)$. Thus, Eq.~\eqref{eq:lossdef} becomes 
\begin{align*}
    \mathcal{L}_{\theta}=\tr\big(M\ R(T, \theta)\big), \qquad M= O \otimes \mathbf{1}_e, \qquad R(t, \theta)=U(t, \theta)(\rho_0 \otimes |e\rangle \langle e|)U(t, \theta)^{\dagger}, 
\end{align*}
where $\mathcal{L}_{\theta}$ and the unitary dynamics for $R(t, \theta)$ can still be described with a (closed) QNODE/QNPDE. Then the same applications in Section~\ref{sec:appclosed} still apply here, except here we are instead learning the purification of a target state and the Hamiltonians generating the purification.
\\

We can also look at an alternative loss function that is nonlinear in $\rho(t, \theta)$, namely one minus the purity of the state
\begin{align}
\label{eqn::purity}
    \mathcal{L}_{\theta}\ =\ 1 - \tr\big(\rho^{2}(T, \theta)\big), \qquad \rho(t, \theta)=\tr_e \big(U(t, \theta)(\rho_0 \otimes |e\rangle \langle e|) U(t, \theta)^{\dagger}\big).
    \end{align}
This loss function is associated with learning the purification of $\rho(t, \theta)$ itself. Here $$\dloss \Big\vert_{\rho=\rho(T)} \propto \rho(T, \theta),$$
which we assume can be prepared to serve as an input quantum state to the gradient estimation algorithms in Section~\ref{sec:mixed}. 
Another nonlinear loss function is the R{\'e}nyi entropy, $\mathcal{R}_n(\rho) = \frac{1}{1-n} \ln \tr\big(\rho^n)$, which has been used to study the entanglement \cite{wang_calculating_2020} if $\rho$ is a reduced density matrix as above, and the $n=2$ case is related to purity. For $n>3$, $\delta \mathcal{L}/\delta \rho \propto \rho^{n-1}/\tr(\rho^n)$ needs to be prepared, and this could be achieved if one prepares multiple copies of $\rho(T,\theta)$. For simplicity, we will not consider these settings in this paper. 

\subsection{Learning ODEs and PDEs}
\label{sec:appodepde}

Learning ODEs and PDEs is an important application area in machine learning. 
When the system dimension is very large, computing gradients of learning objectives becomes costly because it at least first requires simulating high-dimensional ODEs/PDEs.
Most classical numerical methods, like finite-difference, finite element or spectral methods, generally suffer from the curse of dimensionality. 
Quantum algorithms offer an option: by simulating classical ODEs/PDEs on quantum devices,  it could potentially mitigate the issue of the curse of dimensionality; 
to achieve the training, we could embed these classical learning problem into their corresponding QNODE/QNPDEs framework, so as to employ quantum algorithms also for the gradient computation required (see algorithms in \secref{sec:mixed}). 
This quantum approach could serve as an alternative route to classical machine learning for high-dimensional PDEs \cite{e_deep_2018,raissi_physics-informed_2019}. \\

{\noindent \emph{Part (i): Embed the dynamics into Hamiltonian learning.}}\\

We focus on linear ODEs and PDEs for $\mathbf{u}(t)$ of the form
\begin{align} \label{eq:generalde}
    \frac{d \mathbf{u}(t)}{dt}=-i\mathbf{A}(t) \mathbf{u}(t), \qquad \mathbf{u}(t=0)=\mathbf{u}_0,
\end{align}
where in general $\mathbf{A}(t) \neq \mathbf{A}^{\dagger}(t)$. Eq.~\eqref{eq:generalde} denotes a system of $D$-dimensional linear ODEs when $\mathbf{u}(t)=\sum_{i=1}^D u_i(t)|i\rangle$ is a $D$-dimensional quantum state and $\mathbf{A}(t)$ is a $D \times D$ matrix. Eq.~\eqref{eq:generalde} can also denote a $d+1$-dimensional linear PDE when $\mathbf{u}(t)=\int u(t,x)|x\rangle dx$ is a system of $d$ qumodes with $x \in \mathbb{R}^d$, and $\mathbf{A}(t)$ is composed of quadrature operators $\hat{x}, \hat{p}$ using the correspondence $x \rightarrow \hat{x}$, $\partial/\partial x \rightarrow i \hat{p}$, see \cite{analogPDE}. Eq.~\eqref{eq:generalde} can also denote a system of $D$ linear $d+1$-dimensional PDEs, then $\mathbf{u}(t)=\int \sum_{i=1}^D u_i(t, x)|i\rangle |x\rangle dx$ is a hybrid CV-DV quantum state.\\

It is possible to use the (closed) QNODE/QNPDE formalism for general \textit{non-unitary} evolution in Eq.~\eqref{eq:generalde} if we can map the system in Eq.~\eqref{eq:generalde} in a very simple way to the system in Eq.~\eqref{eq:schrdensity1}, i.e., a simple correspondence between $\mathbf{A}(t)$ and the Hamiltonian $\mathbf{H}(t)$. This can be realized by  Schr\"odingerisation (see \cite{analogPDE, PRLschr} and references therein), where the corresponding system in Eq.~\eqref{eq:schrdensity1} acts on a system of the same size as Eq.~\eqref{eq:generalde} plus one ancilla qumode. With Schro\"odingerization, one can recover $\mathbf{u}(t)$ from Eq.~\eqref{eq:generalde} by first simulating the solution $\mathbf{v}(t)$ following unitary dynamics (equivalent to Eq.~\eqref{eq:schrdensity1} with the association $\rho(t)=\mathbf{v}(t) \mathbf{v}^{\dagger}(t)$ up to a normalisation)
\begin{align}
\label{eqn::schrodingerisation}
\begin{aligned}
  &  \frac{d \mathbf{v}(t)}{dt}=-i \mathbf{H}(t) \mathbf{v}(t), \qquad \mathbf{H}(t)=\mathbf{A}_2(t) \otimes \hat{\eta}+\mathbf{A}_1(t) \otimes \mathbf{1}_{\eta}, \\
    & \mathbf{v}(t=0)=\mathbf{u}_0 \otimes |\Xi\rangle, \quad |\Xi\rangle = \int e^{-\abs{\xi}}\ |\xi\rangle d\xi, \qquad \xi \in \mathbb{R}, \\
    & \mathbf{A}_1(t)=\frac{1}{2}\big(\mathbf{A}(t)+\mathbf{A}^{\dagger}(t)\big) = \mathbf{A}_1^{\dagger}(t), \qquad \mathbf{A}_2(t)=\frac{i}{2}\big(\mathbf{A}(t)-\mathbf{A}^{\dagger}(t)\big) = \mathbf{A}_2^{\dagger}(t),
    \end{aligned}
\end{align}
where $\hat{\xi}, \hat{\eta}$ are conjugate quadrature operators obeying $\comm{\hat{\eta}}{\hat{\xi}}=\ i\mathbf{1}_{\eta}$ and they have corresponding orthonormal bases $\{|\xi\rangle\}_{\xi \in \mathbb{R}}, \{|\eta\rangle \}_{\eta \in \mathbb{R}}$ with $\hat{\xi}|\xi\rangle=\xi |\xi\rangle$ and $\hat{\eta}|\eta\rangle=\eta |\eta \rangle$. Then $\mathbf{u}(t)$ can be recovered from $\mathbf{v}(t)$ by projecting the extra ancilla qumode onto the positive eigenvalues $|\xi>0\rangle$, so that for the normalised state $|u(t)\rangle=\frac{\mathbf{u}(t)}{\|\mathbf{u}(t)\|}$, one has 
\begin{align*}
   \lvert u(t)\rangle\langle u(t) \rvert \propto \tr_{\xi}\Big( \big(\mathbf{1}_x \otimes \Pi_{\xi>0}\big) \mathbf{v}(t)\mathbf{v}(t)^\dagger \big(\mathbf{1}_x \otimes \Pi_{\xi>0}\big)\Big), \qquad \Pi_{\xi>0}= \int_{\xi>0} \ketbra{\xi} d\xi = \Pi^{\dagger}_{\xi>0}.
\end{align*}

Suppose we want to learn a linear ODE/PDE with a known form, but with unknown parameters. Most ODEs and PDEs that arise from applications in scientific computing and engineering  usually have fairly restrictive forms. This means that the universality of representation is in fact not a requirement, and one might be learning only a relatively small number of unknown parameters (for example viscosity and heat conductivity in gas dynamics). We can thus expand 
\begin{align} \label{eq:aexpansion}
    \mathbf{A}(t, \theta)=\sum_{k=1}^K c_k(t, \theta) A_k,
\end{align}
with the known set of operators $\{A_k\}_{k=1}^K$ that is independent of time and tuneable parameters $\theta$. 
For classical ODE/PDE problems, typically we may assume that $c_k\in \Real$.
These are easily determined from the given form of ODE/PDE. The corresponding Hamiltonian in the QNODE/QNPDE formalism can now also parameterised as 
\begin{align} \label{eq:hexpansion}
\begin{aligned}
    \mathbf{H}(t, \theta)
    =&\ \frac{1}{2}\sum_{k=1}^K (c_k(t, \theta) A_k + c_k^*(t,\theta) A_k^\dagger) \otimes \mathbf{1}_{\eta}+ \frac{1}{2} \sum_{k} (c_k(t,\theta) i A_k - i c_k^*(t,\theta) A_k^\dagger)  \otimes \hat{\eta} \\
    =&\ \sum_{k=1}^K \underbrace{c_k(t,\theta)}_{f_k(t,\theta)} \underbrace{\frac{A_k + A_k^\dagger}{2} \otimes \mathbf{1}_{\eta}}_{H_k} + \sum_{k=1}^K \underbrace{c_k(t,\theta)}_{f_{k+K}(t,\theta)} \underbrace{\frac{i A_k - i A_k^\dagger}{2} \otimes \hat{\eta}}_{H_{k+K}} \qquad (\text{ if } c_k\in \Real) \\
    =&\  \sum_{k=1}^{2K} f_k(t,\theta) H_k.
    \end{aligned}
\end{align}

{\noindent \emph{Data availability and the loss function}}:\\

To learn these parameters $\theta$, we assume access to values of observables of the true solution $\sigma(T) := \mathbf{u}(t=T)\mathbf{u}(t=T)^\dagger$, which takes the form $\tr\big(\obs_j \sigma(T)\big) \equiv \tr\big(\obs_j \mathbf{u}(t=T) \mathbf{u}(t=T)^\dagger \big)$, as well as access to the normalisation constants $\tr\big(\sigma(T)\big)$. 
More concretely, 
suppose we have data at collocation points, e.g., $u_j = \mathbf{u}(x_j, t=T)$ for which we can write this observable data as 
\begin{align}
\label{eqn::data}
\abs{u_j}^2 \approx \frac{1}{\abs{\bregion_j}} \tr\big(\int_{\bregion_j}\ \ketbra{x} \sigma(T) dx \big),
\end{align}
where $\bregion_j$ is a small sub-region with Lebesgue measure $\abs{\bregion_j}$. Hence, we can choose $O_j = \frac{1}{\abs{\bregion_j}}\int_{\bregion_j} \ketbra{x}\ dx$. Therefore, we also know the normalisation constant, 
\begin{align*}
\tr\big(\sigma(T)\big) \approx \sum_{j} \abs{\bregion_j} \cdot \abs{u_j}^2.
\end{align*}
For physical problems, e.g., when $u_j$ represents the temperature, then $u_j$ is naturally positive. Assume the access to the above quadratic form in Eq.~\eqref{eqn::data} is then natural.  We will then denote the normalised data as
\begin{align*}
\frac{\tr\big(O_j \sigma(T)\big)}{\tr\big(\sigma(T)\big)} = \bar{u}_j.
\end{align*}

Similar to the PINN loss \cite{raissi_physics-informed_2019}, we can write the corresponding loss function as 
\begin{align} \label{eq:lossode}
\begin{aligned}
   \mathcal{L}_{\theta} &= \frac{1}{N} \sum_{j=1}^N \mathcal{L}_j\big(\rho(T, \theta)\big)
   := \frac{1}{N}\sum_{j=1}^N \big\lvert \frac{\tr\big(O_j\sigma(T)\big)}{\tr\big(\sigma(T)\big)} - \tr\big(O_j \rho(T, \theta)\big) \big\rvert^2 = \frac{1}{N}\sum_{j=1}^N \big\lvert \bar{u}_j - \tr\big(O_j \rho(T, \theta)\big) \big\rvert^2,
   \end{aligned}
   \end{align}
   where the reduced density matrix
\begin{align*}
\rho(T, \theta) = \frac{\tr_{\xi}\Big(\big(\mathbf{1}_x \otimes \Pi_{\xi>0}\big)\ \ketbra{v(T, \theta)}\Big)}{\tr\Big(\big(\mathbf{1}_x \otimes \Pi_{\xi>0}\big)\ |v(T, \theta)\rangle \langle v(T, \theta)|\Big)},
\end{align*}
and where the normalised wave function $ |v(t,\theta)\rangle=\frac{\mathbf{v}(t,\theta)}{\|\mathbf{v}(t,\theta)\|}$.
We can rewrite 
\begin{align*}
\begin{aligned}
   \tr \big(O_j \rho(T, \theta)\big) & = \frac{\tr\Big(\big(O_j \otimes \mathbf{1}_{\xi})(\mathbf{1}_x \otimes \Pi_{\xi>0}) \ketbra{v(T, \theta)} \Big)}{\tr\Big(\big(\mathbf{1}_x \otimes \Pi_{\xi>0}\big)\ \ketbra{v(T, \theta)}\Big)}  \propto \tr\big(M_j |v(T, \theta)\rangle \langle v(T, \theta)|\big),
   \end{aligned}
   \end{align*}
   where $M_j = O_j \otimes \Pi_{\xi>0}$ is Hermitian.
Thus, the loss function $\mathcal{L}_{\theta}$ in Eq.~\eqref{eq:lossode} can be rewritten in a form similar to Eq.~\eqref{eq:lossdef}, and also corresponds to a (closed) QNODE/QNPDE. In this  case, the functional derivative takes the following form
\begin{align*}
&\ \frac{\delta\loss_j}{\delta\ketbra{v(T,\theta)}} \\
=&\ -2 \frac{\Big(\bar{u}_j - \tr\big(O_j \rho(T, \theta)\big) \Big)}{\tr\Big(\big(\mathbf{1}_x \otimes \Pi_{\xi>0}\big)\ \ketbra{v(T, \theta)}\Big)}  
\Big(\big(O_j \otimes \Pi_{\xi>0} \big) - \tr\big(O_j \rho(T,\theta)\big) \big(\unit_x \otimes \Pi_{\xi>0} \big)\Big).
\end{align*}
Since any positive fixed constant could be theoretically absorbed into the learning rate (as will be more clear in later sections), it is sufficient to consider 
\begin{align*}
\frac{\delta\loss_j}{\delta\ketbra{v(T,\theta)}} 
\propto &\ - \Big(\bar{u}_j - \tr\big(O_j \rho(T, \theta)\big) \Big)
\Big(\big(O_j \otimes \Pi_{\xi>0} \big) - \tr\big(O_j \rho(T,\theta)\big) \big(\unit_x \otimes \Pi_{\xi>0} \big)\Big).
\end{align*}
Note that the value $\bar{u}_j$ comes from data and is known; $\tr\big(O_j \rho(T,\theta)\big)$ could be estimated from simulating Eq.~\eqref{eqn::schrodingerisation}.\\

{\noindent \emph{Discussions and remarks:}}\\

Note that the loss function in $\mathcal{L}_{\theta}$ in Eq.~\eqref{eq:lossode} might resemble loss functions in classical machine learning for learning PDEs via values at different collocation points, but here we are not learning the representation of the solution directly, but rather learning the Hamiltonian that generates the solution (which is closer in spirit to finding a generative model). Once the parameters $\theta$ in $\mathbf{H}(t, \theta)$ in Eq.~\eqref{eq:hexpansion} are learned, this is easily translatable to identifying the unknown parameters in the original ODE/PDE with the corresponding $\mathbf{A}(t, \theta)$ in Eq.~\eqref{eq:aexpansion}, and is a simple one-to-one direct correspondence. \\

This is not only useful for learning unknown parameters of known ODEs/PDEs that arise naturally in physical systems, but also potentially in machine learning itself. For instance, in both continuous normalising flows and flow matching the task is to identify the unknown parameters in the flux of a transport equation.\\

We can also extend these methods by using Schr\"odingerisation applied to more general linear ODEs/PDEs with inhomogeneous terms and higher-order time derivatives, and also to nonlinear ODEs and nonlinear Hamilton-Jacobi and nonlinear hyperbolic PDEs \cite{analogPDE} and even ill-posed linear PDEs \cite{illposed}. \\

We emphasise that for many application where we know the form of the underlying dynamics, then the number of local terms $H_k$ is relatively low (the number of unknown parameters may also possibly be low). Thus, we do not require a universal representation for the corresponding Hamiltonian, and the ansatz for the form of the Hamiltonian is directly related to domain knowledge and has the least number of parameters to learn.  This shows the naturalness and power of the QNODE/QNPDE formalism for these problems, since ODEs and PDEs from nature are automatically continuous in time, PDEs are continuous in space and time, and this continuous form can be preserved in the learning procedure. 

\subsection{Learning non-autonomous systems}
\label{sec:appnonauto}

While the above applications can be generally non-autonomous (i.e., the corresponding $H(t, \theta)$ have explicit time-dependence), there is an alternative method of making this into a time-independent problem as well, using the Sambe-Howland clock formalism \cite{sambe_steady_1973,howland_stationary_1974,Nonauto,Nonauto2}. We write down the formulation for unitary evolution, but this can also be applied to non-unitary evolution from Section~\ref{sec:appodepde}. \\

Suppose we have an initial quantum state $\rho_0$ and a time-dependent Hamiltonian ansatz $H(t, \theta)$. We can turn the task of learning an unknown time-dependent Hamiltonian $H(t, \theta)$ (for any of the loss functions in Section~\ref{sec:appclosed}, ~\ref{sec:appopen} and ~\ref{sec:appodepde}) into learning a corresponding time-independent Hamiltonian $\bar{H}(\theta)=\hat{p}_s \otimes \mathbf{1}+H(\hat{s}, \theta)$, where $H(\hat{s}, \theta)$ is $H(t, \theta)$ with $t$ replaced by a quadrature operator $\hat{s}$. If $\rho(t, \theta)$ obeys Eq.~\eqref{eq:schrdensity1}, then the solution $\rho(t, \theta)$ can be reproduced by evolving an augmented system $\sigma(t, \theta)$ (by including one qumode representing a clock mode) with a time-independent Hamiltonian $\bar{H}(\theta)$:
\begin{align*}
\begin{aligned}
\frac{d}{d t}\sigma(t, \theta)) = -i\comm{\bar{H}(\theta)}{\sigma(t, \theta)}, \qquad \sigma(0) = G\otimes \rho_0, \\ G = \int g(s) \ketbra{s}\ ds, \qquad \bar{H}(\theta)=\hat{p}_s \otimes \mathbf{1}+H(\hat{s}, \theta).
\end{aligned}
\end{align*}
When $g(s) \approx \delta(s)$, one has $\rho(t, \theta) \approx \tr_s \big(\sigma(t, \theta)\big)$. This means that the loss functions of the form in Eq.~\eqref{eq:lossdef} can be rewritten in the form suitable for (closed) QNODE/QNPDE
\begin{align*}
    \tr \big(O \rho(t, \theta)\big) \approx \tr\Big(O \tr_s\big(\sigma(t, \theta)\big)\Big)=\tr\big(M \sigma(t, \theta)\big), \qquad M=O \otimes \mathbf{1}_s. 
\end{align*}
Thus, simulating a non-autonomous machine learning problem is, in principle, not fundamentally different from using a time-independent ansatz. The impact of the approximation $g(s)\approx \delta(s)$ has been analyzed in \cite{Nonauto,Nonauto2} for Hamiltonian simulation. Investigating how this bias affects optimisation outcomes is left to future work.

\subsection{Framework for new quantum neural network architectures}
\label{sec:appqnn}

The continuous-time framework for classical dynamics offers advantages in designing novel ResNet-type architectures \cite{lu_beyond_2018, haber_stable_2017, ruthotto_deep_2020}. This connection comes from a perspective that the continuous-time framework reduces to the canonical ResNet architecture under a particular time-discretisation (namely Euler discretisation) and would give rise to other architectures under different discretisation of the time coordinate. Despite significant progress, continuous-time analog classical algorithms remain relatively under-explored.\\

We anticipate a similar approach would  also apply to quantum machine learning. In quantum machine learning, most existing work instead focuses on discrete neural network architectures, leading to discrete-time quantum gates. Most of these architectures based on qubit systems are, a fact, a particular time-dependent schedule $f_k$ to the (closed) QNODE, and similarly for (closed) QNPDE for those systems based on qumodes. \\

\noindent \textit{Connection of QNODE/QNPDE to quantum neural networks:}\\ 

In what follows, we will first revisit and explain why many quantum neural networks are essentially (closed) QNODE and QNPDEs, where we are seeking a functions $f_k(t, z, \theta)$ using a \textit{piecewise-in-time expansion}. It has been briefly mentioned in e.g., \cite{magann_pulses_2021,leng2022differentiable}, and this focus on piecewise-in-time expansion could be one reason behind why the analog approach to quantum machine learning has not been initiated much earlier.  \\

Typically in quantum neural networks one is trying to approximate a function $g(z) \in \mathbb{R}$ for data $z \in \mathbb{R}^D$ in $D$-dimensional feature space which is partly or fully embedded in a quantum state $|\psi(z)\rangle$, using $L$ layers of discrete gates $U_l(z, \theta)$
\begin{align} \label{eq:universalg}
\begin{aligned}
    g(z) \approx \tilde{g}_L(z) &= \tr\left(U(z, \theta)(|\psi(z)\rangle \langle \psi(z)|\otimes |0\rangle \langle 0|)U(z, \theta)^{\dagger} (\mathbf{1}_{\psi} \otimes O)\right),  \\
    &U(z, \theta)=\Pi_{l=1}^L U_l (z, \theta), \qquad O=O^{\dagger}.
\end{aligned}
\end{align}
With appropriate architecture design, one could expect $|g(z)-\tilde{g}_L(z)| \rightarrow 0$ as  the layer depth $L \rightarrow \infty$.
Each $U_\ell$ typically involves the unitary evolution of time-independent Pauli or multi-Pauli matrices, where the rotation angle (or say the time) is the parameter to tune.
Let us denote all such a collection of Pauli or multi-Pauli matrices as $\{H_k\}_{k=1}^K$.
In the case of a quantum classifier for binary classification, one can use $O=\sigma_z$ \cite{farhi2018classification}. Since loss functions in these cases will include terms similar to  
\begin{align*}
    \frac{1}{N}\sum_{i=1}^N \abs{g(z_i)-\tilde{g}_L(z_i)}^2,
\end{align*}
for $N$ training data, this ansatz takes the form of a (closed) QNODE or QNPDE, where the corresponding unitary dynamics
\begin{align*}
\begin{aligned}
& U(z, \theta)= U(T, z, \theta)=\mathcal{T}e^{-i \int_0^T H(\tau, z, \theta)d \tau}=\Pi_{l=1}^L U_l(z, \theta), \\
& U_l(z, \theta)=\mathcal{T}e^{-i \int_{t_l}^{t_{l+1}}H_l(\tau, z, \theta) d \tau}, \qquad t_L=T, 
\end{aligned}
\end{align*}
is generated by the Hamiltonian 
\begin{align}
    H(t, z, \theta)=\sum_{l=1}^L \big(\Theta(t-t_l)-\Theta(t-t_{l+1})\big) h_{\ell}(z, \theta) H_{k(\ell, z)},
\end{align}
where $k(\ell, z)$ is the index of the basis Hamiltonian, which could repeat over the trajectory, namely, we can apply the same operation at multiple times; $\{ t_\ell \}$ generates a sequence of time interval, which could be simply chosen as $t_\ell = \ell$ without loss of generality; $h_\ell(z,\theta)$ is some generic function as coefficients.
Then we can clearly also rewrite this as a sum of Hamiltonians $H_k$:
\begin{align*}
    H(t, z, \theta)=\sum_{k=1}^K f_k(t, z, \theta) H_k, \quad H_k \neq H_{k'}, \quad \forall k \neq k'.
\end{align*}
The coefficient $f_k(t,z,\theta)$ could be expressed in terms of $h_\ell$. The design of quantum neural network could be viewed as determining these piecewise continuous control functions $f_k$.\\

\noindent \textit{Discussion on the complexity:}\\ 

The QNODE/QNPDE formalism in fact could assist in a deeper understanding of the concerns about the complexity scaling for many quantum neural networks.
Since quantum neural networks typically correspond to a piecewise-in-time expansion of the scheduling function, this means that the larger number of time intervals, the more accurate the approximation of each $f_k(t, z, \theta)$, and this translates into a larger $L$. For example, in many quantum neural networks, each layer $U_l(z, \theta)=e^{i \theta_l P_l}V_l$ for fixed Pauli operators $P_l$ and fixed unitaries $V_l$. For certain architectures like \cite{perez-salinas_data_2020}, $z$-dependence appears in each layer, and in other cases data embedding of $z$ is only at the beginning and does not appear at each $l$. Thus the total number of layers $L$ is equivalent to the total number of tuneable parameters $M$. 

From the analysis in Section~\ref{sec:complexity}, we learn that the total cost in estimating the gradient of the loss function quantified using $L_\infty$ norm for (closed) QNODEs could be generally $\bigO(T^2 K^2)$ (neglecting here the error scaling due to the uncertainty in quantum measurements for simplicity). Using the fairly common quantum neural network ansatz form, $K$ is the number of local operators which is thus typically $\text{poly}(n)$ where $n$ is the number of qubits. 
Since quantum neural networks typically uses the piecewise-in-time expansion for $f_k$, this leads to $T = \order{M}$, so the total cost is $\bigO(M^2)$, and this is the typical scaling that is often discussed in the quantum neural network literature. Herein, we could use a continuous-time framework to recover the complexity scaling for a discrete-time quantum neural network. As explained in Section~\ref{sec:complexity}, if we use improved sampling methods, e.g., the shadow tomography, we can expect certain parameters to improve \cite{abbas2305quantum}. However, we note that for many physics-motivated machine learning problems, $K = \text{poly}(n)$ and $T = \order{1}$, which means the quantum gradient algorithm is expected to be efficient. \\ 

From this we see that the concern about the quantum neural network  backpropagation scaling with respect to the number of parameters $M$ is only a result of the particular way that $U_l$ is defined such that $L = \order{M}$. Since here one deals with a discrete number of layers, and improvement in estimation being achieved with increasing $L$ means one is approximating $f_k(t, z, \theta)$ with a piecewise-in-time basis, so the precision of this approximation must therefore grow with the layer depth $L$, and hence also with the number of parameters $M$.  However, for particular problems where we fix the $K$ and we do not use the piecewise-in-time basis for estimating $f_k(t, \theta)$, then in fact, the cost of the gradient estimation scaling in $M$ is not a concern. For general quantum machine learning problems without prior information, a fixed $K$ may raise concern about the universality; however, this is not an intrinsic problem for physics-driven problems such as Hamiltonian learning of some many-body systems, or many quantum control problems.\\

Thus in this continuous-time framework, we see that particular quantum neural networks are particular instances with special choices on the basis of expansion for approximating $f_k(t, z, \theta)$. In other kinds of quantum neural networks like quantum convolutional neural networks \cite{cong_quantum_2019}, for a continuous-time version we will require (open) QNODEs/QNPDEs. \\ 

\noindent \textit{Discussion on potential advantages:}\\ 

The framework of quantum neural ODEs/PDEs could provide a new perspective for designing quantum circuit architectures inspired by continuous-time dynamics, analogous to how classical ODEs guide variants of ResNet design. We can design new quantum neural network architectures by considering both different parameterisation of $f_k(t, z, \theta)$, and applying different time discretisations. We will examine this in more detail in a later work. \\

For universal approximation in the sense of approximating an unknown function $g(z)$ (e.g., Eq.~\eqref{eq:universalg}), it is currently unclear, even for classical algorithms, which architecture -- ODE-based or discrete-time ResNet-based -- is more efficient, as this is likely problem-dependent, as far as we know.
Similarly, it is difficult to claim whether Hamiltonian dynamics-based architectures are generally more efficient than discrete-time quantum neural networks, if one's aim is universal approximation. In this context, this formulation can still serve as a mathematical tool from which one can study more general quantum neural networks and their gradient estimation algorithms, and provide an alternative viewpoint to generate new quantum architectures. \\

\section{Gradient estimation algorithm for QNODE and QNPDE}
\label{sec:mixed}

In this section, we will present quantum algorithms for estimating the gradient of loss function corresponding to the QNODE and QNPDE. This is a generalisation of classical backpropagation algorithm by extending to the continuous-time setting, and we use the idea of the backward adjoint equation used in classical neural ODEs \cite{2018NODE}. We will first revisit the adjoint equation in \secref{sec::revisiting} for the general setting, and then present our main gradient estimation algorithms in \secref{section::quantum_alg} (with time discretisation) and \secref{alg::no_discretization} (without time discretisation). The main quantum circuits for our two algorithms are summarized in the following \figref{fig::alg1}, and their costs will be discussed in Section~\ref{sec:complexity}. 
Finally in \secref{subsec::landscape}, we discuss the energy landscape for a Hamiltonian learning problem, and identify how the hardness of training connects to the quantum correlation matrix for the optimal Hamiltonian dynamics. 

\begin{figure}[h!]
\begin{subfigure}{0.32\textwidth}
\includegraphics[width=\textwidth]{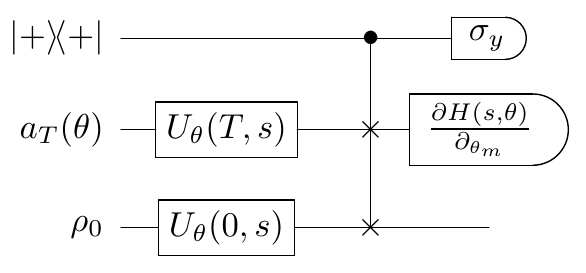}
\caption{Schematic circuit for Theorem~\ref{thm::main}}
\label{fig::alg1::parta}
\end{subfigure}
~\hspace{0.02\textwidth}
\begin{subfigure}{0.55\textwidth}
\includegraphics[width=\textwidth]{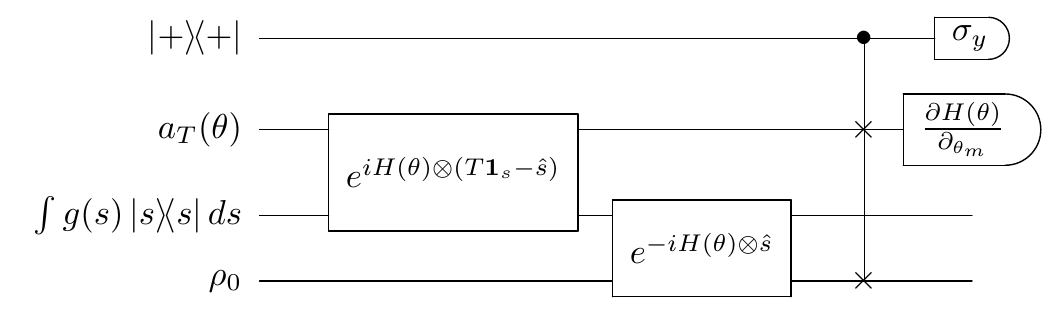}
\caption{Schematic circuit for \thmref{thm5}}
\label{fig::alg1::partb}
\end{subfigure}
\caption{The schematic quantum circuit for the gradient computation scheme with respect to parameter $\theta_m$,  $1\le m\le M$. 
We remark that if we parameterise $H(t,\theta)$ as in Eq.~\eqref{eqn::H_ansatz}, one does not need to perform quantum measurement for each $\frac{\partial H(s,\theta)}{\partial \theta_m}$ where $1\le m\le M$; see \secref{sec::cost}.}
\label{fig::alg1}
\end{figure}

\subsection{Revisiting the adjoint equation: application to QNODEs and QNPDEs}
\label{sec::revisiting}

From Definition~\ref{def:qnode}, we see that a (closed) QNODE and QNPDE evolves with the equation of motion 
\begin{align} \label{eq:schrdensity}
    \frac{\partial \rho(t, \theta)}{\partial t}=-i[H(t, \theta), \rho(t, \theta)], \qquad \rho(0, \theta)=\rho_0,
\end{align}
so an initial state $\rho_0$ evolves via unitary dynamics with respect to the unitary $U_{\theta}(0, t)$ where
\begin{align}
    \rho(t, \theta)=U_{\theta}(0,t)\rho_0U_{\theta}(0,t)^\dagger, \qquad U_{\theta}(0, t)=\mathcal{T} \exp(-i \int_0^t H(\tau, \theta) d \tau). 
\end{align}
Here $\rho(t, \theta)$ is a density matrix which consists of qubits (or qumodes), and this represents a corresponding (closed) QNODE (or QNPDE) once we define a loss function $\mathcal{L}_{\theta}$. 
Its gradient could be expressed in the following way:

\begin{lemma} \label{lem:3}
    The gradient of the loss function $\mathcal{L}_{\theta} = \mathcal{L}\big(\rho(T,\theta)\big)$ with respect to the parameter $\theta_m$ where $m = 1, \cdots, M$ can be written as 
    \begin{align} \label{eq:rhodensitylemma}
       &  \frac{\partial \mathcal{L}_{\theta}}{\partial \theta_m}= i A_T(\theta)\int_0^T \tr\left(\left[\frac{\partial H(s, \theta)}{\partial \theta_m},a(s, \theta)\right] \rho(s, \theta)\right) ds, 
    \end{align}
    when the state $\rho(s, \theta)$ evolves according to Eq.~\eqref{eq:schrdensity}, and $a(s, \theta)$ evolves from a terminal condition with the same unitary dynamics as Eq.~\eqref{eq:schrdensity}
    \begin{align} \label{eq:asmatrix}
    \begin{aligned}
        \frac{\partial a(s, \theta)}{\partial s} = & -i[H(s, \theta), a(s, \theta)], \\
  a(T, \theta) = &\ a_T(\theta) = \frac{1}{A_T(\theta)} \frac{\delta \loss}{\delta \rho} \Big|_{\rho(T, \theta)},  \qquad 
        A_T(\theta) =\tr \Big(\frac{\delta \loss}{\delta \rho} \Big|_{\rho(T, \theta)}\Big).
        \end{aligned}
    \end{align}
\end{lemma}
\begin{proof}
This result essentially relies on the adjoint equation techniques used in e.g., classical neural ODEs \cite{2018NODE}, but here specialised to the case of unitary dynamics. We provide a customized and self-contained proof in Appendix~\ref{app:theorem3}.
\end{proof}
 We make two remarks below:
\begin{itemize}
\item {\it Preparing the terminal state $a_T(\theta) \propto \dloss \vert_{\rho=\rho(T)}$}: In this formalism, we require access to the preparation of the quantum state $a_T(\theta)$. It's important to note that this is the \textit{only} place where the form of the loss function appears. For certain loss functions, this state could be difficult to prepare. However, for many important examples that we have discussed in Section~\ref{sec:scopeapps}, it is possible as the functional derivative is directly computable and has a simple form. In fact it is easy to see that $a_T(\theta)$ can also be a linear combination of quantum states that can be prepared, and the gradient in Eq.~\eqref{eq:rhodensitylemma} is then computed as a sum of the contribution to the gradient from each of those quantum states. See Corollary~\ref{eq:cor1}.

\item {\it Normalisation $A_T(\theta)$:} When we use the classical optimisation algorithms to update the parameters, then the magnitude of $A_T$, in fact, does not matter, as it can be absorbed into the learning rate. This also brings certain theoretical convenience: we only need to prepare the quantum state proportional to $\dloss \vert_{\rho=\rho(T)}$ and in this formalism, one does not need to worry about the normalisation. However, we do need to know the sign of $A_T(\theta)$, which is easily computable for many examples discussed in Section~\ref{sec:scopeapps}.
\end{itemize}

For more concrete examples, one may refer to \secref{sec::example} about numerical examples.

\subsection{Quantum gradient algorithm: partial time discretisation}
\label{section::quantum_alg}

To develop a quantum algorithm for computing the gradient of the loss function, we want to rewrite the expression in Eq.~\eqref{eq:rhodensitylemma} in terms of quantum expectation values. Based on Lemma~\ref{lem:3}, we have the following theorem. A schematic quantum circuit is provided in Figure~\ref{fig::alg1::parta}.

\begin{theorem}[Quantum gradient algorithm: partial time discretisation]
\label{thm::main}
Consider the Hamiltonian dynamics of quantum state $\rho_0$ evolving under a parametric Hamiltonian dynamics in Eq.~\eqref{eq:schrdensity} with an accompanying loss function $\mathcal{L}_{\theta} = \mathcal{L}\big(\rho(T,\theta)\big)$ at a terminal time $T$. Assume access to the quantum state $a_T(\theta)\propto  \delta \mathcal{L}/\delta \rho \vert_{\rho=\rho(T,\theta)}$ (see Lemma~\ref{lem:3}). Then the gradient of the loss function can be estimated in terms of the time-integral of the expectation of the Hermitian operator $\obs_m(s,\theta)$ for arbitrary $m$ ($1\le m \le M$)
\begin{align} \label{eq:mixedgradient2}
\begin{aligned}
\frac{\partial \mathcal{L}_{\theta}}{\partial \theta_m} &=  2 A_T(\theta) \int_0^T  \tr\big(\obs_m(s, \theta) \eta(s,\theta)\big)\ ds,  \\
\obs_m(s, \theta) &= \sigma_y \otimes \frac{\partial H(s,\theta)}{\partial \theta_m} \otimes \unit_{\hbt}, \qquad m=1, \cdots, M.
\end{aligned}
\end{align}
The quantum state $\eta(s,\theta)$ can be prepared using a unitary operation $\mathcal{V}_{\theta}(s)$ in the following way:
\begin{align} \label{eq:psisstate2}
\begin{aligned}
  \eta(s,\theta) &= \mathcal{V}_\theta(s) \eta(0,\theta) \mathcal{V}_{\theta}(s)^\dagger, \\
  \eta(0, \theta) &= \ketbra{+} \otimes a_T(\theta) \otimes \rho_0, \\
   \mathcal{V}_{\theta}(s) & := \controlswap
    \big(\unit \otimes U_{\theta}(T, s) \otimes U_{\theta}(0, s)\big), \\
  \controlswap &:= |0\rangle \langle 0|\otimes \mathbf{1}\otimes \mathbf{1} + |1\rangle \langle 1| \otimes \boldsymbol{S},
    \end{aligned}
\end{align}
where $\boldsymbol{S}$ is the SWAP operator and $U_{\theta}(s,s')=\mathcal{T} \exp(-i \int_{s}^{s'} H(\tau, \theta) d \tau)$. The Hilbert space for the quantum circuit consists of a qubit as ancilla $\hbt_{\aux}$, adjoint state in $\hbt_{\adj}$ (with the same size as $\hbt$), and the original Hilbert space $\hbt$.
\end{theorem}
\begin{proof}
See Appendix~\ref{app:theoremmixedtopure}.   
\end{proof}

To turn the above theorem into a more concrete algorithm, we can do the following:
\begin{itemize}
\item It is possible to compute the time integral in Eq.~\eqref{eq:mixedgradient2} by discretising in time, so we can choose $\numtime$ grid points $s_1, s_2, \cdots, s_{\numtime}$ from $[0,T]$. Then each gradient can be computed from  $\numtime$ expectation values.  
\item Alternatively, one may directly sample $s\in \text{Uniform}(0,T)$ which provides a stochastic approximation of the gradient, as is also used in \cite{leng2022differentiable}.
\end{itemize}
A complexity of resource estimate for both cases will be analyzed in \secref{sec::cost}.

\begin{corollary} \label{eq:cor1}
  Suppose $a_T(\theta)  \propto \frac{\delta \loss}{\delta \rho}\vert_{\rho=\rho(T,\theta)}$ is itself not a density matrix, but can be written as a known linear combination of density matrices $\big\{a_{T,j}(\theta)\big\}_{j=1}^J$, so $a_T(\theta)=\sum_{j=1}^J c_j a_{T,j}(\theta)$, where each $a_{T,j}(\theta)$ is proportional to a density matrix and $c_j \in \mathbb{R}$. This can happen for example when $\loss_{\theta} = \sum_{j=1}^{J} c_j \loss_{j, \theta}$; see e.g., Eq.~\eqref{eqn::deri_observable} in \secref{sec:appclosed} for a concrete example. Then
\begin{align*}
\begin{aligned}
& \frac{\partial \loss_{\theta}}{\partial \theta_m} =2 \sum_{j=1}^{J}  c_j A_{T,j}(\theta)\int_0^T  \tr\big(\obs_m(s, \theta) \eta_j(s,\theta)\big)\ ds, \qquad \eta_j(s,\theta) = \mathcal{V}_\theta(s) \eta_j(0,\theta) \mathcal{V}_{\theta}(s)^\dagger,\\
& \eta_j(0, \theta) = \ketbra{+} \otimes a_{T,j}(\theta) \otimes \rho_0, \qquad a_{T,j}(\theta)=\frac{1}{A_{T,j}(\theta)}\frac{\delta \mathcal{L}_{j,\theta}}{\delta\rho}\Big|_{\rho(T, \theta)}, \qquad A_{T, j}(\theta)=\tr\left(\frac{\delta \mathcal{L}_{j,\theta}}{\delta\rho}\Big|_{\rho(T, \theta)} \right).
\end{aligned}
\end{align*}
\end{corollary}

When the number of terms $J$ is moderately small, one could simply estimate observables for each term and then sum them up classically. When the size of $J$ is extremely large, SGD or mini-batch SGD are efficient tools to reduce the cost. When $c_j$ has certain multi-scale behavior, we can further consider applying the importance sampling adaptively. Examining how more advanced sampling techniques can be employed can be highly dependent on the problem and we leave this for later work on specific applications. \\

It is important to emphasise that the form of the loss function $\mathcal{L}_{\theta}$ \textit{only appears in the initial state preparation} of the algorithm -- preparation of the state $\eta(0, \theta)$. Apart from this preparation step, the algorithm is itself agnostic to the form of the loss function, and no change to the form of the quantum circuits for the gradient estimation is necessary even if the loss function changes. Thus, once the quantum circuit itself is prepared and the Hamiltonian form is decided upon, then the circuit can be reused for any other loss function, as long as $\eta(0, \theta)$ can be prepared as an initial input. 

\subsection{Quantum gradient algorithm: no time discretisation}
\label{alg::no_discretization}

While the formulation in Theorem~\ref{thm::main} requires a discretisation of $s$ to compute the final integral in Eq.~\eqref{eq:mixedgradient2}, it is also possible to devise a fully analog algorithm where no such discretisation is necessary, and the gradient of the loss function with each tuneable parameter is computed with a single expectation value. To achieve this, we introduce an augmented quantum state
\begin{align} \label{eq:etahat}
\widehat{\eta}(\theta) = \int g(s) \ketbra{s} \otimes \eta(s, \theta)\ ds, \qquad \int\ g(s) ds=1,
\end{align}
where $\int g(s) \ketbra{s} ds$ is a normalised (continuous-variable) quantum state in $L^2(\Real)$, and $\eta(s,\theta)$ is given in Eq.~\eqref{eq:psisstate2}. Ideally, we would like the density function $g(s)$ to be close to $g_{\text{top}}(s)$ which is defined as follows:
\begin{align}
\label{eqn::g_top}
g_{\text{top}}(s) = \left\{\begin{aligned}
1/T, \qquad s\in [0,T];\\
0, \qquad \text{otherwise}.
\end{aligned}\right.
\end{align} 
In a real physical implementation, this exact top-hat function would need to be approximated by a smooth function $g(s)$. We first look at the case where the (closed) QNODE or QNPDE is generated by a time-independent Hamiltonian $\bar{H}(\theta)$. Then we have the following theorem where the gradient of the loss function with respect to each $\theta_m$ can be estimated using a single expectation value.

\begin{theorem}[Quantum gradient algorithm: fully continuous time]
\label{thm5}
Consider the Hamiltonian dynamics of quantum state $\rho_0$ evolving under a parametric time-independent Hamiltonian $\bar{H}(\theta)$ and  an accompanying loss function $\mathcal{L}_{\theta} = \mathcal{L}\big(\rho(T,\theta)\big)$ at a terminal time $T$. Assume access to the quantum state $a_T(\theta) \propto \delta \mathcal{L}_{\theta}/\delta \rho \vert_{\rho=\rho(T,\theta)}$. Then the gradient of the loss function with respect to $\theta_m$ can be estimated in terms of a single expectation value with respect to the state $\widehat{\eta}(\theta)$. The density function $g(s)$ is expected to be close to the normalised top-hat function $g_{top}(s)$. 
\begin{itemize}
\item The gradient of the loss function $\mathcal{L}_{\theta}$ with respect to $\theta_m$ can also be estimated using a single expectation value with respect to the state $\widehat{\eta}$.
The above results easily imply that when $g(s)=g_{top}(s)$,
\begin{align} \label{eq:lem2exact}
   & \frac{\partial \mathcal{L}_{\theta}}{\partial \theta_m}= 2 A_T(\theta) T\ \tr\big((\unit_s\otimes \obs_m(\theta))\ \widehat{\eta}(\theta) \big), \nonumber \\
    & \obs_m(\theta)= \sigma_y \otimes \frac{\partial \bar{H}(\theta)}{\partial \theta_m} \otimes \unit, \qquad m=1, \cdots, M,
\end{align}
where the state
\begin{align}
\begin{aligned}
\widehat{\eta} (\theta) &= \big(\unit_s\otimes \controlswap \big) W \big(\int g(s) \ketbra{s} \otimes  \eta(0, \theta)\big) W^\dagger \big(\unit_s\otimes \controlswap \big), \\ 
 \eta(0, \theta) &= |+\rangle \langle +| \otimes a_T(\theta) \otimes \rho_0,\\
W &= \big(\unit_{\hbt_{\aux}} \otimes e^{-i \bar{H}(\theta)\otimes  \hat{s}}\otimes \unit_{\hbt_{\adj}}\big) \big(\unit_{\hbt_{\aux}} \otimes e^{i \bar{H}(\theta) \otimes (T\mathbf{1}_s-\hat{s})}\otimes \unit_{\hbt}\big).
\end{aligned}
\end{align}
\item When $g(s) \neq g_{top}(s)$, we can estimate the gradient of the loss function to precision 
\begin{align} \label{eq:gradient3}
\begin{aligned}
    & \left| \frac{1}{2T A_T(\theta) }\frac{\partial \mathcal{L}(\rho(T))}{\partial \theta_m}- 
    \tr\big(\unit_s\otimes \obs_m(\theta)\ \widehat{\eta}(\theta) \big)
    \right| \le \norm{\obs_m(\theta)}_{\infty} \norm{g_{top} - g}_{\text{TV}} ,
\end{aligned}
\end{align}
where $\norm{\obs_m(\theta)}_{\infty}$ is the operator norm of $\obs_m(\theta)$, and $\norm{g_{top} - g}_{\text{TV}}$ is the total variation distance of two probability distributions $g_{top}$ and $g$.
\end{itemize}
\end{theorem}
\begin{proof}
See Appendix~\ref{app:cor1}. 
\end{proof}
It is straightforward to generalise this for time-dependent Hamiltonians $H(t, \theta)$ by employing the method described in Section~\ref{sec:appnonauto}. In this case we set the time-independent Hamiltonian $\bar{H}(\theta)=\hat{p}_s \otimes \mathbf{1}+H(\hat{s}, \theta)$ where $H(\hat{s}, \theta)$ is just $H(t, \theta)$ with the replacement $t \rightarrow \hat{s}$. 
Then from Section~\ref{sec:appnonauto}, we see that we can recover the solution $\rho(t, \theta) \approx \tr_s(\sigma(t, \theta))$ where $\sigma(t, \theta)$ evolves unitarily with respect to the time-independent Hamiltonian $\bar{H}(\theta)$. If we rewrite loss functions of the form in Eq.~\eqref{eq:lossdef} in the form suitable for (closed) QNODE/QNPDE (that evolves with a time-independent Hamiltonian), we have 
\begin{align}
    \tr(O \rho(t, \theta)) \approx \tr\big(O \tr_s(\sigma(t, \theta))\big)=\tr\big(M \sigma(t, \theta)\big), \qquad M=O \otimes \mathbf{1}_s.
\end{align}
Then Theorem~\ref{thm5} can be directly applied with the modification of the original $a_T(\theta) \propto O \otimes \mathbf{1}_s$ for this case.

\subsection{Cost for gradient estimation} \label{sec:complexity}
\label{sec::cost}

We have two gradient estimation algorithms, in Theorems~\ref{thm::main} and ~\ref{thm5}. We will discuss their complexity and the situations where our quantum algorithms could be beneficial.\\

\noindent \textit{Discussion on the complexity of Theorem~\ref{thm::main}}:\\

We first consider the formalism in \thmref{thm::main} and recall that QNODE or QNPDE evolves via unitary evolution generated by the Hamiltonian in Eq.~\eqref{eqn::H_ansatz} also listed as follows:
\begin{align*}
H(t, \theta) = \sum_{k=1}^K f_k(t, \theta) H_k, \qquad H_k \neq H_{k'}, \quad  \forall k \neq k'.
\end{align*}
Then we can estimate the gradient using a discretisation in time to estimate the integration over time $s \in [0, T]$
\begin{align}
\label{eqn::discretise}
\begin{aligned}
\frac{\partial \mathcal{L}_{\theta}}{\partial \theta_m} \propto&\ \int_{0}^{T} \tr\Big(\big(\sigma_Y \otimes \sum_{k=1}^K \frac{\partial f_k(t, \theta)}{\partial \theta_m} H_k \otimes \unit\big)\  \eta(s, \theta)\Big)\ ds \\
\approx&\ \sum_{j=1}^{\numtime} \sum_{k=1}^K \Delta s_j\ \tr\Big(\big(\sigma_Y \otimes  \frac{\partial f_k(s_j, \theta)}{\partial \theta_m} H_k \otimes \unit\big)\ \eta(s_j, \theta)\Big)\\
=&\ \sum_{j=1}^{\numtime} \sum_{k=1}^K \Delta s_j  \frac{\partial f_k(s_j, \theta)}{\partial \theta_m} \tr\big(\sigma_Y \otimes H_k \otimes \unit\ \eta(s_j, \theta)\big),
\end{aligned}
\end{align}
where we have discretised the time at grid points $s_1, s_2, \cdots, s_\numtime$ and $\Delta s_j$ is the weight from the quadrature scheme.  Now each of the $\numtime K$ terms $\partial f_k(s_j, \theta)/\partial \theta_m$ can be computed using purely classical processing, so this does incur any cost for the quantum simulator. Thus, even if there are $M$ different tuneable parameters, we only need quantum algorithms to compute $\order{\numtime K}$ different expectation values $\tr\big(\sigma_Y \otimes H_k \otimes \unit\ \eta(s_j, \theta)\big)$ with possible repetitive  measurements. {We can also estimate the time integral via Monte Carlo scheme by picking grid points $s_j$ randomly.}\\

{In the following Table~\ref{table::2}, we summarize quantum resources required for estimating the gradient measured in both $L_\infty$ norm and $L_2$ norm. The estimate in $L_\infty$ norm is provided in a probabilistic argument and the one for $L_2$ is stated as the expected error. It is clear that for $L_\infty$ norm, the gradient estimate does not depend on the number of parameters $M$ directly, and the midpoint scheme may provide a slightly better scaling. In general, for the machine learning training part, the key scaling with respect to $M$, $\Delta$, $T$ and $K$ are typically the same for both midpoint scheme and Monte Carlo scheme in handing the time integral, which is very likely stem from the intrinsic nature of Monte Carlo type estimate for quantum measurements.} Though some preliminary estimate were discussed \cite{leng2022differentiable}, to the best of our knowledge, no previous literature has rigorously estimated and also proved the complexity scaling in quantum gradient estimation for continuous-time quantum processes. \\

\begin{table}[h!]
\caption{Summary of the quantum measurements required for various error types ($L_\infty$ and $L_2$) as well as for the optimization. To distinguish gradient estimate and training part, we use $\delta$ as error for gradient estimate, and $\Delta$ for  training. The bound may not be tight for some parameters; the overall scaling  with respect to our main focus (namely, $M$ and $T$) should be expected. $\Lipschitzmax$ is the Lipschitz constant for the loss, and $\mu$ is the \PL{} constant. 
}
\vspace{\baselineskip}
\begin{NiceTabular}[width=15cm]{X[3,l,m]X[7,l,m]X[5,l,m]}[]
\hline 
{\bf Algorithm } & {\bf Midpoint in time \eqref{eqn::approxiamte_Im}}  & \medskip {\bf Monte Carlo in time \eqref{eqn::mc_time}}  \medskip \\
\hline
{\bf Assumptions} & Assumptions~\ref{assume::norm_Hk_2}, \ref{assume::PL}, \ref{assume::AT}  & \medskip Assumptions~\ref{assume::norm_Hk}, \ref{assume::PL}, \ref{assume::AT}  \medskip \\
\hline
{\bf ${\bf L_\infty}$ error} & 
\medskip $\max\Big\{ \ordertilde{(\frac{T K}{\delta})^2}, \order{T^{\frac{3}{2}} K^{\frac{5}{2}} \delta^{-\frac{1}{2}}} \Big\}$ \par \smallskip
(Proposition~\ref{prop::midpoint_inf}) \medskip
& 
\medskip $\ordertilde{(\frac{T K}{\delta})^2 K}$ \par  \smallskip (Proposition~\ref{prop::mc_inf})\medskip
\\
\hline
{\bf ${\bf L_2}$ error} & \medskip $\max\big\{\order{(\frac{T K}{\delta})^2 M}, \order{M^{1/4} K^{5/2} T^{3/2}\delta^{-1/2}}\big\}
$ \par  \smallskip (Proposition~\ref{prop::midpoint_L2}) \medskip &  \medskip $\order{(\frac{T K}{\delta})^2 M K}
$ \par  \smallskip (Proposition~\ref{prop::mc_L2}) \medskip  \\
\hline
{\bf Training} \par (Limit $M\gg 1$,\par\ \  $\Delta \ll 1$) & \medskip $\widetilde{\mathcal{O}}\bigg(\frac{(T K)^2 \Lipschitzmax M}{\mu^2 \Delta} \bigg)$ \par  \smallskip
(Proposition~\ref{prop::ml::midpoint})
\medskip
& \medskip $\ordertilde{(\frac{T K}{\mu})^2 \frac{\Lipschitzmax M K}{\Delta}}$ \par   \smallskip (Proposition~\ref{prop::mc_train}) \medskip 
 \\
 \hline
\end{NiceTabular}
\label{table::2}
\end{table}

The linear dependence on $M$ for training cost is perhaps inevitable, as it is a typical situation for stochastic gradient descent and $M$ indicates the dimension of the optimization problem.
For the detailed statement, please refer to Appendix~\ref{sec::complexity}. Given the current analysis, the midpoint scheme appears to have slightly better scaling; this probably comes from a slightly stronger assumption on the regularity of $f_k$ for midpoint scheme. It remains to investigate whether the above analysis is tight for some parameters e.g., the \PL constant $\mu$, which will be left as continuing work. In Appendix, we include the discussion of the scaling with more parameters; please refer to the detailed explanation therein.
Note that the scaling can be improved if cleverer measurement schemes are applied, for instance in using shadow tomography \cite{huang_information-theoretic_2021,huang2022learning} and related methods. Using these methods, the estimation of $K$ expectation values would only require $\log(K)$ measurements, thus, in principle, we can expect an improvement over the scaling of $K$.
\\ 

The above scaling suggests the potential efficacy of quantum neural ODEs/PDEs, at least for certain families of quantum machine learning problems. 
Importantly, this cost, regardless of measurement scheme, \textit{does not} a priori  depend on the total number of tuneable parameters $M$ (or at most logarithmic dependence on $M$, which is generally inessential). 
The state preparation for $\eta(s,\theta)$ depends on the given QNODE and QNPDE: the analog quantum simulation, the unitary gates $U_{\theta}(s, s')$ and the controlled-SWAP operator. It is possible to design the analog quantum simulation for the problem at hand. For instance the simulation of specific ODEs/PDEs in Section~\ref{sec:appodepde}, for example see \cite{analogPDE, JC}. 
If we resort to digital quantum simulation, it is already clear that for whenever simulating  $U_{\theta}(s, s')$ digitally is advantageous over the classical simulation of $U_{\theta}(s,s')$, the quantum algorithm for gradient estimation will be more efficient than its classical counterpart. \\

We can already note a few cases where QNODE/QNPDE could be beneficial, and leave a more systematic study for future work on more specialised applications:
\begin{itemize}

\item This formulation appears to be well-suited to learning the time-dependence of Hamiltonians for some small $K$ and moderately small $T$, since the quantum resource cost does not depend on the complexity of the parameterisation of each $f_k(t, \theta)$. This is applicable, for example, to learning non-autonomous ODEs/PDEs. 

\item If the Hamiltonian model is the quantum Ising model, then $K \sim \text{poly}(n)$ where $n$ is the number of qubits. This scaling is desirable when learning a large-scale quantum system (e.g., $n \sim 10^2$) using our algorithm, compared with classical machine learning. For the short-time Hamiltonian learning problem $T=\order{1}$, the above scaling is polynomial in the number of qubits $\order{\text{poly}(n)}$, which suggests that the gradient estimate could be effective.

\end{itemize}

Although heuristic, this discussion already outlines the promising scaling of our approach; a more comprehensive mathematical analysis is left for future work.\\

\noindent \textit{Discussion on the complexity of Theorem~\ref{thm5}}:\\

In using the fully continuous-time algorithm in Theorem~\ref{thm5}, we do not need to discretise in time at all. In that case, we only need to consider the sum of $K$ different expectation values.  Thus, not including the cost of $\widehat{\eta}(\theta)$ preparation (the cost here is not significantly different from the cost in preparing $\eta(s, \theta)$ except with the addition of a qumode), we only need, in the naive case $\order{ \text{poly}(K) \delta^{-2}}$
measurements, and this can in principle be improved to $ \order{ \log(K) \delta^{-2}}$ with shadow tomography. 
For non-universal Hamiltonians, this $K$ can be very small. \\

In principle, even for \textit{near universal} Hamiltonians, $K$ may not need to be very large, but $f_k(t, \theta)$ may need to contain \textit{high frequency} terms in time, since we would be rapidly alternating between the application of different Hamiltonians in time, for instance in a similar setting as quantum approximate optimisation algorithm (QAOA), where we alternate between the application of the cost and the mixer Hamiltonians \cite{qaoa}. \\

We note that while these algorithms can offer up to exponential cost efficiency by delegating the relevant classical dynamics simulation to the corresponding quantum dynamics simulation and simulating this on a quantum device (especially useful for very high-dimensional problems), we do need to prepare many copies of either $\eta(s, \theta)$ using the algorithm in Theorem~\ref{thm::main}. This is purely due to the quantum non-cloning theorem.  We need to prepare $\eta(s, \theta)$ for each $s=s_1, \cdots, s_{\numtime}$ using at least $\numtime$ copies of $\eta(s, \theta)$ since we cannot keep $\eta(s, \theta)$ while also using it to creare $\eta(s'>s, \theta)$. On the other hand, by using the algorithm in Theorem~\ref{thm5}, we do not have the same quantum non-cloning theorem issue and we just need to prepare a time-independent state $\widehat{\eta}(\theta)$. However, this algorithm requires a continuous-variable treatment of the extra $s$-mode and it remains to be seen what kinds of implementations are possible in this setting.\\

\subsection{On the local energy landscape}
\label{subsec::landscape}

The shape of the energy landscape plays a crucial role in the efficiency of training quantum neural networks \cite{ge_optimization_2022}.
For analog quantum machine learning, these landscape theories merit reexamination, because QNODE/QNPDE offers a continuous perspective for parameterizing the dynamics, in contrast to the widely studied (discrete) QNNs.\\

Consider the loss function for the time-dependent Hamiltonian $H(t, \theta)=\sum_{k=1}^{\numop} f_k(t, \theta) H_k$ from Eq.~\eqref{eqn::H_ansatz},
where, in principle, each $f_k$ can be parameterized in various ways (e.g., polynomials or classical neural networks).
Let the exact function be $f_k^\star$, and define the optimal Hamiltonian as
$H^\star(t)=\sum_{k=1}^{K} f_k^\star(t) H_k$.
For now, we ignore the \emph{generalization error} arising from any mismatch between the chosen parametric family $\{f_k\}$ and the exact function $f_k^\star$:

\begin{assumption}
\label{assume::parametric_family}
 We assume that there is a unique $\theta^\star$ such that $f_k(t, \theta^\star) = f_k^\star(t)$.
\end{assumption}
This assumption could also be understood as a restriction of parameter near the exact $\theta^\star$, and we study the local energy landscape rather than the global landscape.\\

In the following section, we shall explore sufficient conditions to ensure that the \PL{} condition to hold  locally. 
In case the \PL{} condition does hold, we can ensure an efficient training of this machine learning problem using QNODE/QNPDE as outlined in Table~\ref{table::2}. 
Whether a problem satisfying local \PL{} condition in general must be problem dependent, and for the time being, we shall consider the state learning problem from Eq.~\eqref{eqn::loss_generative} with the expression of the loss function as follows:
\begin{align}
\label{eqn::loss_for_landscape}
\begin{aligned}
\mathcal{L}_{\theta}(\rho_0, T) &= 1 - \tr\big(\sigma(T) \rho(T, \theta)\big),
\end{aligned}
\end{align}
where $\rho_0$ and $T$ are temporarily fixed and we consider $\rho_0$ as a pure state only. $\sigma(T)$ is obtained by evolving $\rho_0$ under the exact Hamiltonian $H^\star$, and $\rho(T,\theta)$ is obtained by evolving the same $\rho_0$ under parameterized $H(t, \theta)$. 
Under Assumption~\ref{assume::parametric_family}, it is clear that the global minimum is always zero, achievable when $\theta = \theta^\star$. \\

\begin{lemma}[Hessian for QNODEs for state learning]
\label{lemma::Hessian_QNODE}
The Hessian matrix for QNODEs for the problem in Eq.~\eqref{eqn::loss_for_landscape} with fixed pure state $\rho_0$ and $T$, at the exact parameter $\theta^\star$ is given as follows, 
\begin{align*}
\nabla^2 \loss_\theta\rvert_{\theta^\star} =&\ 2 \sum_{k=1}^{\numop}\sum_{k'=1}^{\numop} \int_{0}^T \int_{0}^T \nabla_{\theta^{\star}} f_k(t) \nabla_{\theta^{\star}} f_{k'}^\top(s) G_{k,k'}(t,s)\ ds dt \\
=&\ 2  \int_{0}^T \int_{0}^T \nabla_{\theta^{\star}} \vect{f}(t) G(t,s)  \nabla_{\theta^{\star}} \vect{f}(s)^\top\ ds dt,
\end{align*}
where the (quantum) correlation matrix is defined as 
\begin{align*}
G_{k,k'} (t,s)\ :=&\ \bra{\psi(t)} H_{k} U_{\theta^\star}(s,t) H_{k'} \ket{\psi(s)} - \langle H_k \rangle_{\psi(t)} \langle H_{k'}\rangle_{\psi(s)}\\
=&\ \bra{\psi(0)} \tilde{H}_k(t) \tilde{H}_{k'}(s) \ket{\psi(0)}  - \langle \tilde{H}_k(t) \rangle_{\psi(0)} \langle \tilde{H}_{k'}(s) \rangle_{\psi(0)},
\end{align*}
and $\vect{f} = \begin{bmatrix} f_1 & f_2 & \cdots & f_K \end{bmatrix}$, and $\ket{\psi(s)}\bra{\psi(s)} = \rho(s, \theta^\star)$ is the optimal trajectory, and $\tilde{H}_k(t) := U(t,0) H_k U(0,t)$ is the Heisenberg picture of $H_k$.  
\end{lemma}
The above form of $G$ is essentially the quantum correlation, and existing analytical and numerical tools related to this quantity may provide help to estimate the Hessian of the energy landscape near the global minimum.\\

\begin{definition}
Let us define an integral operator $\mathcal{G}$ for any function $\phi: [0,T]\to \Real^\numop$ via 
\begin{align}
\label{eqn::G}
\mathcal{G} \phi(t) = \int_{0}^T G(t, s) \phi(s)\ ds,
\end{align}
and denote the smallest eigenvalue of $\mathcal{G}$ as $\muG$.
\end{definition}

Then acting on the test vector $\nu\in \Real^{M}$, one has
\begin{align}
\label{eqn::Hessian_acting_nu}
\begin{aligned}
\nu^\top \nabla^2 \loss_\theta\rvert_{\theta^\star} \nu =&\ 2 \int_{0}^T \int_{0}^T \nu^\top \nabla_{\theta^\star} \vect{f}(t) G(t,s) \nabla_{\theta^\star} \vect{f}(s)^\top \nu\ ds\ dt \\
\ge&\ 2\mu_\mathcal{G} \int_{0}^T \nu^\top \nabla_{\theta^\star} \vect{f}(t) \nabla_{\theta^\star} \vect{f}(t)^\top \nu\ dt.
\end{aligned}
\end{align}
With this tool, for some parametric family (or say architectures), we can estimate the Hessian using the information of quantum correlation function.

\begin{example}[Time-independent Hamiltonians] Suppose that the Hamiltonian is time-independent, namely, we parameterize $f_k(\theta,t) = \theta_k$.
By Eq.~\eqref{eqn::Hessian_acting_nu}, since the term $\nabla_{\theta} \vect{f} \equiv \unit$, the Hessian lower bound is $2 \muG T \unit$.
\end{example}

\begin{example}[Quantum ResNet, and QNN] Let us consider the family of control schedule $f_k$, which is a piecewise linear function, 
\begin{align*}
f_k(\theta, t) = \theta_{k,j}, \qquad \text{ if } t_j \le t \le t_{j+1}, 
\end{align*}
where $t_j = j \Delta t$ are uniformly chosen time grids.
Then by Eq.~\eqref{eqn::Hessian_acting_nu},
we can straightforwardly show that
\begin{align*}
\nu^\top \nabla^2 \loss_\theta\rvert_{\theta^\star} \nu\ \ge&\ 2\mu_\mathcal{G} \sum_{j}\int_{t_j}^{t_{j+1}} \nu^\top \nabla_{\theta^\star} \vect{f}(t) \nabla_{\theta^\star} \vect{f}(t)^\top \nu\ dt \\
=&\ 2\mu_\mathcal{G} \Delta t \sum_{\ell} \nu_\ell^2\ =\ 2\mu_G \Delta t \norm{\nu}^2.
\end{align*}
namely, the local Hessian lower bound is $2\muG \Delta t \unit$ (regardless of how large the $\Delta t$ is chosen). \\

This result reveals the difficulty of training quantum ResNet with very small time steps.
In reality, if the optimal schedule is indeed smooth, then one may consider approximating it using piecewise continuous step function with small $\Delta t$. However, this higher representability may come with the burden of more parameters, and more importantly the hardness of training, as suggested in the above example. One may instead consider Fourier series.

\end{example}

\begin{example}[Quantum Fourier Net] Suppose that we consider control schedule $f_k$ via Fourier series, or say sinusoidal functions, e.g., let us take 
\begin{align*}
f_k(\theta, t) = \sum_{j=0}^{\numtime} \theta_{k,j} \cos\big(2 \pi j t/T\big),
\end{align*}
then by Eq.~\eqref{eqn::Hessian_acting_nu}, one can easily compute that
\begin{align*}
\nu^\top \nabla^2 \loss_\theta\rvert_{\theta^\star} \nu \ge &\ 2\muG \int_{0}^T \nu^\top \nabla_{\theta^\star} \vect{f}(t) \nabla_{\theta^\star} \vect{f}(t)^\top \nu\ dt \\
=&\  2 \muG \int_{0}^T \sum_k \big(\sum_{j} \nu_{j,k} \cos(2 \pi j t/T)\big)^2 dt \\
=&\ 2 \muG \sum_k \sum_{j} \nu_{j,k}^2 \frac{T}{2}  \\
=&\ \muG T \norm{\nu}^2,
\end{align*}
so that the Hessian lower bound is $\muG T \unit$. In this case, the total number of parameters is $K (\numtime+1)$ and we denote the element $\nu_{j,k}$ as the adjoint of the parameter $\theta_{k,j}$ inside the test vector $\nu$, which takes the form $\nu^\top = \begin{bmatrix}\nu_{0,1} & \nu_{1,1} & \cdots & \nu_{\numtime,1} & \nu_{0,2} & \nu_{1,2} & \cdots & \nu_{\numtime,2} & \cdots \end{bmatrix}$.
\end{example}

\section{Numerical examples}
\label{sec::example}

We have presented several types of loss functions suitable for Hamiltonian learning in \secref{sec:appclosed}, and in this section, we will provide numerical demonstrations and validations. 
We mainly consider the following examples: (a) the hydrogen molecule; (b) the quantum Ising model.
The Hamiltonians for these two examples are described in Appendix~\ref{app::Hamil}.
Details of the training hyper-parameters can be found in Appendix~\ref{app::trainning_details}.
Note that here we provide simulation results using Theorem~\ref{thm::main}. At this moment, we have not done the same thing for Theorem~\ref{thm5} since we do not currently have a fully continuous-time general purpose emulation device particularly for learning ODEs and PDEs. Demonstrations using Theorem~\ref{thm5} on fully analog quantum devices is left to future work. 
All results reported below are based on classical emulation using PennyLane. In the following, we consider the noise from quantum measurements to better reflect the practice; \enquote{shots} means the number of quantum measurements performed each time; \enquote{shots = inf} means exact quantum measurements without fluctuation. We will use a small batch of data for training. We will report the test error, namely the error measured by the fidelity between the exact quantum state and the approximated quantum state; for Hamiltonian learning problems, the reported error has been averaged over different random initial states; see Eq.~\eqref{eqn::test_error} in Appendix~\ref{app::trainning_details} for full details.

\subsection{Learning state preparation}
\label{eg::qsp}

A key example is the quantum state preparation for a target state $\sigma$, where we choose the loss function as follows:
\begin{align*}
    \mathcal{L}(\rho(T, \theta))=1-\tr\big(\sigma \rho(T,\theta)\big) \equiv \tr\big( (\unit - \sigma) \rho(T,\theta)\big).
\end{align*}
See also Section~\ref{sec:appclosed}. To demonstrate \thmref{thm::main}, we consider a toy example: a single qubit with target state $\sigma = \ketbra{+} = a_T(\theta)$, and it is easy to see that $A_T(\theta) = -1$. We begin with $\rho_0=|0\rangle \langle 0|$, and use the following ansatz: $H(\theta) = \theta_1 \sigma_X + \theta_2 \sigma_Y + \theta_3 \sigma_Z$, and set the time period $T = 1$. The time integral is discretised using the trapezoidal rule with $\Delta s = 0.1$, and we take different numbers of shots to estimate the observable for each $s$. One feasible solution is $\theta^\star = [0, \pi/4, 0]$, corresponding to a Pauli-Y rotation, and we compare the optimised $\theta$ to this value $\theta^\star$. Although the solution is not unique (due to the periodicity in $\theta_2$), we initialize $\theta = [0,0,0]$ so that this $\theta^\star$ is the closest global minimizer. As shown in \figref{fig::eg1}, even with a few shots, the error decays exponentially with respect to the training step in the early stage. When $\theta$ is close to $\theta^\star$, quantum measurement noise prevents further improvement in accuracy, as expected. With precise quantum measurement (see the black curve), the error consistently exponentially decays to zero, demonstrating good convergence for this example using the above loss function and our quantum algorithms. This example clearly demonstrates the feasibility and promise of our approach.

\begin{figure}[h!]
\includegraphics[width=0.85\textwidth]{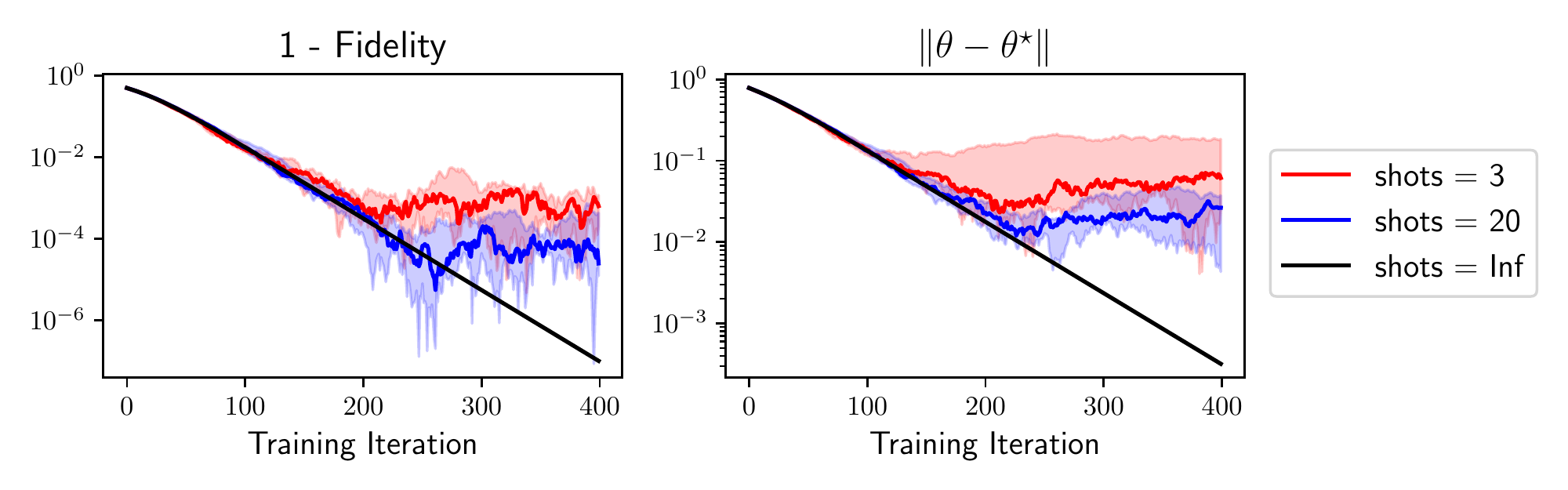}
\caption{The error $1-\abs{\bra{+} U(T,\theta) \ket{0}}^2$ versus training iterations for different numbers of shots. Shots refer to the number of repeated measurements to estimate the observables for each $s$ in Theorem~\ref{thm::main}. Each curve represents the median error from five independent experiments; the shaded region (with the same color) indicates the corresponding minimum and maximum error during training. For the case of exact quantum measurement (black curve), there is no fluctuation in optimisation.}
\label{fig::eg1}
\end{figure}

\subsection{Hamiltonian learning using quantum state}
\label{eg::hl}

This problem was discussed in  Section~\ref{sec:appclosed}. Its  goal is to learn the unknown parameters of a Hamiltonian when we are only given the known input and output states of the dynamics generated by this Hamiltonian. To probe the system, we prepare a set of initial pure states $\{\rho_j\}_{j=1}^{\sample}$ at time zero, which are inputs into the black box. The black-box system returns the corresponding quantum states at various times $T_j$, denoted as $\big\{\sigma_j(T_j) \big\}_{j=1}^{\sample}$. We can also input these states into a unitary system generated by a parameterised Hamiltonian $H(t, \theta)$, with outputs $\big\{\rho_j(T_j, \theta)\big\}_{j=1}^{\sample}$. Our aim is to find a parameterised Hamiltonian $H(t, \theta)$ that best approximates the corresponding black-box unitary evolution. The loss function is naturally defined as in Eq.~\eqref{eqn::loss_eg2}, which is also copied here 
$\loss_{\theta} = \frac{1}{\sample} \sum_{j=1}^{\sample} \Big(1 - \tr\big(\sigma_j(T_j) \rho_j(T_j,\theta)\big)\Big).$\\

Simulation results for time-independent cases are shown in \figref{fig::eg2}. We use the trapezoidal rule to discretise the time integral and take different numbers of shots to estimate the observable for each $s$, obtaining the gradient estimate in Eq.~\eqref{eq:mixedgradient2}.
During training, we use only one sample $\rho_i$ per parameter update, so each iteration reduces to a single unknown state learning task as in \secref{eg::qsp}. This explains the noisy fluctuations in optimisation, even for infinite shots (i.e., exact quantum measurements) in the black curve.
For the hydrogen molecule and Ising model Hamiltonians, the training is very stable even when using one sample per parameter update, as long as quantum measurement accuracy is sufficient. Even with only $10$ shots per quantum estimate, the loss error decays to around $10^{-3}$ in terms of fidelity measurements. Notably, as the number of sites in the quantum Ising model increases, the decay rate of the loss remains nearly unchanged, even for larger systems, provided that one is given precise quantum measurements. With noisy quantum estimate, the training performance is only slightly worsen as the dimension increases. This indicates that the training problem remains efficient, provided that trustful and certified quantum hardware is available.\\

Moreover, we have considered learning a time-dependent Ising model whose Hamiltonian and model parameterization are explained in Appendix~\ref{app::Hamil}. We consider a two-site model for simplicity and the training results are visualized in Figure~\ref{fig::eg2_td}. The result is overall similar to the time-independent case. We also show a typical training result for the time-schedule in Figure~\ref{fig::eg2_td} to validate nice accuracy of the training using our gradient estimates.

\begin{figure}[h!]
\begin{subfigure}[b]{\textwidth}
\includegraphics[width=0.7\textwidth]{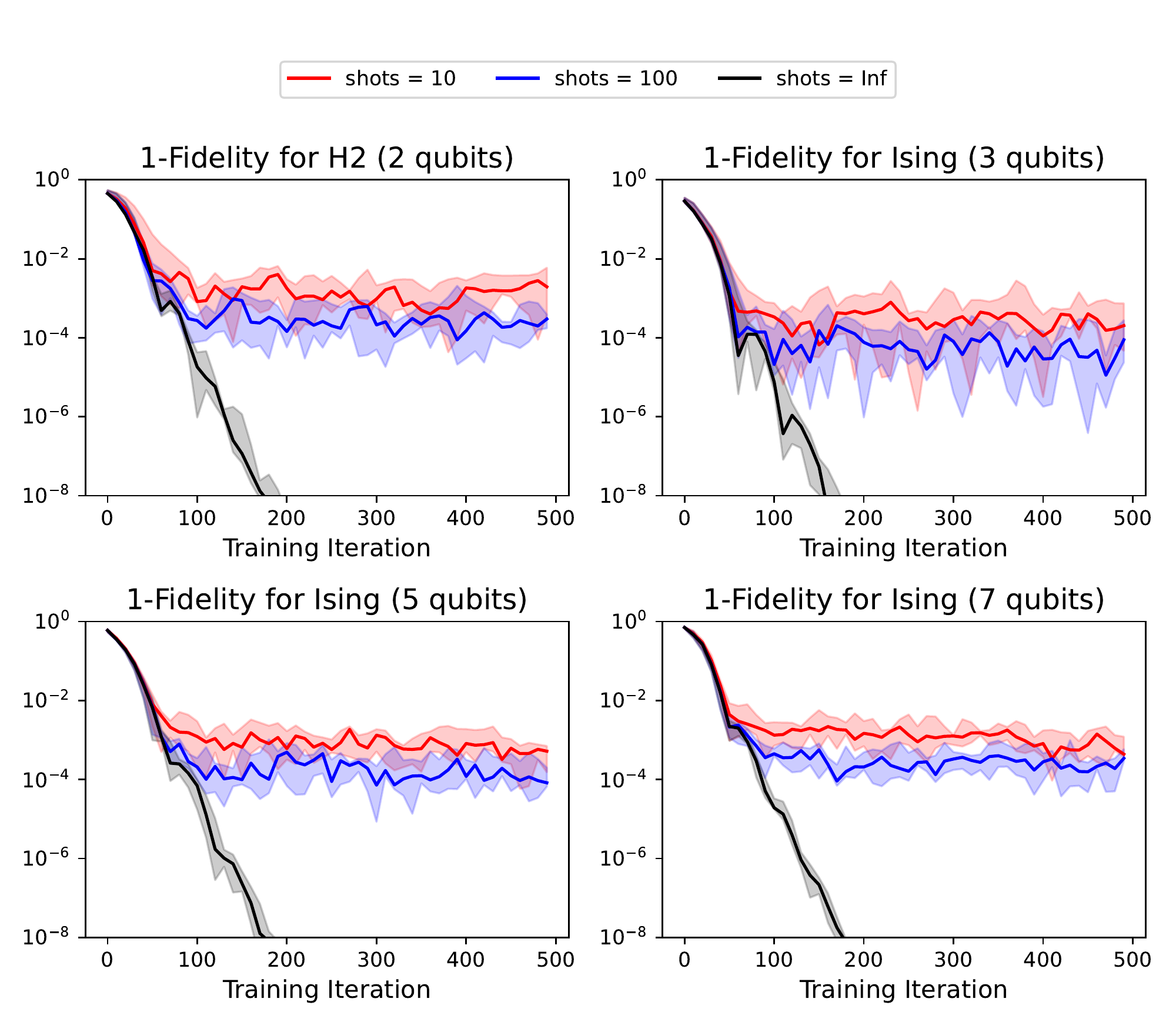}
\end{subfigure}
\caption{Training results for \secref{eg::hl} for different time-independent Hamiltonians using the loss function in Eq.~\eqref{eqn::loss_eg2}. Each curve represents the median error from five independent training runs, and the shaded region indicates the maximum and minimum error. The y-axis range is truncated for better visualization.}
\label{fig::eg2}
\end{figure}

\begin{figure}[h!]
\begin{subfigure}[b]{0.4\textwidth}
\includegraphics[width=\textwidth]{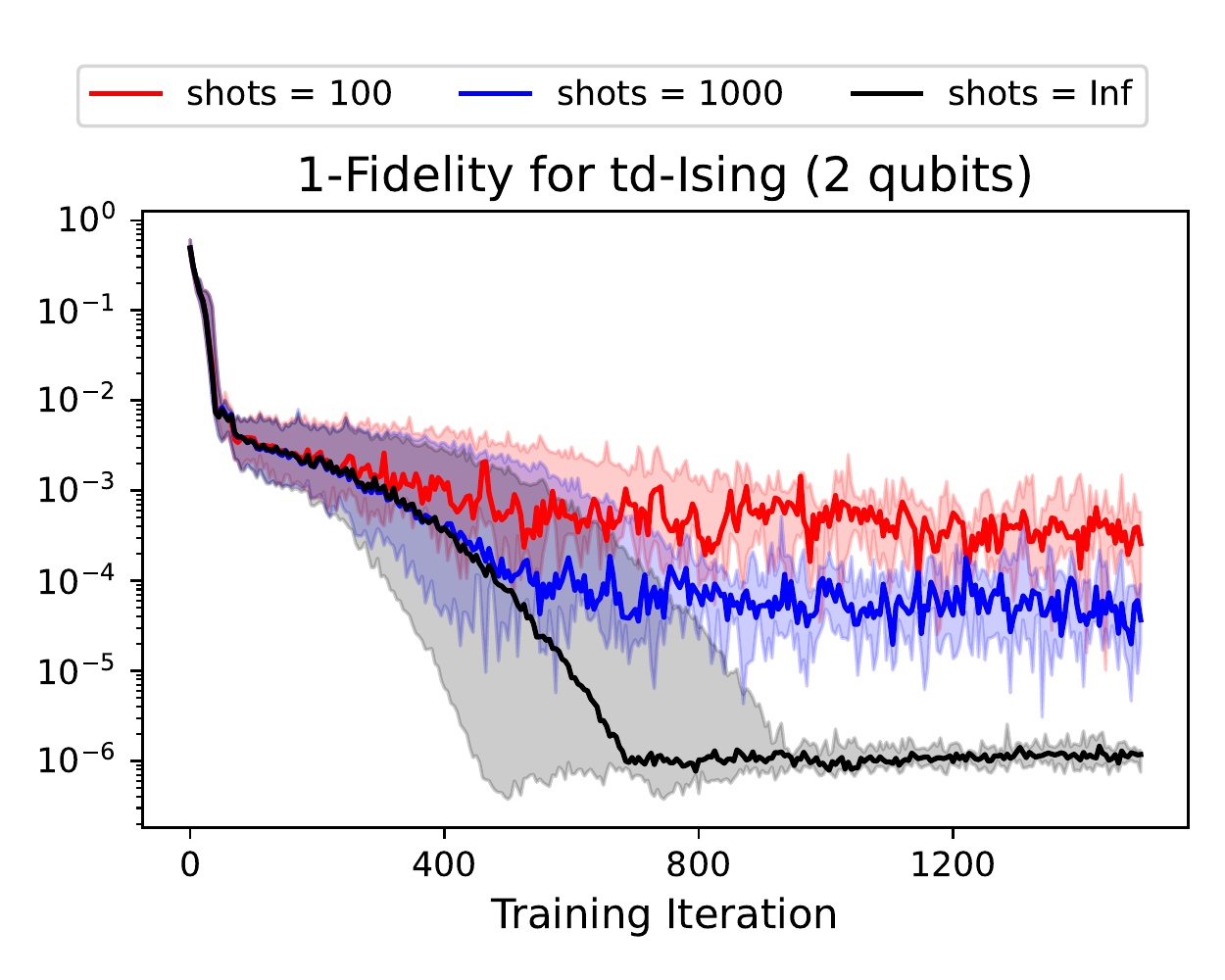}
\caption{Test error during training}
\end{subfigure}
~
\begin{subfigure}[b]{0.55\textwidth}
\includegraphics[width=\textwidth]{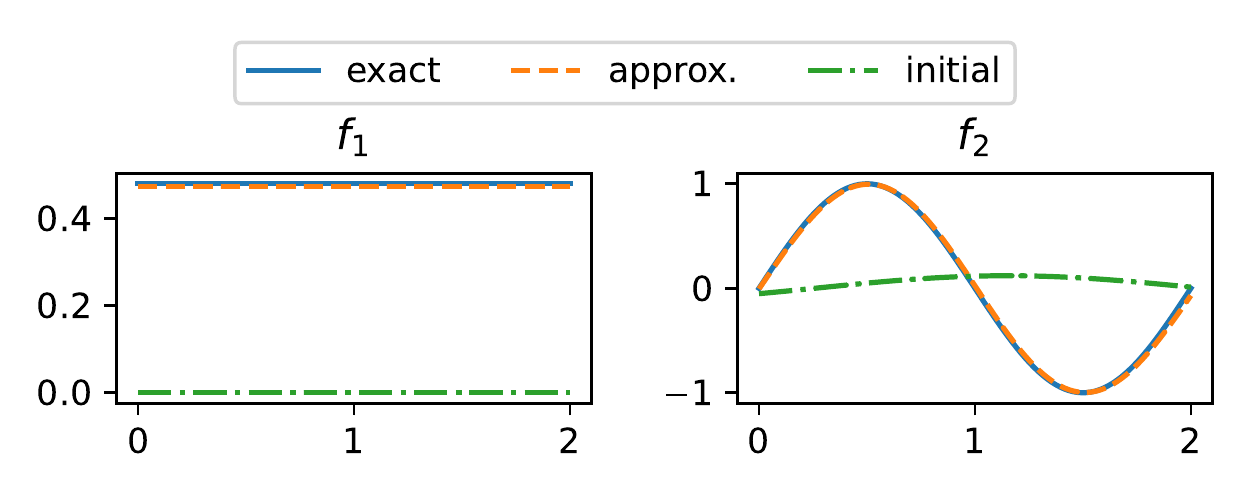}
\caption{A typical training result of time-schedule}
\end{subfigure}

\caption{Training results for \secref{eg::hl} for two-site time-dependent quantum Ising chain (td-Ising) using the loss function in Eq.~\eqref{eqn::loss_eg2}. (a) Each curve represents the median error from five independent training runs, and the shaded region indicates the maximum and minimum error. (b) We show a typical training result of time-schedule using $10^3$ shots per quantum estimate.}
\label{fig::eg2_td}
\end{figure}

\subsection{Hamiltonian learning using observables}
\label{section::hl_observable}

We consider the same Hamiltonian learning problem as in \secref{eg::hl}, except now we assume access only to the expectation values of observables, not the final quantum state outputs. This is also discussed in Section~\ref{sec:appclosed}. In this case, we minimize
\begin{align*}
\frac{1}{M N} \sum_{i=1}^{M} \sum_{j=1}^{N} \abs{\tr\big(\obs_j \sigma_i(T_i)\big) - \tr\big(\obs_j \rho_i(T_i,\theta)\big)}^2, 
\end{align*}
where $O_j$ are observables. We consider the same examples as in \secref{eg::hl}, except we use observable data for Eq.~\eqref{eqn::loss_observable}, which is also listed above for convenience. 
Due to the relatively higher cost of classically emulating mixed states, we consider smaller systems for this example. However, for quantum analog machines, evolving a mixed state is not fundamentally different from evolving a pure state, which brings potential advantages to run our algorithm directly on quantum devices.
For each parameter update, we randomly choose one-site observables $\big\{\sigma_{X,Y,Z}^{(k)}\big\}_{k=1}^{n}$ (where $n$ is the number of qubits), initialize the state at $ \ketbra{+}^{\otimes n}$ followed by random rotations using Rotation-(X,Y,Z), and also pick a random time $T_i$. 
These random initial sampling of input state $\rho_i$ and $T_i$ can help break the symmetry and overcome the appearance of new local minima when using tomography data as the loss function in Eq.~\eqref{eqn::loss_observable}.
Therefore, there is stochastic fluctuation in the datasets $\big(\rho_i, T_i\big)$, as well as in the random selection of observables. 
Simulation results are shown in \figref{fig::eg3}, demonstrating the effectiveness and efficiency of our approach. Although the loss function in Eq.~\eqref{eqn::loss_observable} is effective, we  empirically note that, using tomography data as the loss function appears comparatively harder to train than using an unknown state.

\begin{figure}[h!]
\begin{subfigure}[b]{0.6\textwidth}
\includegraphics[width=\textwidth]{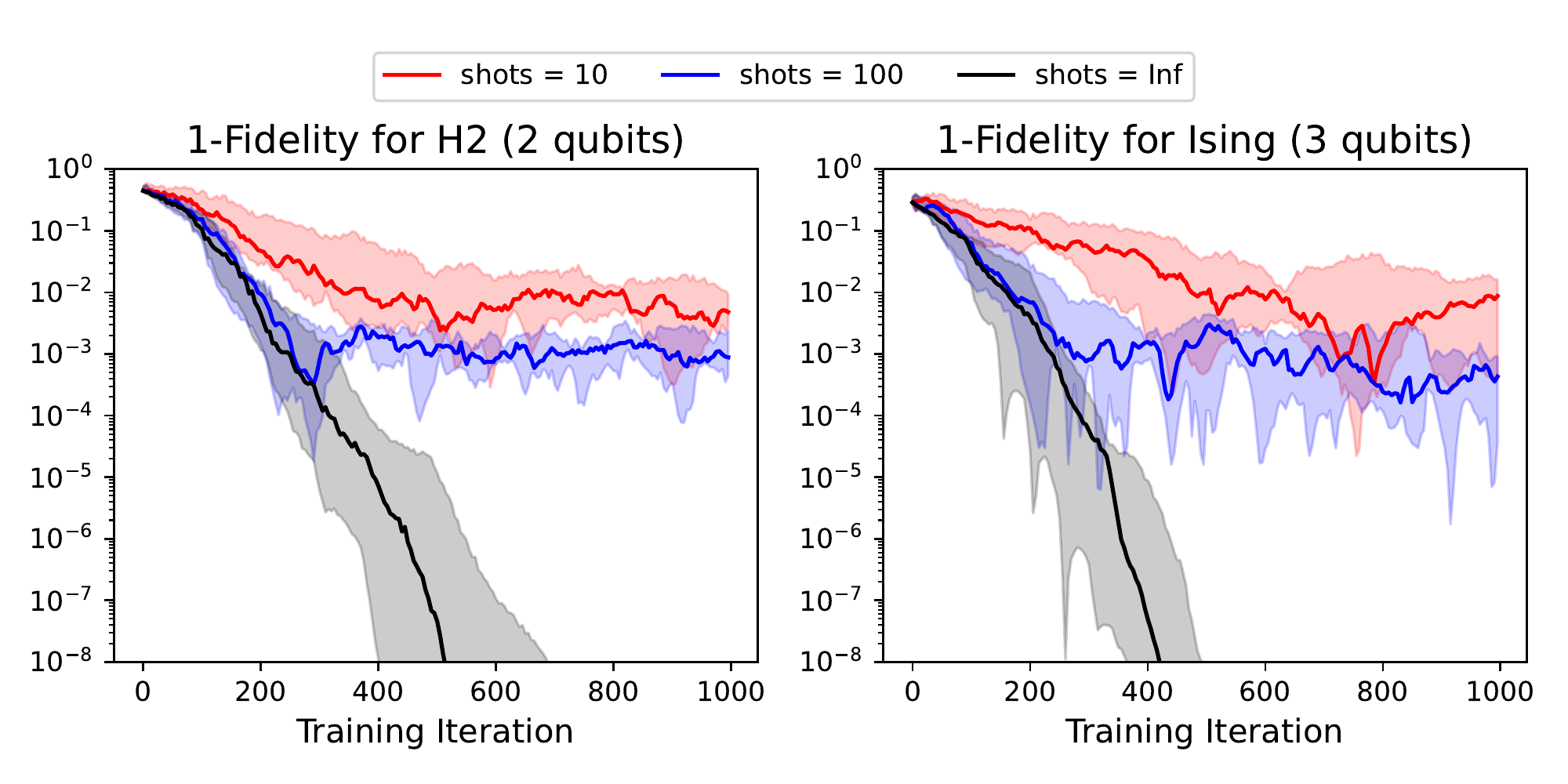}
\caption{For time-independent problems}
\label{fig::eg3}
\end{subfigure}
~
\begin{subfigure}[b]{0.37\textwidth}
\includegraphics[width=\textwidth]{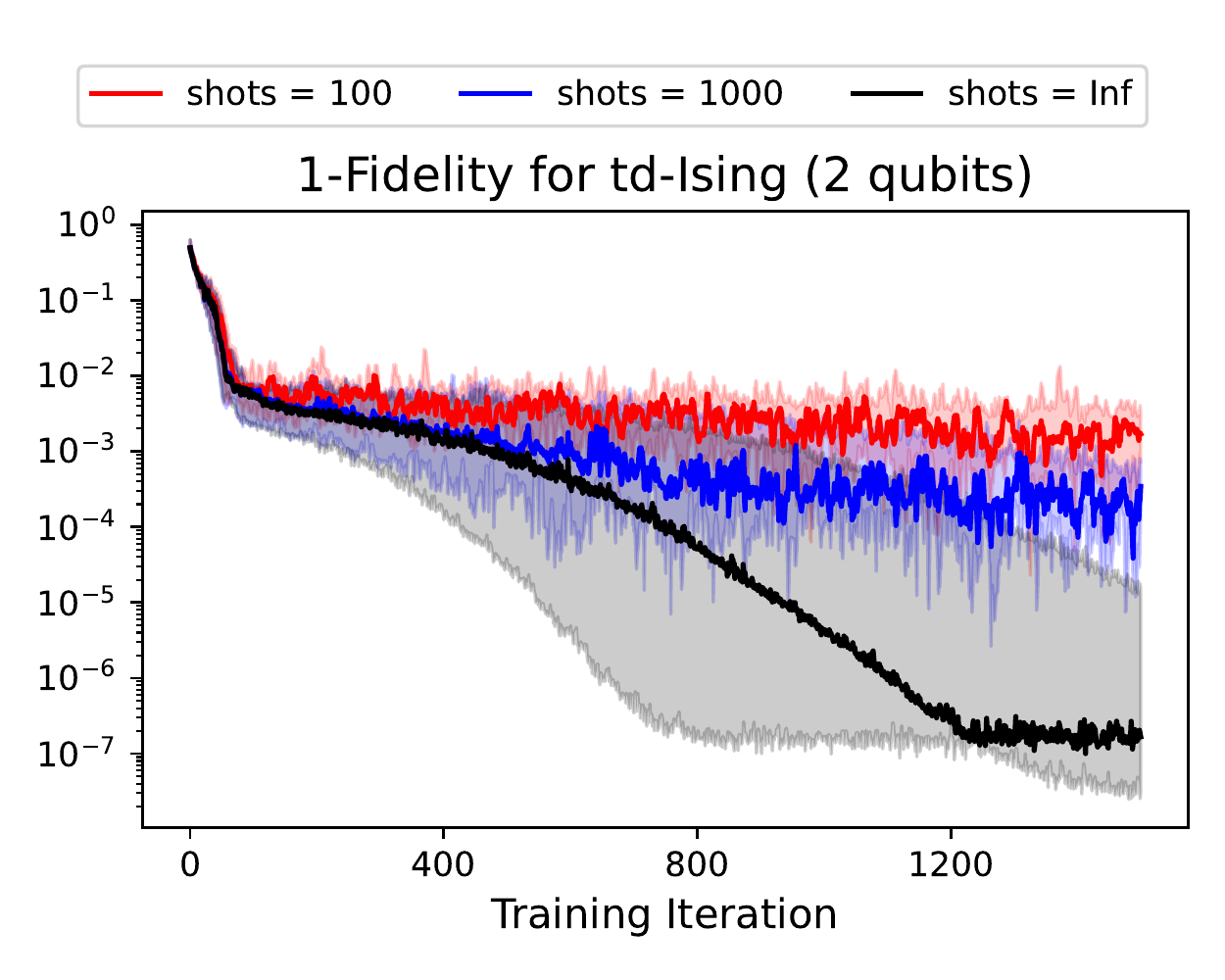}
\caption{For time-dependent problem}
\label{fig::eg3_td}
\end{subfigure}
\caption{Training results for \secref{section::hl_observable} using the loss function in Eq.~\eqref{eqn::loss_observable}: (a) for different time-independent Hamiltonians; (b) for two-site time-dependent quantum Ising chain (td-Ising). Each curve represents the median error from five independent training runs, and the shaded region indicates the maximum and minimum error.}
\end{figure}

\section{Summary and Outlook}

The QNODE and QNPDE framework is applicable for two classes of problems: 
one where a universal Hamiltonian is not required, and one where it is required. In the former case -- when the dynamics that needs to be learned can naturally be mapped onto a corresponding Hamiltonian -- QNODEs and QNPDEs are particularly suitable and allow for better problem-inspired ansatzes. These applications include learning problems in both closed and some open quantum systems, learning of (classical) ODEs and PDEs, in both autonomous and and non-autonomous settings. In the latter case -- when the learning involves data generated from dynamics too complicated for simple forms, and therefore a universal representation is needed -- then the QNODE and QNPDE can serve as a mathematical framework that can recover existing models for quantum neural networks (which are discrete time models). \\

 We introduce two quantum algorithms for gradient estimation for QNODEs and QNPDEs: one based on partial time discretisation and another that operates entirely in continuous time. These algorithms generalise backpropagation to the quantum regime and support a wide range of applications. For problems that do not require universal representation of Hamiltonians, these gradient estimation algorithms can be very efficient and the costs do not a priori depend on the number of unknown parameters to be learned. Each of these applications requires more in-depth study in their own right for specific scenarios, and we leave this to future work for specific models. \\

Looking forward, QNODEs and QNPDEs provide a powerful lens through which quantum machine learning models can be viewed or constructed, especially in contexts where physical dynamics are naturally continuous. Beyond efficiency gains in gradient estimation, this framework opens the door to systematic architectural design inspired by physical control theory and could catalyse the development of continuous-depth quantum models. Future works will explore rigorous complexity bounds, hardware implementations, and applications to real-world quantum control, generative modeling tasks, and extensions to open QNODEs/QNPDEs. 

\section*{Acknowledgements}
The authors would like to thank Lei Wang (Chinese Academy of Sciences) for helpful discussions.
The authors  acknowledge funding from Shanghai Pilot Program for Basic Research, the Science and Technology Commission of Shanghai Municipality (STCSM) grant no. 24LZ1401200 (21JC1402900), the NSFC grant No. 12341104, the Shanghai Jiao Tong University 2030 Initiative, and the Fundamental Research Funds for the Central Universities.
YC was also supported by NSFC grant No. 12401573.
 SJ was also partially supported by the NSFC grant No. 12426637. NL is in addition supported by grant NSFC No. 12471411.

\bibliography{RefQGF}

\appendix 

\section{Proofs for Sections~\ref{sec::revisiting} - \ref{alg::no_discretization}}

\subsection{Proof of Lemma~\ref{lem:3}} \label{app:theorem3}
Without loss of generality, we can only consider the time-derivative with respect to a particular index $m$.
Given the loss function $\mathcal{L}\big(\rho(T, \theta)\big)$, the chain rule gives 
\begin{align} \label{eq:losssigma}
    \frac{\partial \mathcal{L}(\rho(T, \theta))}{\partial \theta_m}=\tr\left(\frac{\delta \loss(\rho(T, \theta))}{\delta \rho} \sigma_m(T, \theta)\right),
\end{align}
where we define 
\begin{align}
     \sigma_m(s, \theta) := \frac{\partial \rho(s, \theta)}{\partial \theta_m}, \qquad s\in [0,T].
\end{align}
To solve for $\sigma_m(T, \theta)$, we write down the dynamics of $\sigma_m(s, \theta)$
\begin{align} \label{eq:sigmaevo}
   &  \frac{\partial \sigma_m(s, \theta)}{\partial s}=\frac{\partial}{\partial \theta_m}\frac{\partial \rho(s, \theta)}{\partial s}=-i\frac{\partial}{\partial \theta_m}[H(s, \theta), \rho(s, \theta)] \nonumber \\
   &=-i[H(s, \theta), \sigma_m(s, \theta)]-i \left[\frac{\partial H(s, \theta)}{\partial \theta_m}, \rho(s, \theta)\right], \qquad \sigma_m(0, \theta)=0,
\end{align}
where the initial condition $\sigma_m(0, \theta)=0$ is due to the fact that the initial state $\rho_0$ is independent of the parameter $\theta$ inside Hamiltonians. 
Note that the homogeneous part of Eq.~\eqref{eq:sigmaevo} has solution $\sigma_m^h(T, \theta)=U_{\theta}(s, T) \sigma_m^h(s, \theta)U_{\theta}(s, T)^{\dagger}$ where $U_{\theta}(s, T)=\mathcal{T} \exp(-i\int_s^T H(\tau, \theta) d \tau)$. The inhomogeneous part of Eq.~\eqref{eq:sigmaevo} is $-i[\frac{\partial H(s,\theta)}{\partial \theta_m}, \rho(s, \theta)]$, which can be viewed as an additional force term. Either by introducing $U_\theta(s,T)$ as integrating factors or by applying Duhamel's principle directly, Eq.~\eqref{eq:sigmaevo} can be expressed in the integral form as follows:
\begin{align} \label{eq:sigmafinal}
    \sigma_m(T, \theta)=i \int_0^T U_{\theta}(s, T)\left[\rho(s, \theta), \frac{\partial H(s, \theta)}{\partial \theta_m}\right] U_{\theta}(s, T)^{\dagger} ds.
\end{align}
Inserting Eq.~\eqref{eq:sigmafinal} into Eq.~\eqref{eq:losssigma} gives 
\begin{align}
    &  \frac{\partial \mathcal{L}(\rho(T, \theta))}{\partial \theta_m}=i \int_0^T \tr\left(\frac{\delta \loss(\rho(T, \theta))}{\delta \rho}U_{\theta}(s, T)\left[\rho(s, \theta), \frac{\partial H(s, \theta)}{\partial \theta_m}\right]U_{\theta}(s, T)^{\dagger} \right) ds \nonumber \\
    &=i \int_0^T \tr\left(U_{\theta}(s, T)^{\dagger}\frac{\delta \loss(\rho(T, \theta))}{\delta \rho} U_{\theta}(s, T)\left[\rho(s, \theta), \frac{\partial H(s, \theta)}{\partial \theta_m}\right] \right) ds \nonumber \\
    &= i \int_0^T \tr\left(U_{\theta}(T, s)\frac{\delta \loss(\rho(T, \theta))}{\delta \rho} U_{\theta}(T, s)^{\dagger} \left[\rho(s, \theta), \frac{\partial H(s, \theta)}{\partial \theta_m}\right] \right) ds \nonumber \\
    &=i\int_0^T A_T(\theta) \tr\left(a(s, \theta)\left[\rho(s, \theta), \frac{\partial H(s, \theta)}{\partial \theta_m}\right] \right) ds,
\end{align}
where in the last line, we introduce the notation
\begin{align*}
a(s,\theta) :=  \frac{U_{\theta}(T, s)\frac{\delta \loss(\rho(T, \theta))}{\delta \rho}U_{\theta}(T, s)^{\dagger}}{A_T(\theta)}, \qquad A_T(\theta) := \tr\big(\delta \loss(\rho(T, \theta))/\delta \rho\big).
\end{align*}
Clearly, the state $a(s, \theta)$ undergoes unitary evolution according to Eq.~\eqref{eq:asmatrix} with a terminal condition, and $\tr\big(a(s,\theta)\big) = 1$ is preserved for all $s \in [0, T]$. One can also rewrite
\begin{align}
    \tr\left(a(s, \theta)\left[\rho(s, \theta), \frac{\partial H(s, \theta)}{\partial \theta_m}\right]\right)=\tr\left(\left[\frac{\partial H(s, \theta)}{\partial \theta_m}, a(s, \theta)\right] \rho(s, \theta)\right),
\end{align}
and the results of the lemma follows. 

\subsection{Proof of Theorem~\ref{thm::main}} \label{app:theoremmixedtopure}

It is clear that the state
\begin{align*}
\eta(s,\theta) &= \big(\ketbra{0} \otimes \mathbf{1}\otimes \mathbf{1} + \ketbra{1} \otimes \boldsymbol{S}\big) \big(\ketbra{+}\otimes a(s,\theta) \otimes \rho(s,\theta)\big)  \big(\ketbra{0} \otimes \mathbf{1}\otimes \mathbf{1} + \ketbra{1} \otimes \boldsymbol{S}\big)\\
&= \frac{1}{2}\left(\begin{aligned} & \ketbra{0} \otimes a(s,\theta)\otimes \rho(s,\theta) + \ket{0}\bra{1} \otimes \big(a(s,\theta)\otimes \rho(s,\theta)\big) \boldsymbol{S}  \\
&\qquad + \ket{1}\bra{0} \otimes \boldsymbol{S}\big(a(s,\theta) \otimes \rho(s,\theta) \big) + \ketbra{1} \otimes \rho(s,\theta)\otimes a(s,\theta) \end{aligned}\right).
\end{align*}

It can be observed that the integrand in Eq.~\eqref{eq:rhodensitylemma} can be written as
\begin{align*}
& i A_T(\theta)  \tr\left(\left[\frac{\partial H(s, \theta)}{\partial \theta_m},a(s, \theta)\right] \rho(s, \theta)\right) \\
=& i A_T(\theta) \tr\big(\frac{\partial H(s, \theta)}{\partial \theta_m} a(s,\theta)\rho(s,\theta) - \frac{\partial H(s, \theta)}{\partial \theta_m} \rho(s,\theta) a(s,\theta)\big) \\
=& 2 A_T(\theta) \tr\big(\obs_m(s,\theta) \eta(s,\theta)\big).
\end{align*}

To validate the last equality explicitly, let us decompose $\rho(s,\theta) = \sum_{k} \lambda_k \ketbra{b_k}$ and $a(s,\theta) = \sum_{k} \mu_k \ketbra{a_k}$. Then by direct decomposition, one can verify that 
\begin{align*}
& 2 \tr\big(\obs_m(s,\theta) \eta(s,\theta)\big) \\
=& -i \tr\Big(\big(\frac{\partial H(s, \theta)}{\partial \theta_m}\otimes \mathbf{1} \big) \boldsymbol{S} \big(a(s,\theta)\otimes \rho(s,\theta)\big)\Big)  + i \tr\Big(\big(\frac{\partial H(s, \theta)}{\partial \theta_m}\otimes \mathbf{1} \big) \big(a(s,\theta)\otimes \rho(s,\theta)\big)\boldsymbol{S} \Big) \\
=& -i \sum_{j,k} \mu_j \lambda_k \bra{b_k}{a_j}\rangle \bra{a_j} \frac{\partial H(s, \theta)}{\partial \theta_m}\ket{b_k}   + i \sum_{j,k} \mu_j \lambda_k \bra{b_k}  \frac{\partial H(s, \theta)}{\partial \theta_m} \ket{a_j} \bra{a_j} {b_k}\rangle \\
=& -i \sum_{k} \lambda_k \bra{b_k} a(s,\theta)  \frac{\partial H(s, \theta)}{\partial \theta_m} \ket{b_k} + i \sum_{k} \lambda_k \bra{b_k}  \frac{\partial H(s, \theta)}{\partial \theta_m} a(s,\theta) \ket{b_k}\\
=& -i \tr\big( \frac{\partial H(s, \theta)}{\partial \theta_m} \rho(s,\theta) a(s,\theta)\big)  + i \tr\big( \frac{\partial H(s, \theta)}{\partial \theta_m} a(s,\theta) \rho(s,\theta)\big).
\end{align*}
Thus the proof is complete.

\subsection{Proof of \thmref{thm5}} \label{app:cor1}

The eigenstates $\{|s\rangle\}_{s\in \Real}$ form an orthonormal basis with $\int |s\rangle \langle s|\ ds=\mathbf{1}_s$ and the position operator $\hat{s} :=\int s|s\rangle \langle s|\ ds$.
Recall from Eq.~\eqref{eq:etahat} that
\begin{align*}
\widehat{\eta}(\theta) = \int g(s) \ketbra{s} \otimes \eta(s, \theta)\ ds, \qquad \int\ g(s) ds=1,
\end{align*}

By confining to time-independent Hamiltonians,  the operator (from Theorem~\ref{thm::main}) $O_m(s, \theta) \rightarrow O_m(\theta)$ is independent of $s$. Then one easily observes that 
\begin{align*}
\tr\big((\unit_s\otimes \obs_m(\theta))\ \widehat{\eta} (\theta)\big) = \int  g(s) \tr\big(\obs_m(\theta) \eta(s,\theta)\big) ds .
\end{align*}
When $g(s)=g_{top}(s)$, then clearly $\int  g(s) \tr\big(\obs_m(\theta) \eta(s,\theta)\big) ds= \frac{1}{T} \int_0^{T} \tr\big(\obs_m(\theta) \eta(s,\theta)\big) ds$. \\

{\noindent \emph{Proof for the error bounds of imperfect $g(s)$:}}\\

Since the quantum $\int g_{{top}}(s) \ketbra{s} ds$ cannot be exactly prepared (if we keep $s$ continuous) due to discontinuity at $s=0$ and $s=T$, we  instead approximate it using a qumode $\int g(s) \ketbra{s} ds$. 
By Eq.~\eqref{eq:lem2exact} and the triangle inequality,
\begin{align*}
 & \left| \frac{1}{2T A_T(\theta) }\frac{\partial \mathcal{L}_\theta}{\partial \theta_m}- 
    \tr\big(\unit_s\otimes \obs_m(\theta)\ \widehat{\eta}(\theta) \big) \right| \\
    \myeq{\eqref{eq:lem2exact}}&\ \abs{\tr\Big(\unit_s\otimes \obs_m(\theta) \int \big(g_{top}(s) - g(s)\big)\ketbra{s}\otimes \eta(s,\theta)\ ds \Big)}\\
    = &\ \abs{\int \big(g_{top}(s) - g(s)\big)\ ds \tr\big( \obs_m(\theta) \eta(s,\theta)\big)}  \\
    \le &\ \int \abs{g_{top}(s) - g(s)}\ \abs{\tr\big( \obs_m(\theta) \eta(s,\theta)\big)} ds.
 \end{align*}
Since $\eta(s,\theta)$ is a density matrix which is positive semi-definite and $\tr\big(\eta(s,\theta)\big) = 1$, one has
\begin{align*}
\abs{\tr\big( \obs_m(\theta)\eta(s,\theta)\big)} \le \norm{\obs_m(\theta)}_{\infty},
\end{align*}
where $\norm{\obs_m(\theta)}_{\infty}$ is the operator norm of $\obs_m(\theta)$, or equivalently, the larges eigenvalue of $\abs{\obs_m(\theta)}$. 
More specifically, if we decompose $O_m(\theta) = \sum_k \lambda_k \ket{k}\bra{k}$, then 
$\abs{\tr\big( \obs_m(\theta)\eta(s,\theta)\big)} \le \sum_{k} \abs{\lambda_k} \bra{k}\eta(s,\theta)\ket{k} \le \norm{\obs_m(\theta)}_{\infty}\sum_k \bra{k} \eta(s,\theta)\ket{k} = \norm{\obs_m(\theta)}_{\infty} \tr\big(\eta(s,\theta)\big) = \norm{\obs_m(\theta)}_{\infty}$.
Therefore, 
\begin{align*}
& \left| \frac{1}{2T A_T(\theta) }\frac{\partial \mathcal{L}_\theta}{\partial \theta_m}- 
    \tr\big(\unit_s\otimes \obs_m(\theta)\ \widehat{\eta}(\theta) \big) \right| \\
    &\qquad \le\ \norm{\obs_m(\theta)}_{\infty}  \int \abs{g_{top}(s) - g(s)} ds \\
    &\qquad =  \norm{\obs_m(\theta)}_{\infty}\ \norm{g_{top} - g}_{\text{TV}}.
\end{align*}

{\noindent \emph{Proof for the realization for time-independent case:}}

\begin{align*}
\widehat{\eta}(\theta)  &\myeq{\eqref{eq:etahat}} \int g(s) \ketbra{s} \otimes \eta(s, \theta)\ ds \\
&= \int g(s) \ketbra{s} \otimes \mathcal{V}_\theta(s) \eta(0, \theta)\mathcal{V}_\theta(s)^\dagger\ ds \\
&\myeq{\eqref{eq:psisstate2}}  \int g(s) \ketbra{s} \otimes \controlswap \big(\unit \otimes e^{i H(\theta) (T-s)} \otimes e^{-i H(\theta) s}\big) \eta(0, \theta)  \big(\unit \otimes e^{-i H(\theta) (T-s)} \otimes e^{i H(\theta) s}\big) \controlswap ds \\
&= \big(\unit_s\otimes \controlswap \big) W \big(\int g(s) \ketbra{s} \otimes  \eta(0, \theta) ds\big) W^\dagger \big(\unit_s\otimes \controlswap \big) ,
\end{align*}
where the unitary gate $W$ acting on the $s$-mode, auxiliary qubit $\hbt_{\aux}$, the adjoint mode $\hbt_{\adj}$ and the original mode $\hbt$, consists of the application of two unitary gates sequentially 
\begin{align*}
W =  \big(\unit_{\hbt_{\aux}} \otimes e^{-i H(\theta)\otimes  \hat{s}}\otimes \unit_{\hbt_{\adj}}\big) \big(\unit_{\hbt_{\aux}} \otimes e^{i H(\theta) \otimes (T\unit_s -\hat{s})}\otimes \unit_{\hbt}\big).
\end{align*}

\section{Complexity scaling of quantum measurements for \secref{sec::cost}}
\label{sec::complexity}

\subsection{A general setup and assumptions}

The approximation of the gradient is denoted as $\mathscr{X}_m$:
\begin{align*}
\frac{1}{2 A_T(\theta)} \frac{\partial \mathcal{L}_{\theta}}{\partial \theta_m} \approx \mathscr{X}_m,
\end{align*}
where $\mathscr{X}_m$ is a general notation with possibly different implementations. Since the gradient value takes the integral form, we will consider both midpoint scheme in Eq.~\eqref{eqn::approxiamte_Im} and direct Monte Carlo sampling in Eq.~\eqref{eqn::mc_time} to approximate the integral. 
One may also improve the Monte Carlo sampling via adopting variance reduction techniques like the stratified sampling \cite{liu_monte_2004}, which we will not pursue herein. For simplicity, we consider the above two simple approximations below.\\

Then the classical (batch) stochastic gradient descent algorithm is
\begin{align}
\label{eqn::sgd}
\theta^{(k+1)}  = \theta^{(k)} - h 2 A_T(\theta^{(k)}) \mathscr{X},
\end{align}
where $h$ is the learning rate or say the step size, and $\mathscr{X}$ is a vector whose components are $\mathscr{X}_m$ ($m=1,2,\cdots,M$).
Technically, since we will approximate $\mathscr{X}$ via quantum measurement outcomes with a limited number of repetitive samples, it is typically a \emph{batch SGD algorithm}. The approximation $\mathscr{X}$ depends on the current stage $k$, but for simplicity, we shall not indicate that explicitly for ease of notations. \\

In what follows, we will discuss two approximations, and estimate the resources required to ensure that 
\begin{align}
\label{eqn::sgd_error}
\ee\big[\loss_{\theta^{(\Niter)}} - \loss_{\min}\big] \le \Delta,
\end{align}
where $\Delta$ is some prescribed error tolerance, and 
\begin{align*}
 \loss_{\min} = \inf_{\theta} \loss_\theta
\end{align*}
is the (globally) minimum function value for the loss.\\

To find the quantum resources required for $L_\infty$ error approximation, we will  assume the following for convenience:
\begin{assumption}
\label{assume::norm_Hk}
Without loss of generality, we assume that  each term $H_k$ has operator norm $\norm{H_k}_{\infty}=1$.
Suppose that for each $s\in [0,T]$ and index $1\le k\le K$, $1\le m\le M$, we have uniform bounds 
$
\abs{\frac{\partial f_k(s, \theta)}{\partial \theta_m}} \le \alpha.
$
\end{assumption}

An \emph{alternative} stronger assumption is the following, in case $f_k$ are relatively smoother,
\begin{assumption}
\label{assume::norm_Hk_2}
Without loss of generality, we assume that  each term $H_k$ has operator norm $\norm{H_k}_{\infty}=1$.
Suppose that for each $s\in [0,T]$ and index $1\le k\le K$, $1\le m\le M$, we have uniform bounds 
\begin{align*}
\abs{f_k(s, \theta)},\ \abs{\frac{\partial f_k(s, \theta)}{\partial s}},\ \abs{\frac{\partial f_k(s, \theta)}{\partial \theta_m}},\ \abs{\frac{\partial^2 f_k(s, \theta)}{\partial \theta_m\partial s}},\ \abs{\frac{\partial^3 f_k(s, \theta)}{\partial \theta_m\partial s^2}} \le \alpha.
\end{align*}
\end{assumption}
Sine the scaling of this upper bound is not our priority, we slightly abuse the notation and use the notation $\alpha$ for both assumptions.
To analyse the performance of min-batch gradient descent, we make two additional assumptions.

\begin{assumption}
\label{assume::PL}
$\loss_\theta$ satisfies the $\mu$-\PL{} condition (or $\mu$-PL in short), meaning that
\begin{align}
\label{eqn::PL}
\mathcal{L}_\theta - \loss_{\min} \le \frac{1}{2\mu} \norm{\nabla \mathcal{L}_\theta}^2_{2},\ \  \forall \theta.
\end{align}
\end{assumption}

\begin{assumption}
\label{assume::AT}
Assume that $A_T(\theta)$ is known a priori with uniform upper bound $\mathsf{A}$ for any $\theta$, namely,
\begin{align*}
\abs{A_T(\theta)} \le \mathsf{A}, \qquad \forall \theta.
\end{align*}
\end{assumption}

The assumption that $A_T(\theta)$ is known and bounded holds e.g., for quantum state preparation in \secref{eg::qsp}. The condition for $\mu$-PL condition needs to be validated case by case.
We acknowledge that the error bound below may not be tight for some parameters. In this work, we mainly focus on the scaling with respect to the number of local operator $K$, the time span $T$ and the number of parameters $M$.
A summary of main results could be found in Table~\ref{table::2}.

       \subsection{Gradient estimate with time discretization}
    
    Let us denote the time-discretised expression in Eq.~\eqref{eqn::discretise} as 
    \begin{align}
    \label{eqn::approxiamte_Im}
    \mathscr{I}_m :=  \sum_{j=1}^{\numtime} \sum_{k=1}^K \Delta s  \frac{\partial f_k(s_j, \theta)}{\partial \theta_m} \tr\big(\sigma_Y \otimes H_k \otimes \unit\ \eta(s_j, \theta)\big), \qquad s_j = (j-1/2) \Delta s.
    \end{align}
    For simplicity, we shall simply choose the midpoint rule with $\Delta s = \frac{T}{\numtime}$.
    Suppose we use $L$ samples to estimate the quantum observable $\tr\big(\sigma_Y \otimes H_k \otimes \unit\ \eta(s_j, \theta)\big)$ for each $k$ and $j$, and let its quantum measurement outcome be $X_{k,j}$, which are independent random variables. Let 
    \begin{align*}
    \mathscr{X}_m :=  \sum_{j=1}^{\numtime} \sum_{k=1}^K \Delta s \frac{\partial f_k(s_j, \theta)}{\partial \theta_m} X_{k,j},
    \end{align*}
    which is an unbiased estimator of $\mathscr{I}_m$. The quantity $\mathscr{I}_m$ characterized the approximation of gradient with some bias, and $\mathscr{X}_m$ is the stochastic realization of $\mathscr{I}_m$ due to random nature of quantum measurements. \\
    
    We need the following lemma about the moment generating function for quantum measurements:
    \begin{lemma}
    \label{lemma::Hoeffding}
    Under Assumption~\ref{assume::norm_Hk}, suppose that $X$ is a random variable for the averaged outcome of quantum measurement of $\tr\big(\sigma_Y \otimes H_k \otimes \unit \eta(s, \theta)\big)$ for some fixed $s$ and $k$, using in total $L$ samples, 
    its moment generating function is bounded by
    \begin{align}
    \label{eqn::Hoeffding}
    \mathbb{E}\big( e^{\lambda (X - \mathbb{E}(X))}\big) \le e^{\frac{\lambda^2}{2 L^2}}, \qquad \forall \lambda\in \Real.
    \end{align}
    \end{lemma}
    
    \begin{proof}
    Suppose $Y_1, Y_2, \cdots, Y_L$ are i.i.d. random samples for $X$, then due to the given assumptions, we know that $Y_i\in [-1,1]$. The estimator $X = \frac{1}{L}\sum_{i=1}^L Y_i$ is the average. By independence of $Y_i$, we have
    \begin{align*}
    \mathbb{E}\big( e^{\lambda (X - \mathbb{E}(X))}\big) = \prod_{i=1}^{L} \ee\big(e^{\frac{\lambda}{L}\big(Y_i - \mathbb{E}(Y_i)\big)}\big) \le \prod_{i=1}^L e^{\frac{\lambda^2}{2L^2}} =  e^{\frac{\lambda^2}{2 L}},
    \end{align*}
    where we used Hoeffding’s lemma to obtain the above inequality.
    \end{proof}
    
    \begin{lemma}
    \label{lemma::moment_bound}
    Under Assumptions~\ref{assume::norm_Hk}, let us fix $m$, and fix the choice of grid points $s_j$ where $\Delta s = \frac{T}{\numtime}$. suppose that $X_{k,j}$ is a random variable for measuring $\tr(\sigma_Y \otimes H_k \otimes \unit \eta(s_j, \theta)\big)$ with $L$ samples for each index pair $(k,j)$, then we have
    \begin{align*}
     \ee\bigg[e^{\lambda \abs{\sum_{j=1}^{\numtime}\sum_{k=1}^K \Delta s \frac{\partial f_k(s_j,\theta)}{\partial \theta_m} (X_{k,j} - \ee(X_{k,j}))} }\bigg]  \le 2 \exp\big(\frac{\lambda^2 T^2 \alpha^2 K}{2 \numtime L}\big).
    \end{align*}
    \end{lemma}
    
    \begin{proof}
    This conclusion can be obtained by the last lemma:
    \begin{align*}
     &\ \ee\bigg[e^{\lambda \abs{\sum_{j=1}^{\numtime}\sum_{k=1}^K \Delta s \frac{\partial f_k(s_j,\theta)}{\partial \theta_m} (X_{k,j} - \ee(X_{k,j}))} }\bigg]  \\
     \le &\  \ee\bigg[e^{\lambda \sum_{j=1}^{\numtime}\sum_{k=1}^K \Delta s \frac{\partial f_k(s_j,\theta)}{\partial \theta_m} (X_{k,j} - \ee(X_{k,j}))} \bigg] +  \ee\bigg[e^{-\lambda \sum_{j=1}^{\numtime}\sum_{k=1}^K \Delta s \frac{\partial f_k(s_j,\theta)}{\partial \theta_m} (X_{k,j} - \ee(X_{k,j}))} \bigg] \\
     \myle{\eqref{eqn::Hoeffding}} &\  \prod_{j=1}^{\numtime} \prod_{k=1}^K e^{\lambda \Delta s \frac{\partial f_k(s_j,\theta)}{\partial \theta_m} (X_{k,j} - \ee(X_{k,j}))} + \prod_{j=1}^{\numtime} \prod_{k=1}^K e^{-\lambda \Delta s \frac{\partial f_k(s_j,\theta)}{\partial \theta_m} (X_{k,j} - \ee(X_{k,j}))}  \\
     \le &\ 2 \prod_{j=1}^{\numtime} \prod_{k=1}^K \exp\big(\frac{(\lambda \Delta s \frac{\partial f_k(s_j,\theta)}{\partial \theta_m})^2}{2L}\big) \\
     \le&\ 2 \prod_{j=1}^{\numtime} \prod_{k=1}^K \exp\big(\frac{(\lambda \Delta s \alpha)^2}{2L}\big) = 2 \exp\big(\frac{\lambda^2 T^2 \alpha^2 K}{2 \numtime L}\big) .
    \end{align*}
   We used $e^{\lambda \abs{x}} \le e^{\lambda x}+e^{-\lambda x}$ for the first inequality.
    \end{proof}

    \subsubsection{$L^\infty$ error estimate}
    
    In what follows, we will estimate quantum measurements required to ensure the closeness of gradient approximation in $L_\infty$ norm:
\begin{align}
\label{eqn::prob_error}
\mathbb{P}\big(\norm{\frac{1}{2 A_T(\theta)} \nabla \loss_\theta - \mathscr{X}}_{\infty} \ge \delta \big) \le \varepsilon,
\end{align}
where $\delta$ and $\varepsilon$ are some prescribed small parameters. Though this is not needed below for estimating quantum resources required in machine learning training, such a probabilistic type estimate is common in e.g., developing quantum algorithms for gradient estimate, and hence we provide a rigorous estimate for comparison purposes. Then we have the following conclusion:
    \begin{proposition}
    \label{prop::midpoint_inf}
    Under Assumption~\ref{assume::norm_Hk_2},
    \begin{itemize}
    \item[(i)] Suppose that we choose midpoint scheme with uniform grid points for Eq.~\eqref{eqn::discretise}. We only need to choose 
    \begin{align*}
    \numtime \ge  \order{(\alpha T K)^{3/2} \delta^{-1/2}},
    \end{align*}
    then one could ensure that $\abs{\frac{1}{2 A_T(\theta)} \frac{\partial \mathcal{L}_\theta}{\partial \theta_m} - \mathscr{I}_m}\le \delta/2$ for any $1\le m\le M$. 
    
    \item[(ii)] To ensure that
    \begin{align*}
    \mathbb{P}\Big(\norm{\mathscr{X} - \frac{1}{2 A_T(\theta)} \nabla_\theta \mathcal{L}_\theta}_{\infty} \ge \delta\Big) \le  \varepsilon,
    \end{align*}
    it is sufficient to choose the total number of quantum measurements as
    \begin{align*}
     \max\Big\{ \order{(\frac{\alpha T K}{\delta})^2 \log\big(\frac{M}{\varepsilon})}, \order{(\alpha T)^{3/2} K^{5/2} \delta^{-1/2}}\Big\}.
    \end{align*}
    \end{itemize}
    \end{proposition}
    
    \begin{proof}
    
    \noindent \textit{Proof of part (i):} 
    \begin{align*}
    \frac{1}{2 A_T(\theta)} \frac{\partial \mathcal{L}_{\theta}}{\partial \theta_m} =  \int_0^T  \tr\big(\obs_m(s, \theta) \eta(s,\theta)\big)\ ds = \sum_{k=1}^{K} \int_{0}^{T} \frac{\partial f_k(s,\theta)}{\partial \theta_m} \tr\big(\sigma_Y \otimes H_k \otimes \unit\ \eta(s, \theta)\big) ds.
    \end{align*}
    To ensure that the error from time-discretization is bounded by $\delta/2$, it is sufficient to ensure that the midpoint scheme for each term $k$ has error
    \begin{align*}
    \abs{\int_{0}^{T}  \frac{\partial f_k(s,\theta)}{\partial \theta_m} \tr\big(\sigma_Y \otimes H_k \otimes \unit \eta(s, \theta)\big)\ ds - \sum_{j=1}^{\numtime} \Delta s  \frac{\partial f_k(s_j, \theta)}{\partial \theta_m} \tr\big(\sigma_Y \otimes H_k \otimes \unit\ \eta(s_j, \theta)\big)}\le \frac{\delta}{2K}.
    \end{align*}
    We can directly verify that for $h(s, \theta) =  \frac{\partial f_k(s,\theta)}{\partial \theta_m} \tr\big(\sigma_Y \otimes H_k \otimes \unit \eta(s, \theta)\big)$, its second-order derivative with respect to time is bounded by
    \begin{align*}
    \abs{\frac{\partial^2 h(s,\theta)}{\partial s^2}} \le &\ \abs{\frac{\partial^3 f_k(s,\theta)}{\partial \theta_m\partial s^2} \tr\big(\sigma_Y \otimes H_k \otimes \unit\ \eta(s, \theta)\big)}  \\
    & + 2\abs{\frac{\partial^2 f_k(s,\theta)}{\partial \theta_m\partial s} \tr\big(\sigma_Y \otimes H_k \otimes \unit \partial_s \eta(s, \theta)\big)} \\
    &+ \abs{\frac{\partial f_k(s,\theta)}{\partial \theta_m} \tr\big(\sigma_Y \otimes H_k \otimes \unit\  \partial_{s}^2 \eta(s, \theta)\big)} \\
    \le &\ \alpha + 2\alpha \abs{\tr\big(\sigma_Y \otimes H_k \otimes \unit \partial_s \eta(s, \theta)\big)} + \alpha \abs{\tr\big(\sigma_Y \otimes H_k \otimes \unit \partial_{s}^2 \eta(s, \theta)\big)} \\
    \le&\ \alpha + 2 \alpha (4 \alpha K) + \alpha \big(16 \alpha^2 K^2 + 4 \alpha K\big) \\
    \le&\ \alpha (1 + 6 \alpha K)^2.
    \end{align*}
   To get the second last line, we need to first use the estimate that
   \begin{align*}
     \partial_s \mathcal{V}_{\theta}(s) \myeq{\eqref{eq:psisstate2}} & \controlswap (-i) \big(\unit \otimes H_\theta(s) \otimes \unit + \unit \otimes \unit \otimes H_\theta(s) \big)
    \big(\unit \otimes U_{\theta}(T, s) \otimes U_{\theta}(0, s)\big)\\
         \partial_{ss} \mathcal{V}_{\theta}(s) \myeq{\eqref{eq:psisstate2}} & \controlswap (-i) \big(\unit \otimes \partial_s H_\theta(s) \otimes \unit + \unit \otimes \unit \otimes \partial_s H_\theta(s) \big)
    \big(\unit \otimes U_{\theta}(T, s) \otimes U_{\theta}(0, s)\big) \\
    &\qquad - \controlswap \big(\unit \otimes \partial_s H_\theta(s) \otimes \unit + \unit \otimes \unit \otimes \partial_s H_\theta(s) \big)^2 
    \big(\unit \otimes U_{\theta}(T, s) \otimes U_{\theta}(0, s)\big),
    \end{align*}
    so that 
    \begin{align*}
    \norm{\partial_s \mathcal{V}_{\theta}(s)}_{\infty} \le &\ \norm{\unit \otimes H_\theta(s) \otimes \unit + \unit \otimes \unit \otimes H_\theta(s)}_{\infty} \le 2 \alpha K,
    \end{align*}
    and similarly 
    \begin{align*}
    \norm{\partial_{ss} \mathcal{V}_{\theta}(s)}_{\infty} \le &\  2 \alpha K + (2\alpha K)^2.
    \end{align*}
    Hence,
    \begin{align*}
\abs{\tr\big(\sigma_Y \otimes H_k \otimes \unit \partial_s \eta(s, \theta)\big)} = &\  \abs{\tr\big((\sigma_Y \otimes H_k \otimes \unit) (\partial_s \mathcal{V}_\theta(s) \eta(0,\theta) \mathcal{V}_\theta(s)^\dagger + \mathcal{V}_\theta(s) \eta(0,\theta) \partial_s \mathcal{V}_\theta(s)^\dagger)\big)}\\
\le &\ \norm{\mathcal{V}_\theta(s)^\dagger (\sigma_Y \otimes H_k \otimes \unit) \partial_s \mathcal{V}_\theta(s)}_{\infty} + \norm{\partial_s \mathcal{V}_\theta(s)^\dagger (\sigma_Y \otimes H_k \otimes \unit) \mathcal{V}_\theta(s)}_{\infty} \\
\le &\ 4 \alpha K
   \end{align*}
   and similarly,
   \begin{align*}
   \abs{\tr\big(\sigma_Y \otimes H_k \otimes \unit \partial_{ss} \eta(s, \theta)\big)} \le 16 \alpha^2 K^2 + 4 \alpha K
   \end{align*}
    
    Therefore, it is easy to show that for the midpoint scheme with step size $\Delta s$, the time-discretization error is bounded by 
    \begin{align}
    \label{eqn::error_midpoint}
    \frac{1}{24} \Delta s^2 T \alpha (1 + 6 \alpha K)^2,
    \end{align}
    which is required to be smaller than $\frac{\delta}{2K}$. Then it is easy to show that
    \begin{align*}
    \numtime \ge \sqrt{\frac{T^3 \alpha (1 + 6 \alpha K)^2 K}{12\delta}} = \order{(\alpha T K)^{3/2} \delta^{-1/2}}.
    \end{align*}
    
    \noindent \textit{Proof of part (ii):}
    We assume that the conclusion in part (i) holds as a starting point.
    By triangle inequality, we obtain
    \begin{align*}
    & \mathbb{P}\Big(\text{max}_{m=1}^{M} \abs{\mathscr{X}_m - \frac{1}{2 A_T(\theta)} \frac{\partial \mathcal{L}_\theta}{\partial \theta_m}} \ge \delta\Big) \\
    \le&\ \mathbb{P}\Big(\text{max}_{m=1}^{M} \abs{\mathscr{X}_m - \mathscr{I}_m} \ge \delta/2\Big)\\
    \le&\ \sum_{m=1}^{M} \mathbb{P}\Big(\abs{\mathscr{X}_m - \mathscr{I}_m} \ge \delta/2\Big)\\
    =& \sum_{m=1}^M \mathbb{P}\Big(\abs{\sum_{j=1}^{\numtime}\sum_{k=1}^K \Delta s \frac{\partial f_k(s_j,\theta)}{\partial \theta_m} (X_{k,j} - \ee(X_{k,j}))} \ge \delta/2\Big)\\
    \le &\ \sum_{m=1}^M e^{-\lambda \delta/2} \ee\bigg[e^{\lambda \abs{\sum_{j=1}^{\numtime}\sum_{k=1}^K \Delta s \frac{\partial f_k(s_j,\theta)}{\partial \theta_m} (X_{k,j} - \ee(X_{k,j}))} }\bigg] \qquad \text{ (Markov inequality for any $\lambda > 0$)}\\
    \le &\ 2 M\ e^{-\lambda \delta/2} \exp\big(\frac{\lambda^2 T^2 \alpha^2 K}{2 \numtime L}\big).
    \end{align*}
    By optimizing the upper bound over all $\lambda$, we have
    \begin{align*}
     \mathbb{P}\Big(\text{max}_{m=1}^{M} \abs{\mathscr{X}_m - \frac{1}{2 A_T(\theta)} \frac{\partial \mathcal{L}_\theta}{\partial \theta_m}} \ge \delta\Big) \le 2 M\ \exp\big(-\frac{\delta^2}{16 A}\big), \qquad A = \frac{T^2 \alpha^2 K}{2 \numtime L}.
     \end{align*}
     Since we want the probability to be less than $\varepsilon$, we require
     \begin{align*}
     \numtime L K \ge 8 T^2 \alpha^2 K^2 \delta^{-2} \log\big(\frac{2 M}{\varepsilon}) = \order{ (\alpha T K \delta^{-1})^2 \log\big(\frac{M}{\varepsilon})},
     \end{align*}
     where $N_t L K$ is the number of measurements required to estimate to $\partial \mathcal{L}_{\theta}/\partial \theta_m$. $L$ is the number of copies of the state required to compute the expectation value of $X_{k, j}$, \textit{for each $(k,j)$ pair}, so that the gradient of the loss function can be estimated to high precision with high probability, and there are $N_t K$ such pairs. \\
     
     We could simply pick $L = 1$, which is a special strategy. By part (i), we require $
     \numtime L K \ge \order{(\alpha T)^{3/2} K^{5/2} \delta^{-1/2}}.
$
     By combing the above results, we know that one at most needs
     \begin{align*}
     \max\Big\{ \order{(\frac{\alpha T K}{\delta})^2 \log\big(\frac{M}{\varepsilon})}, \order{(\alpha T)^{3/2} K^{5/2} \delta^{-1/2}}\Big\}.
     \end{align*}
    \end{proof}

\subsubsection{Bias estimate}

In below, we will need the following estimate:
\begin{lemma}
\label{lem::bias}
    Under Assumption~\ref{assume::norm_Hk_2}, and suppose we choose the midpoint scheme, one has
\begin{align*}
\ee \big[\frac{1}{2 A_T(\theta)}\nabla_\theta \loss_\theta \cdot (\frac{1}{2 A_T(\theta)} \nabla_\theta \loss_\theta - \mathscr{X})\big] \le \frac{1}{A_T(\theta)} \order{\frac{M K^4 T^4 \alpha^4}{\numtime^2}}.
\end{align*}
\end{lemma}

\begin{proof}
This is the error coming from numerical discretization only, 
\begin{align*}
&\ \ee \big[\frac{1}{2 A_T(\theta)} \nabla_\theta \loss_\theta \cdot (\frac{1}{2 A_T(\theta)} \nabla_\theta \loss_\theta - \mathscr{X})\big] = \frac{1}{2 A_T(\theta)} \nabla_\theta \loss_\theta \cdot \big(\frac{1}{2 A_T(\theta)} \nabla_\theta \loss_\theta - \ee\big[\mathscr{X}\big]\big) \\
\le &\  \norm{\frac{1}{2 A_T(\theta)} \nabla_\theta \loss_\theta}_{1} \times \norm{\frac{1}{2 A_T(\theta)} \nabla_\theta \loss_\theta - \ee\big[\mathscr{X}\big]}_{\infty} \\
\le &\ \frac{1}{2 A_T(\theta)}  M (K \alpha T) \times \frac{K \alpha (1+ 6 \alpha K)^2 T^3}{24 \numtime^2} = \frac{1}{A_T(\theta)} \order{\frac{M K^4 T^4 \alpha^4}{\numtime^2}}.
\end{align*}
The inequality mainly comes from the error of midpoint scheme; see e.g., Eq.~\eqref{eqn::error_midpoint}.
\end{proof}

\subsubsection{$L^2$ error estimate}

Next, we consider the $L^2$ error estimate
\begin{proposition}
\label{prop::midpoint_L2}
The error is bounded by
\begin{align*}
\ee\ \norm{\frac{1}{2 A_T(\theta)} \nabla_\theta \loss_\theta -  \mathscr{X}}_2^2 \lesssim \frac{M K T^2 \alpha^2}{\numtime L} + M (\frac{(K T\alpha)^6}{\numtime^4}).
\end{align*}
and hence to ensure $\delta$ in $L_2$ error, one generally requires at most
$$
\max\{\order{M K^2 T^2 \alpha^2 \delta^{-2}}, \order{M^{1/4} K^{5/2} T^{3/2} \alpha^{3/2} \delta^{-1/2}}\}
$$
quantum measurements.
\end{proposition}

\begin{proof}
\begin{align*}
\ee\ \norm{\frac{1}{2 A_T(\theta)} \nabla_\theta \loss_\theta -  \mathscr{X}}_2^2 \le 2 \ee\ \norm{\frac{1}{2 A_T(\theta)} \nabla_\theta \loss_\theta -  \ee \mathscr{X}}_2^2 + 2 \ee\ \norm{\ee[\mathscr{X}] -  \mathscr{X}}_2^2.
\end{align*}
The first term comes from numerical discretization and is thus 
\begin{align*}
\mathcal{O}\bigg({2 M \Big(K \frac{1}{24} \frac{T^2}{\numtime^2} T \alpha \alpha^2 K^2\Big)^2}\bigg) = \order{M (\frac{(K T\alpha)^6}{\numtime^4})}.
\end{align*}
The second term comes from sampling and is
\begin{align*}
 \lesssim&\ \sum_{m=1}^M \ee \Big[ \sum_{j=1}^{\numtime} \sum_{k=1}^K \Delta s  \frac{\partial f_k(s_j, \theta)}{\partial \theta_m} (X_{k,j} - \ee[X_{k,j}]) \big]^2  \\
=&\ \sum_{m=1}^M \ee \Big[ \sum_{j=1}^{\numtime} \sum_{k=1}^K \Delta s^2   (\frac{\partial f_k(s_j, \theta)}{\partial \theta_m})^2 \ee\big[(X_{k,j} - \ee[X_{k,j}])^2\big] \\
\lesssim &\ M \numtime K (\frac{T}{\numtime})^2 \alpha^2 \frac{1}{L} = \frac{M K T^2 \alpha^2}{\numtime L}.
\end{align*}
To ensure that the error is bounded by $\delta$, and minimizing the total sample cost $\numtime L K$, one obtains the above estimate. 

\end{proof}

\subsection{Optimization efficiency for time discretization case}

\begin{proposition}
\label{prop::ml::midpoint}
Under Assumptions~\ref{assume::norm_Hk_2}, \ref{assume::PL}, and \ref{assume::AT}, if we use constant learning rate, it is enough to use the following amount of samples to ensure Eq.~\eqref{eqn::sgd_error}
\begin{align}
\label{eqn::ml_cost_midpoint}
\widetilde{\mathcal{O}}\bigg(\max\bigg\{\frac{\Lipschitzmax K\mathsf{A}^{1/2} M^{1/2} (K T \alpha)^2}{\mu^{3/2} \Delta^{1/2}},  \frac{\Lipschitzmax \mathsf{A}^2 M (K T \alpha)^2}{\mu^2 \Delta} \bigg\}\bigg).
\end{align}
\end{proposition}

\begin{proof}
By Assumption~\ref{assume::norm_Hk_2}, it is easy to know that the exact value of the gradient $\frac{1}{2 A_T(\theta)} \frac{\partial \mathcal{L}_\theta}{\partial \theta_m}$ are uniformly bounded by $\alpha T K$, as well as any such approximations arising from quantum measurements, for each fixed index $m$. Let us denote the approximated gradient as
\begin{align*}
g := 2 A_T(\theta^{(k)}) \mathscr{X}.
\end{align*}

By the $\Lipschitzmax$-Lipschitz condition, 
\begin{align*}
&\ \ee \big[\mathcal{L}_{\theta^{(k+1)}}\big] \\
\le &\ \ee\big[\mathcal{L}_{\theta^{(k)}} - h \nabla \mathcal{L}_{\theta^{(k)}} \cdot g + \frac{\Lipschitzmax h^2}{2} \norm{g}_2^2\big] \\
=&\ \ee\big[\mathcal{L}_{\theta^{(k)}} - h \norm{\nabla \mathcal{L}_{\theta^{(k)}} }^2_2 - h \nabla \mathcal{L}_{\theta^{(k)}} \cdot (g - \nabla \mathcal{L}_{\theta^{(k)}}) + \frac{\Lipschitzmax h^2}{2} \norm{g}_2^2\big] \\
\le&\ \ee\big[\mathcal{L}_{\theta^{(k)}} - h \norm{\nabla \mathcal{L}_{\theta^{(k)}} }^2_2 - h \nabla \mathcal{L}_{\theta^{(k)}} \cdot (g - \nabla \mathcal{L}_{\theta^{(k)}}) + \Lipschitzmax h^2 (\norm{g - \nabla \loss_{\theta^{(k)}}}_2^2 + \norm{\nabla \loss_{\theta^{(k)}}}_2^2)  \big]\\
=&\ \ee\big[\mathcal{L}_{\theta^{(k)}}\big] - h (1 - \Lipschitzmax h) \ee \big[\norm{\nabla \mathcal{L}_{\theta^{(k)}} }^2_2\big] - h \ee\big[\nabla \mathcal{L}_{\theta^{(k)}} \cdot (g - \nabla \mathcal{L}_{\theta^{(k)}}) \big] + \Lipschitzmax h^2 \ee\big[\norm{g - \nabla \loss_{\theta^{(k)}}}_2^2\big].
\end{align*}
By choosing $\Lipschitzmax h\le \frac{1}{2}$, and by Assumption~\ref{assume::PL}, one has 
\begin{align*}
\ee \big[\mathcal{L}_{\theta^{(k+1)}} - \loss_{\min} \big] \le (1 - h \mu) &\ee \big[\mathcal{L}_{\theta^{(k)}} - \loss_{\min} \big] + \mathscr{E}_K
\end{align*}
where by Lemma~\ref{lem::bias} and \propref{prop::midpoint_L2},
\begin{align*}
& \mathscr{E}_K\lesssim  h \mathsf{A} \frac{M K^4 T^4 \alpha^4}{\numtime^2} + \Lipschitzmax h^2 \mathsf{A}^2 \big(\frac{M K T^2 \alpha^2}{\numtime L} + M (\frac{(K T\alpha)^6}{\numtime^4})\big).
\end{align*}
It is straightforward to show that
\begin{align*}
\ee \big[\mathcal{L}_{\theta^{(\Niter)}} - \loss_{\min} \big] \lesssim (1 - h \mu)^{\Niter} \ee \big[\mathcal{L}_{\theta^{(0)}} - \loss_{\min} \big] +  \mathsf{A} \frac{M K^4 T^4 \alpha^4}{\numtime^2\mu } +  \frac{\Lipschitzmax \mathsf{A}^2 h}{\mu}  \big(\frac{M K T^2 \alpha^2}{\numtime L}\big) +  \frac{\Lipschitzmax \mathsf{A}^2 h}{\mu} \big( M (\frac{(K T\alpha)^6}{\numtime^4})\big).
\end{align*}

To ensure that the error is bounded by $\Delta$, one requires
\begin{align*}
\Niter = \ordertilde{\frac{1}{h \mu}},
\end{align*}
and also 
\begin{align*}
\mathsf{A} \frac{M K^4 T^4 \alpha^4}{\numtime^2\mu } \lesssim \Delta \qquad 
\frac{\Lipschitzmax \mathsf{A}^2 h}{\mu}  \big(\frac{M K T^2 \alpha^2}{\numtime L}\big) \lesssim \Delta \qquad 
\frac{\Lipschitzmax \mathsf{A}^2 h}{\mu} \big( M (\frac{(K T\alpha)^6}{\numtime^4})\big) \lesssim \Delta.
\end{align*}
Then we require 
\begin{align}
\label{eqn::midpoint_sampletime}
\numtime \gtrsim \mathsf{A}^{1/2} \frac{M^{1/2} K^2 T^2 \alpha^2}{\mu^{1/2} \Delta^{1/2}},
\end{align}
and 
\begin{align*}
\frac{1}{h \mu} \gtrsim \mathcal{O}\bigg(\max\bigg\{\frac{\Lipschitzmax}{\mu},  \frac{\Lipschitzmax \mathsf{A}^2}{\mu^2 \Delta} \big(\frac{M K T^2 \alpha^2}{\numtime L}\big), \frac{\Lipschitzmax \mathsf{A}^2 M}{\mu^2 \Delta} \frac{(K T\alpha)^6}{\numtime^4}\bigg\}\bigg).
\end{align*}
In general, we may simply pick $L = 1$, so that the total number of quantum measurements is
\begin{align*}
\frac{1}{h \mu} \numtime K L \gtrsim &\ \mathcal{O}\bigg(\max\bigg\{\frac{\Lipschitzmax}{\mu} \numtime K,  \frac{\Lipschitzmax \mathsf{A}^2 M K^2 T^2 \alpha^2}{\mu^2 \Delta}, \frac{\Lipschitzmax \mathsf{A}^2 M K}{\mu^2 \Delta} \frac{(K T\alpha)^6}{\numtime^3}\bigg\}\bigg).
\end{align*}
By plugging the estimate of $\numtime$, the last term is
\begin{align*}
\frac{\Lipschitzmax \mathsf{A}^2 M K}{\mu^2 \Delta} \frac{(K T\alpha)^6}{\numtime^3} \lesssim \frac{\Lipschitzmax \mathsf{A}^2 M K}{\mu^2 \Delta} \frac{\mu^{3/2}\Delta^{3/2}}{\mathsf{A}^{3/2} M^{3/2}} = \frac{\Lipschitzmax \mathsf{A}^{1/2} K \Delta^{1/2}}{\mu^{1/2} M^{1/2}}.
\end{align*}
We will generally consider the region $\Delta \ll 1$, $M \gg 1$, so that this term is negligible with respect to the first term, so that the total quantum samples required typically behaves like \eqref{eqn::ml_cost_midpoint}.
\end{proof}

\subsection{Gradient estimate with sampling in time}

Besides, we can also estimate the gradient without time-discretization, but instead uses the Monte Carlo sampling in time:
\begin{align}
\label{eqn::mc_time}
\begin{aligned}
\frac{1}{2 A_T(\theta)} \frac{\partial \mathcal{L}_{\theta}}{\partial \theta_m} = &\ \int_{0}^{T} \tr\Big(\big(\sigma_Y \otimes \sum_{k=1}^K \frac{\partial f_k(t, \theta)}{\partial \theta_m} H_k \otimes \unit\big)\  \eta(s, \theta)\Big)\ ds \\
\approx &\ \frac{T}{\numtime} \sum_{j=1}^{\numtime} \sum_{k=1}^K \frac{\partial f_k(t_j, \theta)}{\partial \theta_m} \tr\Big(\big(\sigma_Y \otimes  H_k \otimes \unit\big)\  \eta(t_j, \theta)\Big)\ \\
\approx&\ \frac{T}{\numtime} \sum_{j=1}^{\numtime} \sum_{k=1}^K \frac{\partial f_k(t_j, \theta)}{\partial \theta_m} X_{k,j} \\
=:&\ \mathscr{X}_m,
\end{aligned}
\end{align}
where we first randomly generate time $t_j$ uniformly from $[0,T]$, and then for each index pair $(j, k)$, use $L$ samples to estimate $\tr\Big(\big(\sigma_Y \otimes  H_k \otimes \unit\big)\  \eta(t_j, \theta)\Big)$ whose random variable is denoted as $X_{k,j}$. Of course, it is possible to adopt importance sampling, but for the time being, let us consider the most straightforward way to sample the time. Unlike the above time-discretized case, this estimator is theoretically always unbiased. We will similarly denote the gradient approximation as $\mathscr{X}_m$, and the number of samples in time as $\numtime$ for simplicity of notations and the meaning of such gradient estimate (either via time-discretization or sampling in time) could be inferred from the context.

\subsubsection{$L^\infty$ error estimate}

For direct Monte Carlo scheme in time, we can similarly estimate the $L_\infty$ error in the probabilistic setting.

\begin{proposition}
\label{prop::mc_inf}
To ensure \eqref{eqn::prob_error} using Monte Carlo scheme \eqref{eqn::mc_time}, namely, to ensure that
\begin{align*}
\mathbb{P}\big(\norm{\frac{1}{2 A_T(\theta)} \nabla \loss_\theta - \mathscr{X}}_{\infty} \ge \delta \big) \le \varepsilon,
\end{align*}
it is sufficient to choose $\ordertilde{\alpha^2 T^2 K^3 \delta^{-2}}$ quantum measurements.
\end{proposition}

\begin{proof}
Since we used Monte Carlo sampling, we easily have $\frac{1}{2 A_T(\theta)} \nabla \loss_\theta = \ee\big[\mathscr{X}\big]$
\begin{align*}
&\mathbb{P}\big(\norm{\frac{1}{2 A_T(\theta)} \nabla \loss_\theta - \mathscr{X}}_{\infty} \ge \delta \big) \\
\le & \sum_{m=1}^M \mathbb{P}\big(\abs{\mathscr{X}_m  - \ee\big(\mathscr{X}_m\big)} \ge \delta \big) \\
\le & \sum_{m=1}^M e^{-\lambda \delta} \mathbb{E}\Big[e^{\lambda \abs{\mathscr{X}_m  - \ee\big(\mathscr{X}_m\big)}}\Big] \qquad \text{ (by Markov inequality)}\\
\le & \sum_{m=1}^M e^{-\lambda \delta} \mathbb{E}\Big[e^{\lambda \big({\mathscr{X}_m  - \ee\big(\mathscr{X}_m\big)}\big) }\Big] + e^{-\lambda \delta} \mathbb{E}\Big[e^{(-\lambda) \big({\mathscr{X}_m  - \ee\big(\mathscr{X}_m\big)}\big) }\Big].
\end{align*}

Next, we shall focus on one term with fixed $m$,
\begin{align*}
 & \mathbb{E}\Big[e^{\lambda \big({\mathscr{X}_m  - \ee\big(\mathscr{X}_m\big)}\big) }\Big]  \\
 =& \ \mathbb{E}\Big[\exp\Big( \frac{T\lambda}{\numtime} \sum_{j=1}^{\numtime} \sum_{k=1}^K \frac{\partial f_k(t_j, \theta)}{\partial \theta_m} X_{k,j}  - \ee\big(\frac{\partial f_k(t_j, \theta)}{\partial \theta_m} X_{k,j}\big)\Big)\Big]\\
 =&\ \prod_{j=1}^{\numtime}  \mathbb{E}\Big[\exp\Big( \frac{T\lambda}{\numtime}  \sum_{k=1}^K \big(\frac{\partial f_k(t_j, \theta)}{\partial \theta_m} X_{k,j}  - \ee\big(\frac{\partial f_k(t_j, \theta)}{\partial \theta_m} X_{k,j}\big)\big) \Big)\Big].
\end{align*}
The above holds due to the independence of sampling for time $t_j$. The above terms $\frac{\partial f_k(t_j, \theta)}{\partial \theta_m} X_{k,j}$ are generally dependent for pair $(j,k)$ due to the appearance of the same $t_j$ in different $k$. However, conditioned on $t_j$, the above terms becomes independent.
\begin{align*}
& \mathbb{E}\Big[e^{\lambda \big({\mathscr{X}_m  - \ee\big(\mathscr{X}_m\big)}\big) }\Big] \\
=&\ \prod_{j=1}^{\numtime}  \mathbb{E}\bigg[\mathbb{E}\Big[\exp\Big( \frac{T\lambda}{\numtime}  \sum_{k=1}^K \big(\frac{\partial f_k(t_j, \theta)}{\partial \theta_m} X_{k,j}  - \ee\big(\frac{\partial f_k(t_j, \theta)}{\partial \theta_m} X_{k,j}\big)\big) \Big) \big\rvert t_j \Big]\bigg] \\
=&\ \prod_{j=1}^{\numtime}  \mathbb{E}\left[\mathbb{E}\left[
\begin{aligned}
& e^{\frac{T\lambda}{\numtime}  \sum_{k=1}^K \big(\frac{\partial f_k(t_j, \theta)}{\partial \theta_m} X_{k,j}  - \ee\big(\frac{\partial f_k(t_j, \theta)}{\partial \theta_m} X_{k,j}\rvert t_j \big)} \\
& \hspace{2em} \cdot e^{\frac{T\lambda}{\numtime} \sum_{k=1}^K \ee\big(\frac{\partial f_k(t_j, \theta)}{\partial \theta_m} X_{k,j}\rvert t_j\big) - \ee\big(\frac{\partial f_k(t_j, \theta)}{\partial \theta_m} X_{k,j}\big)}
\end{aligned}
 \bigg\rvert t_j \right]\right]\\
 =&\ \prod_{j=1}^{\numtime}  \mathbb{E}\left[e^{\frac{T\lambda}{\numtime} \sum_{k=1}^K \ee\big(\frac{\partial f_k(t_j, \theta)}{\partial \theta_m} X_{k,j}\rvert t_j\big) - \ee\big(\frac{\partial f_k(t_j, \theta)}{\partial \theta_m} X_{k,j}\big)}
\mathbb{E}\left[
e^{\frac{T\lambda}{\numtime}  \sum_{k=1}^K \big(\frac{\partial f_k(t_j, \theta)}{\partial \theta_m} X_{k,j}  - \ee\big(\frac{\partial f_k(t_j, \theta)}{\partial \theta_m} X_{k,j}\rvert t_j \big)}
 \bigg\rvert t_j \right]\right].
\end{align*}
Conditioned on $t_j$, if we estimate the term inside 
\begin{align*}
&\ \mathbb{E}\Big[\exp\Big( \frac{T\lambda}{\numtime}  \sum_{k=1}^K \big(\frac{\partial f_k(t_j, \theta)}{\partial \theta_m} X_{k,j}  - \ee\big(\frac{\partial f_k(t_j, \theta)}{\partial \theta_m} X_{k,j}\rvert t_j \big)  \big) \Big) \big\rvert t_j \Big] \\
= &\ \prod_{k=1}^K \mathbb{E}\Big[\exp\Big( \frac{T\lambda}{\numtime} \big(\frac{\partial f_k(t_j, \theta)}{\partial \theta_m} X_{k,j}  - \ee\big(\frac{\partial f_k(t_j, \theta)}{\partial \theta_m} X_{k,j}\rvert t_j \big)  \big) \Big) \big\rvert t_j \Big] \\
= &\ \prod_{k=1}^K \prod_{\ell=1}^L \mathbb{E}\Big[\exp\Big( \frac{T\lambda}{\numtime L } \big(\frac{\partial f_k(t_j, \theta)}{\partial \theta_m} Y_{k,j}^{(\ell)}  - \ee\big(\frac{\partial f_k(t_j, \theta)}{\partial \theta_m} Y_{k,j}^{(\ell)}\rvert t_j \big)  \big) \Big) \big\rvert t_j \Big] \\
=&\ \prod_{k=1}^K \prod_{\ell=1}^L \exp\big(\frac{T^2 \lambda^2\alpha^2}{2\numtime^2 L^2}\big) = \exp\big(\frac{T^2 \lambda^2 \alpha^2 K}{2 \numtime^2 L}\big).
\end{align*}
Since we use $L$ samples to approximate $X_{k,j} = \frac{1}{L} \sum_{\ell=1}^L Y_{k,j}^{(\ell)}$, and recall that by Assumption~\ref{assume::norm_Hk}, each $Y$ is supposed to approximate $\tr\big(\sigma_Y \otimes H_k \otimes \unit \eta(t_j, \theta)\big)$ which comes from $[-1,1]$, we have the above by applying Hoeffding’s lemma.

By plugging it back,
\begin{align*}
&\ \mathbb{E}\Big[e^{\lambda \big({\mathscr{X}_m  - \ee(\mathscr{X}_m)}\big) }\Big] \\
\le&\ \prod_{j=1}^{\numtime} \ee\bigg[\exp\big(\frac{T^2 \lambda^2 \alpha^2 K}{2 \numtime^2 L}\big) \exp\Big(\frac{T \lambda}{\numtime} \sum_{k=1}^{\numop} \ee\big(\frac{\partial f_k(t_j, \theta)}{\partial \theta_m} X_{k,j}\rvert t_j\big) - \ee\big(\frac{\partial f_k(t_j, \theta)}{\partial \theta_m} X_{k,j}\big)\Big) \bigg] \\
\le&\ \prod_{j=1}^{\numtime} \ee\bigg[\exp\big(\frac{T^2 \lambda^2 \alpha^2 K}{2 \numtime^2 L}\big) \exp\Big(\frac{2 T^2\lambda^2}{\numtime^2} (\alpha K)^2\Big) \bigg] \\
=&\ \exp\big(\frac{T^2 \lambda^2 \alpha^2 K}{2 \numtime L} + \frac{2 T^2\lambda^2 \alpha^2 K^2}{\numtime} \big),
\end{align*}
where to get the second inequality, we applied Hoeffding’s lemma again due to the uniform bound $\abs{\sum_{k=1}^K \ee\big(\frac{\partial f_k(t_j, \theta)}{\partial \theta_m} X_{k,j}\rvert t_j\big) - \ee\big(\frac{\partial f_k(t_j, \theta)}{\partial \theta_m} X_{k,j}\big)}\le 2 \alpha K$. \\

By combing previous results together,
\begin{align*}
&\mathbb{P}\big(\norm{\frac{1}{2 A_T(\theta)} \nabla \loss_\theta - \mathscr{X}}_{\infty} \ge \delta \big) \le 2 M e^{-\lambda \delta} \exp\big(\frac{T^2 \lambda^2 \alpha^2 K}{2 \numtime L} + \frac{2 T^2\lambda^2 \alpha^2 K^2}{\numtime} \big).
\end{align*}
Since $\lambda > 0$ is arbitrary, we can minimizing the upper bound over all $\lambda\ge 0$, and obtain
\begin{align*}
\mathbb{P}\big(\norm{\frac{1}{2 A_T(\theta)} \nabla \loss_\theta - \mathscr{X}}_{\infty} \ge \delta \big)  \le 2 M e^{-\frac{1}{4} \frac{\delta^2}{A}}, \qquad A = \frac{T^2 \alpha^2 K}{2 \numtime L} + \frac{2 T^2 \alpha^2 K^2}{\numtime}.
\end{align*}
To ensure that the error is bounded by $\varepsilon$, one needs
\begin{align*}
\frac{T^2 \alpha^2 K}{\numtime L} & \le \order{\frac{\delta^2}{\log(M/\varepsilon)}}\qquad 
\frac{T^2 \alpha^2 K^2}{\numtime} \le \order{\frac{\delta^2}{\log(M/\varepsilon)}}.
\end{align*}
This leads into 
\begin{align*}
\numtime \ge \order{T^2 \alpha^2 K^2 \delta^{-2} \log\big(M/\varepsilon)}, \qquad L \ge \order{{K}^{-1}}.
\end{align*}
Since the sample size $L \ge 1$, overall, we at most need sample size
$\numtime K L \ge \order{T^2 \alpha^2 K^3 \delta^{-2} \log\big(M/\varepsilon)}$.\end{proof}

\subsubsection{$L^2$ error estimate}

\begin{lemma}
Under Assumption~\ref{assume::norm_Hk}, the variance for the Monte Carlo estimator
\begin{align*}
\text{Var}(\mathscr{X}_m) \le \frac{\alpha^2 T^2}{\numtime} \big(\frac{K}{L} + K^2).
\end{align*}
\end{lemma}

\begin{proof}
Recall that for a fixed $m$, 
\begin{align*}
\mathscr{X}_m = \frac{T}{\numtime} \sum_{j=1}^{\numtime} \sum_{k=1}^K \frac{\partial f_k(t_j, \theta)}{\partial \theta_m} X_{k,j}, \qquad \text{ and } \qquad \ee \mathscr{X}_m = \frac{1}{2 A_T(\theta)} \frac{\partial \mathcal{L}_{\theta}}{\partial \theta_m}. 
\end{align*}
Then since terms associated with each $j$ are independent, we know that 
\begin{align*}
\text{Var}(\mathscr{X}_m) = \frac{T^2}{\numtime} \text{Var}\big(\sum_{k=1}^\numop \frac{\partial f_k(t,\theta)}{\partial \theta_m} X_{k, t}\big),
\end{align*}
where $t\sim \text{Uniform}(0,T)$ and $X_{k,t}$ is a random variable representing the quantum measurement outcomes of $\tr\Big(\big(\sigma_Y \otimes  H_k \otimes \unit\big)\  \eta(t, \theta)\Big)$ using independent $L$ samples, conditioned on the value of $t$. By the law of total variance,
\begin{align*}
\text{Var}(\mathscr{X}_m) =&\ \frac{T^2}{\numtime} \bigg(\ee\Big[ \text{Var}\big(\sum_{k=1}^\numop \frac{\partial f_k(t,\theta)}{\partial \theta_m} X_{k, t} \mid t \Big)\Big] + \text{Var}\Big(\ee\Big[ \sum_{k=1}^{\numop} \frac{\partial f_k(t,\theta)}{\partial \theta_m} X_{k, t} \mid t \Big]\Big) \bigg).
\end{align*}

For the first term, conditioned on $t$, $\frac{\partial f_k(t,\theta)}{\partial \theta_m} X_{k, t}$ are independent random variables, so that
\begin{align*}
\ee\Big[ \text{Var}\big(\sum_{k=1}^\numop \frac{\partial f_k(t,\theta)}{\partial \theta_m} X_{k, t} \mid t \Big)\Big] =&\ \ee\Big[ \sum_{k=1}^\numop (\frac{\partial f_k(t,\theta)}{\partial \theta_m})^2 \text{Var}\big(X_{k, t} \mid t \Big)\Big]
\le\  \ee\Big[ \sum_{k=1}^\numop \alpha^2 \frac{1}{L} \Big] = \frac{\alpha^2 K}{L}.
\end{align*}
To get the inequality in the last line, we used the fact that $X_{k,t}$ is a random variable representing the quantum measurement outcomes of $\tr\Big(\big(\sigma_Y \otimes  H_k \otimes \unit\big)\  \eta(t, \theta)\Big)$ using independent $L$ samples, and also the Assumption~\ref{assume::norm_Hk}. This quantity comes from the stochastic uncertainty of quantum measurement conditioned on the time.\\

As for the second term,
\begin{align*}
\text{Var}\Big(\ee\Big[ \sum_{k=1}^{\numop} \frac{\partial f_k(t,\theta)}{\partial \theta_m} X_{k, t} \mid t \Big]\Big) &= \text{Var}\Big( \sum_{k=1}^{\numop} \frac{\partial f_k(t,\theta)}{\partial \theta_m} \ee\Big[X_{k, t} \mid t \Big]\Big) \\
&= \text{Var}\Big( \sum_{k=1}^{\numop} \frac{\partial f_k(t,\theta)}{\partial \theta_m} \tr\big(\sigma_Y\otimes H_k\otimes \unit \eta(t, \theta)\big) \Big) \\
&\le \alpha^2 K^2.
\end{align*}
This term characterized the stochastic uncertainty from choosing $t$.\\

Therefore, overall, one has
\begin{align*}
\text{Var}(\mathscr{X}_m) \le \frac{T^2}{\numtime} \big(\frac{\alpha^2 K}{L} + \alpha^2 K^2 \big) = \frac{\alpha^2 T^2}{\numtime} \big(\frac{K}{L} + K^2).
\end{align*}

\end{proof}

\begin{proposition}
\label{prop::mc_L2}
Under the Assumption~\ref{assume::norm_Hk}, to ensure the $L_2$ error for gradient estimate using sampling in time to be bounded by $\delta$, it is enough to choose
\begin{align}
\label{eqn::mc_L2}
\order{\alpha^2 T^2 K^3 \delta^{-2} M}
\end{align}
amount of quantum measurements.
\end{proposition}

\begin{proof}
By the above lemma, one has
\begin{align}
\label{eqn::sum_var}
\ee \norm{\mathscr{X} - \frac{1}{2 A_T(\theta)} \nabla \mathcal{L}_{\theta}}_2^2 = \sum_{m=1}^M \text{Var}(\mathscr{X}_m) \le \frac{\alpha^2 T^2 M}{\numtime}\big(\frac{K}{L} + K^2\big).
\end{align}
Since this quantity needs to be bounded by $\delta^2$, one need
\begin{align*}
\numtime \ge \frac{\alpha^2 T^2 M K }{\delta^2} \big(\frac{1}{L} + K).
\end{align*}
The total number of samples one need is 
\begin{align*}
\numtime K L \ge \frac{\alpha^2 T^2 M K^2 }{\delta^2} \big(1 + KL). 
\end{align*}
Since $L\ge 1$ could be chosen freely, the optimal choice is to choose $L = 1$ and thus
the total number of samples required scales like Eq.~\eqref{eqn::mc_L2}.
\end{proof}

\subsection{Optimization efficiency for sampling in time}

\begin{proposition}
\label{prop::mc_train}
Under Assumptions~\ref{assume::norm_Hk}, \ref{assume::PL}, \ref{assume::AT}, we use batch-SGD with uniform time step. In order to ensure \eqref{eqn::sgd_error}, the total number of quantum measurements required is at most
\begin{align*}
\max\Big\{\ordertilde{\frac{\Lipschitzmax \mathsf{A}^2 \alpha^2 T^2 M K^3}{\mu^2 \Delta}}, \ordertilde{\frac{K}{\mu \min\{\mu, \Lipschitzmax\}}}\Big\}.\end{align*}
In the limit of small $\Delta$, the complexity is
\begin{align*}
\ordertilde{\Lipschitzmax \mathsf{A}^2 \alpha^2 \mu^{-2} T^2 K^3 M \Delta^{-1}}.
\end{align*}
\end{proposition}

\begin{proof}
Let us denote the approximated gradient as
$
g := 2 A_T(\theta^{(k)}) \mathscr{X}.
$
By the $L$-Lipschitz condition, 
\begin{align*}
\ee \big[\mathcal{L}_{\theta^{(k+1)}}\big] \le &\ \ee\big[\mathcal{L}_{\theta^{(k)}} - h \nabla \mathcal{L}_{\theta^{(k)}} \cdot g + \frac{\Lipschitzmax h^2}{2} \norm{g}_2^2\big] \\
=&\ \ee\big[\mathcal{L}_{\theta^{(k)}} - h \norm{\nabla \mathcal{L}_{\theta^{(k)}} }^2_2 - h \nabla \mathcal{L}_{\theta^{(k)}} \cdot (g - \nabla \mathcal{L}_{\theta^{(k)}}) + \frac{\Lipschitzmax h^2}{2} \norm{g}_2^2\big] \\
=&\ \ee\big[\mathcal{L}_{\theta^{(k)}}\big] - h \ee\big[\norm{\nabla \mathcal{L}_{\theta^{(k)}} }^2_2\big] + 0 + \frac{\Lipschitzmax h^2}{2} \ee\big[\sum_{m} \text{Var}(g_m\mid \theta^{(k)})\big] + \frac{\Lipschitzmax h^2}{2} \ee\big[\norm{\nabla \loss_{\theta^{(k)}}}^2\big] \\
=&\ \ee\big[\mathcal{L}_{\theta^{(k)}}\big] - h (1 - \Lipschitzmax h/2) \ee\big[\norm{\nabla \mathcal{L}_{\theta^{(k)}} }^2_2\big] + \frac{\Lipschitzmax h^2}{2} \sum_{m} \ee\big[\text{Var}(g_m \mid \theta^{(k)})\big] \\
\myle{\eqref{eqn::sum_var}}&\ \ee\big[\mathcal{L}_{\theta^{(k)}}\big] - h (1 - \Lipschitzmax h/2) \ee\big[\norm{\nabla \mathcal{L}_{\theta^{(k)}} }^2_2\big] + \frac{\Lipschitzmax h^2}{2} 4 \mathsf{A}^2 \frac{\alpha^2 T^2 M}{\numtime}\big(\frac{K}{L} + K^2).
\end{align*}
In the third line above, we used the fact that $g$ is an unbiased estimator of the gradient.
We assume that $\Lipschitzmax h\le 1$, and by $\mu$-PL condition we have
\begin{align*}
\ee \big[\mathcal{L}_{\theta^{(k+1)}} - \loss_{\min}\big] &\le \ee\big[\mathcal{L}_{\theta^{(k)}} - \loss_{\min}\big] - \frac{h}{2} \ee\big[\norm{\nabla \mathcal{L}_{\theta^{(k)}} }^2_2\big] +  \frac{2 \Lipschitzmax h^2 \mathsf{A}^2 \alpha^2 T^2 M}{\numtime}\big(\frac{K}{L} + K^2\big) \\
&\myle{\eqref{eqn::PL}} (1 - h \mu) \ee\big[\mathcal{L}_{\theta^{(k)}} - \loss_{\min}\big] +  
 \frac{2 \Lipschitzmax h^2 \mathsf{A}^2 \alpha^2 T^2 M}{\numtime}\big(\frac{K}{L} + K^2\big).
\end{align*}

For a general training iteration $\Niter$, by choosing $h \mu < 1$, 
\begin{align*}
\ee\big[\loss_{\theta^{(\Niter)}} - \loss_{\min} \big] \le&\  (1 - h \mu)^{\Niter} \ee\big[\mathcal{L}_{\theta^{(0)}} - \loss_{\min} \big] + \frac{1}{h \mu} \cdot  \frac{2 \Lipschitzmax h^2 \mathsf{A}^2 \alpha^2 T^2 M}{\numtime}\big(\frac{K}{L} + K^2\big) \\
=&\ (1 - h \mu)^{\Niter} \ee\big[\mathcal{L}_{\theta^{(0)}} - \loss_{\min} \big] +   \frac{2 \Lipschitzmax h \mathsf{A}^2 \alpha^2 T^2 M}{\mu \numtime}\big(\frac{K}{L} + K^2\big).
\end{align*}
In order for the error to be bounded by $\Delta$, one needs 
\begin{align*}
\frac{1}{h} =&\ \max\{\order{\frac{\Lipschitzmax \mathsf{A}^2 \alpha^2 T^2 M}{\mu \numtime \Delta}\big(\frac{K}{L} + K^2\big)}, \frac{1}{\Lipschitzmax}, \frac{1}{\mu} \}, \\
\Niter =&\ \order{\frac{1}{h\mu} \log(\frac{\ee\big[\mathcal{L}_{\theta^{(0)}} - \loss_{\min} \big]}{\Delta})} = \ordertilde{\frac{1}{h\mu}} = \max\Big\{\ordertilde{\frac{\Lipschitzmax \mathsf{A}^2 \alpha^2 T^2 M}{\mu^2 \Delta \numtime}\big(\frac{K}{L}+K^2\big)}, \ordertilde{\frac{1}{\mu^2}}, \ordertilde{\frac{1}{\mu \Lipschitzmax}}\Big\}.
\end{align*}
Hence, the total number of quantum samples required is at most
\begin{align*}
\Niter \times \big(\numtime K L\big) =  \max\Big\{\ordertilde{\frac{\Lipschitzmax \mathsf{A}^2 \alpha^2 T^2 M K}{\mu^2 \Delta}\big(K+K^2 L\big)}, \ordertilde{\frac{\numtime K L}{\mu \min\{\mu, \Lipschitzmax\}}}\Big\}.
\end{align*}
Since we want to minimize the total cost, we may as well choose $\numtime = 1$, $L = 1$, so that the total cost is given as in the proposition.
\end{proof}

\section{Proofs for \secref{subsec::landscape}}

\emph{Proof of Lemma~\ref{lemma::Hessian_QNODE}}:\\

{\noindent {\bf Step 1:}} It is sufficient to prove the following: For arbitrary $\nu\in \Real^M$, 
\begin{align}
\label{eqn::hess}
\begin{aligned}
\nu^\top \big(\nabla^2 \loss_\theta\rvert_{\theta^\star}\big) \nu=&\ 2\int_{0}^T \int_{0}^T\ \bra{\psi(t)} A(t) U_\theta(s,t) A(s) \ket{\psi(s)} - \langle A(t) \rangle_{\psi(t)} \langle A(s) \rangle_{\psi(s)} ds\ dt, \\
\end{aligned}
\end{align}
where we assume that $\rho_0$ is a pure state with state $\psi(0)$, adopt the notation $\rho(s) = \ketbra{\psi(s)}$, and define
\begin{align}
\label{eqn::Ar}
A(r) &:= \sum_{\ell=1}^M \nu_\ell \partial_{\theta_\ell} H(r) \myeq{\eqref{eqn::H_ansatz}} \sum_{\ell=1}^M \sum_{k=1}^K \nu_\ell \partial_{\theta_\ell} f_k(r) H_k =  \sum_{k=1}^K \big(\nu^\top \nabla_\theta f_k(r)\big) H_k.
 \end{align}
 By plugging the expression of $A$, one has
\begin{align*}
&\ \nu^\top \big(\nabla^2 \loss_\theta\rvert_{\theta^\star}\big) \nu \\
=&\ 2 \sum_{k, k'=1}^K \int_{0}^T \int_{0}^T \nu^\top \nabla_\theta f_k(t) \nu^\top \nabla_\theta f_{k'}(s) \big(\bra{\psi(t)} H_{k} U_\theta(s,t) H_{k'} \ket{\psi(s)} - \langle H_k \rangle_{\psi(t)} \langle H_{k'}\rangle_{\psi(s)} \big)ds dt \\
=&\ {\nu}^\top \bigg(2 \sum_{k, k'} \int_{0}^T \int_{0}^T \nabla_\theta f_k(t) \nabla_\theta f_{k'}^\top(s) G_{k,k'}(t,s) ds dt\bigg) {\nu}.
\end{align*}
 This gives the expression of Lemma~\ref{lemma::Hessian_QNODE}.\\
 
{\noindent {\bf Step 2:}} Let us begin by fixing index $j, k \in \{1, 2, \cdots, M\}$, and the problem set up tells us that 
\begin{align*}
\partial_t \rho(t) &= -i \comm{H(t)}{\rho(t)}\\
\partial_t (\partial_{\theta_j}\rho)(t) &=  -i \comm{\partial_{\theta_j} H(t)}{\rho(t)} -i \comm{H(t)}{\partial_{\theta_j} \rho(t)} \\
\partial_t (\partial_{\theta_j, \theta_k}\rho)(t) &=  -i \comm{\partial_{\theta_j, \theta_k} H(t)}{\rho(t)} -i \comm{\partial_{\theta_j} H(t)}{\partial_{\theta_k} \rho(t)} -i \comm{\partial_{\theta_k} H(t)}{\partial_{\theta_j} \rho(t)} -i \comm{H(t)}{\partial_{\theta_j, \theta_k} \rho(t)},
\end{align*}
By Durhamel's principle and the initial condition $\partial_{\theta_j} \rho(0) = \partial_{\theta_j, \theta_k} \rho(0) = 0$, one has
\begin{align}
\label{eqn::rho_k}
\partial_{\theta_j} \rho(t) &= \int_{0}^t U_\theta(s,t)  (-i) \comm{\partial_{\theta_j} H(s)}{\rho(s)} U_\theta(s,t)^\dagger\ ds.
\end{align}
and
\begin{align*}
(\partial_{\theta_j, \theta_k}\rho)(T) &= \int_{0}^T U_\theta(t, T) \Big(-i \comm{\partial_{\theta_j, \theta_k} H(t)}{\rho(t)} -i \comm{\partial_{\theta_j} H(t)}{\partial_{\theta_k} \rho(t)} -i \comm{\partial_{\theta_k} H(t)}{\partial_{\theta_j} \rho(t)} \Big) U_\theta(t,T)^\dagger\ dt.
\end{align*}

If we evaluate gradient at $\theta = \theta^\star$, 
\begin{align*}
\frac{\loss_\theta}{\partial \theta_j}\Big\rvert_{\theta = \theta^\star} =& - \tr\big(\sigma(T) \partial_{\theta_j} \rho(t)\big) = - \tr\big(\sigma(T) \int_{0}^T U_\theta(s,T)  (-i) \comm{\partial_{\theta_j} H(s)}{\rho(s)} U_\theta(s,T)^\dagger\ ds \big) \\
=&\  i \int_{0}^T \tr\big(\sigma(T) U_\theta(s,T) \big(\partial_{\theta_j} H(s){\rho(s)} - \rho(s) \partial_{\theta_j} H(s)\big) U_\theta(s,T)^\dagger \big) \ ds\\
=&\ i \int_{0}^T \tr\Big(\rho(s) \big(\partial_{\theta_j} H(s){\rho(s)} - \rho(s) \partial_{\theta_j} H(s)\big) \Big)\ ds = 0.
\end{align*}
This is expected.
Next, we perform similar calculations to the second order derivative at $\theta = \theta^\star$ and obtain:
\begin{align*}
&\ \frac{\loss_\theta}{\partial \theta_j\theta_k} \Big\rvert_{\theta = \theta^\star} = -\tr\big(\sigma(T) \partial_{\theta_j, \theta_k} \rho(T)\big) \\
=&\ - \tr\bigg(\sigma(T) \int_{0}^T U_\theta(t, T) \Big(-i \comm{\partial_{\theta_j, \theta_k} H(t)}{\rho(t)} -i \comm{\partial_{\theta_j} H(t)}{\partial_{\theta_k} \rho(t)} -i \comm{\partial_{\theta_k} H(t)}{\partial_{\theta_j} \rho(t)} \Big) U_\theta(t,T)^\dagger\ dt \bigg)\\
=&\ i\ \tr\bigg(\sigma(T) \int_{0}^T U_\theta(t, T) \Big(\comm{\partial_{\theta_j} H(t)}{\partial_{\theta_k} \rho(t)} + \comm{\partial_{\theta_k} H(t)}{\partial_{\theta_j} \rho(t)} \Big) U_\theta(t,T)^\dagger\ dt \bigg),
\end{align*}
for the same reason as above. Next, we plug the expression of $\partial_{\theta_k} \rho(t)$ from Eq.~\eqref{eqn::rho_k}, we have
\begin{align*}
&\ \frac{\loss_\theta}{\partial \theta_j\theta_k}\Big\rvert_{\theta = \theta^\star} \\
=&\ i\ \tr\bigg(\sigma(T) \int_{0}^T U_\theta(t, T) \Big(\comm{\partial_{\theta_j} H(t)}{\partial_{\theta_k} \rho(t)} + \comm{\partial_{\theta_k} H(t)}{\partial_{\theta_j} \rho(t)} \Big) U_\theta(t,T)^\dagger\ dt \bigg)\\
=&\  i\ \tr\bigg(\sigma(T) \int_{0}^T U_\theta(t, T) \Big({\partial_{\theta_j} H(t)}{\partial_{\theta_k} \rho(t)} + {\partial_{\theta_k} H(t)}{\partial_{\theta_j} \rho(t)} \Big) U_\theta(t,T)^\dagger\ dt \bigg) + c.c.\\
=&\ \tr\bigg(\sigma(T) \int_{0}^T \int_{0}^t U_\theta(t, T) \Big({\partial_{\theta_j} H(t)}{U_\theta(s,t) \comm{\partial_{\theta_k} H(s)}{\rho(s)} U_\theta(s,t)^\dagger}  \\
&\qquad\qquad \qquad + {\partial_{\theta_k} H(t)}{U_\theta(s,t) \comm{\partial_{\theta_j} H(s)}{\rho(s)} U_\theta(s,t)^\dagger} \Big) U_\theta(t,T)^\dagger\ ds\ dt \bigg) + c.c.\\
=&\ \int_{0}^T \int_{0}^t\ \tr\bigg(\rho(t) {\partial_{\theta_j} H(t)}{U_\theta(s,t) \comm{\partial_{\theta_k} H(s)}{\rho(s)} U_\theta(s,t)^\dagger}  \\
&\qquad\qquad \qquad + \rho(t)  {\partial_{\theta_k} H(t)}{U_\theta(s,t) \comm{\partial_{\theta_j} H(s)}{\rho(s)} U_\theta(s,t)^\dagger}) \bigg)\ ds\ dt + c.c.\\
=&\ \int_{0}^T \int_{0}^T\ \bra{\psi(t)} \partial_{\theta_j} H(t) U_\theta(s,t) \partial_{\theta_k} H(s) \ket{\psi(s)} - \langle \partial_{\theta_j} H(t)\rangle_{\psi(t)} \langle \partial_{\theta_k} H(s)\rangle_{\psi(s)} \\
&\qquad \qquad \qquad + \bra{\psi(t)} \partial_{\theta_k} H(t) U_\theta(s,t) \partial_{\theta_j} H(s) \ket{\psi(s)} - \langle \partial_{\theta_k} H(t)\rangle_{\psi(t)} \langle \partial_{\theta_j} H(s)\rangle_{\psi(s)} ds\ dt.
\end{align*}
Since parameters are assumed to take real values,
for any given perturbation vector $\nu\in \Real^M$, by the notation of $A(r)$ in \eqref{eqn::Ar}, one has
\begin{align*}
\nu^\top \big(\nabla^2 \loss_\theta\rvert_{\theta^\star}\big) \nu =&\ 2\int_{0}^T \int_{0}^T\ \bra{\psi(t)} A(t) U_\theta(s,t) A(s) \ket{\psi(s)} - \langle A(t) \rangle_{\psi(t)} \langle A(s) \rangle_{\psi(s)} ds\ dt.
\end{align*}
Thus we arrive at the Eq.~\eqref{eqn::hess}. 

\section{Supplemental details for \secref{sec::example}}

\subsection{Hamiltonians used in numerical experiments}
\label{app::Hamil}

{\noindent \emph{Hamiltonian of $H_2$ molecule:}}\\

During numerical experiments, we use the following Hamiltonian representation for hydrogen molecule from \cite{kandala_hardware-efficient_2017}:
\begin{align}
\label{eqn::hydrogen}
\Hhydro = 0.397936\ \sigma_Z^{(1)} + 0.397936\ \sigma_Z^{(2)} + 0.011280\ \sigma_Z^{(1)} \otimes \sigma_Z^{(2)} + 0.180931\ \sigma_X^{(1)} \otimes \sigma_X^{(2)}.
\end{align}
The Hamiltonian has been mapped into the representation of two qubits.\\

{\noindent \emph{Quantum Ising chain:}}\\

We consider the Hamiltonian learning problem from \cite{wiebe_2014_hamiltonian}. The parameterised Hamiltonian is chosen as 1D Ising chain with $L$ sites herein:
\begin{align}
\label{eqn::Ising0}
H = \sum_{i = 1}^{L-1} x_{i} \sigma_Z^{(i)} \otimes \sigma_Z^{(i+1)}.
\end{align}
The parameters $x_i$ are randomly initiated from $[-0.5, -0.08] \cup [0.08, 0.5]$ similarly to \cite{wiebe_2014_hamiltonian}. The initial state is chosen as $\ket{+}^{\otimes L}$ followed by possible random rotation (see Table~\ref{table::hyperparameter}).\\

{\noindent \emph{Time-dependent quantum Ising chain:}}\\

We also consider the time-dependent case:
\begin{align}
\label{eqn::Ising_td}
H =  \sum_{i = 1}^{L-1} x_{i} \sigma_Z^{(i)} \otimes \sigma_Z^{(i+1)} + f(t) \sum_{i=1}^{L} \sigma_X^{(i)},
\end{align}
where $f(t) = \sin( \pi t)$ for experiments in \secref{sec::example}. For ansatz, we consider the same structure as the above, and train the model to learn $x_i$ and $f(t)$. For the function $f(t)$, we use a two-layer network network with $\sin$ function as the activation function, namely, $f(t) \approx f_\theta(t) = \sum_{k=1}^{m} w^{(2)}_i \sin( w^{(1)}_k t + b_k) + b$ where the parameter $\theta$ is simply the collection $\{w_k^{(2)}\}_{k=1}^{m}$, $\{w^{(1)}_k\}_{k=1}^{m}$, $\{b_k\}_{k=1}^m$, and $b$. For the above experiments, we simply pick $m = 2$ for simplicity.

\subsection{Decomposition of Pauli matrices in adjoint states}
\label{appendix::pauli}

It is easy to validate the following decompositions:
\begin{align*}
\sigma_Z &= \ketbra{0} - \ketbra{1} , \\
\sigma_X &= \ketbra{+} - \ketbra{-}, \qquad \ket{\pm} = \frac{1}{\sqrt{2}}\big(\ket{0}\pm \ket{1}\big) , \\ 
\sigma_Y &= \ketbra{\psi_1} - \ketbra{\psi_2}, \qquad \ket{\psi_1} = \frac{1}{\sqrt{2}} \big(\ket{1} - i \ket{0}\big), \qquad \ket{\psi_2} = \frac{1}{\sqrt{2}} \big(\ket{1} + i \ket{0}\big).
\end{align*}
All these pure states are easy to prepare. For two-body Pauli matrix $\sigma_X \otimes \sigma_Y$ as an example, one can decompose it into the summation of $4$ terms:
\begin{align*}
\sigma_X \otimes \sigma_Y =& \big(\ketbra{+} \otimes \ketbra{\psi_1}\big) - \big(\ketbra{+}\otimes \ketbra{\psi_2}\big) \\
&- \big(\ketbra{-} \otimes \ketbra{\psi_1}\big) + \big(\ketbra{-} \otimes \ketbra{\psi_2}\big).
\end{align*}
Therefore, one can run the optimisation algorithm for the Hamiltonian learning using loss function in Eq.~\eqref{eqn::loss_observable}, for all observables constructed via the tensor product of local Pauli matrices.

\subsection{Supplementary training details}
\label{app::trainning_details}

{\noindent \emph{Optimisation - }} The training hyper-parameters are summarized in the following Table~\ref{table::hyperparameter}. For the toy model of learning an unknown state, we simply use the stochastic gradient descent method; for other problems, we use the Adam method for the parameter update (together with batch samples). For the integral in Eq.~\eqref{eq:mixedgradient2}, we always discretise it using the Trapezoidal rule. For the mixed state simulation for the loss function in Eq.~\eqref{eqn::loss_observable}, the choice of  observables as well as whether to randomize the observables play a non-trivial role in the efficiency and training accuracy, which is also a challenging topic in quantum state tomography. 
A more comprehensive study on the optimal way to choose the classical training procedure will be interesting for future endeavors. \\

{\noindent \emph{Test error - }} In the above experiments in \secref{sec::example}, we always report $1-\text{Fidelity}$ which is defined as follows:
\begin{align}
\label{eqn::test_error}
1 - \frac{1}{\sample}\sum_{i=1}^{\sample}\ \abs{\bra{\phi_i} U^\dagger_\theta(0, T_i) U(0, T_i) \ket{\phi_i}}^2, 
\end{align}
where $T_i$ are possibly randomly generated (see Table~\ref{table::hyperparameter}), the initial state $\rho_i = \ketbra{\phi_i}$ are generated randomly by applying random rotations.
The unitary $U(0, T)$ is the exact unitary evolution of the black-box Hamiltonian from time zero to $T$ and $U_\theta(0, T)$ is the unitary for parameterised Hamiltonian. 
In all experiments, we choose $\sample = 50$.

\begin{table}[h!]
\renewcommand{\arraystretch}{2.0}
\caption{Below are key hyper-parameters for the training.}
\begin{NiceTabular}[width=14cm]{X[4,c,m]X[2,c,m]X[2,c,m]X[2,c,m]X[2,c,m]X[2,c,m]}[]
\hline
{\bf Loss function} & {\bf lr}  & {\bf initialization $\rho_i$} & \medskip {\bf $T_i$} &\smallskip {\bf batch size of $(\rho_i, T_i)$}\smallskip & {\bf Figure} \\ \hline \hline
Eq.~\eqref{eqn::loss_generative} & $10^{-2}$  & $\ket{0}$ & $1$ & n/a & Fig.~\ref{fig::eg1}  \\ \hline
Eq.~\eqref{eqn::loss_eg2} for Hydrogen & $10^{-2}$ &  $\ket{+}^{\otimes n}$ & $\text{Uniform}(1,2)$ & 1 & Fig.~\ref{fig::eg2} \\ 
Eq.~\eqref{eqn::loss_eg2} for Ising & $10^{-2}$ &  $\ket{+}^{\otimes n}$ & $\text{Uniform}(1,2)$ & 1 & Fig.~\ref{fig::eg2}  \\ 
Eq.~\eqref{eqn::loss_eg2} for time-dependent Ising  & $2\times 10^{-2}$ & $\ket{+}^{\otimes n}$ with random rotation & $\text{Uniform}(1,2)$ & 10 & Fig.~\ref{fig::eg2_td}  \\ \hline
Eq.~\eqref{eqn::loss_observable} for Hydrogen & $8\times 10^{-3}$ & $\ket{+}^{\otimes n}$ with random rotation & $\text{Uniform}(1,2)$ & 1 & Fig.~\ref{fig::eg3}\\ 
Eq.~\eqref{eqn::loss_observable} for Ising & $8\times 10^{-3}$ & $\ket{+}^{\otimes n}$ with random rotation & $\text{Uniform}(1,2)$ & 1 & Fig.~\ref{fig::eg3} \\ 
Eq.~\eqref{eqn::loss_observable} for time-dependent Ising & $2\times 10^{-2}$ & $\ket{+}^{\otimes n}$ with random rotation & $\text{Uniform}(1,2)$ & 10 & Fig.~\ref{fig::eg3_td} \\ \hline 
\end{NiceTabular}
\label{table::hyperparameter}
\end{table}

\end{document}